\documentclass[manuscript, nonacm]{acmart}

\AtBeginDocument{%
  \providecommand\BibTeX{{%
    \normalfont B\kern-0.5em{\scshape i\kern-0.25em b}\kern-0.8em\TeX}}}

\allowdisplaybreaks

\usepackage{bm} 
\usepackage{enumitem}
\usepackage{subcaption}
\usepackage{xspace}
\usepackage{algorithm}
\usepackage{algorithmicx}
\usepackage{algpseudocode}
\usepackage{tikz}
\usetikzlibrary{calc}
\usetikzlibrary{shapes,decorations.markings,arrows.meta,fit}
\usetikzlibrary{decorations.pathreplacing}
\newcommand{\mvivm}{\textsf{MVIVM}\xspace}
\newcommand{\deltac}{\textsf{Delta}\xspace}
\newcommand{\fivm}{\textsf{F-IVM}\xspace}
\newcommand{\ivmeps}{\textsf{IVM}$^\epsilon$\xspace}
\newcommand{\crown}{\textsf{CROWN}\xspace}
\newcommand{\naive}{\textsf{Na\"ive}\xspace}
\newcommand{\defeq}{\stackrel{\text{def}}{=}}
\newcommand{\bigO}[1]{\mathcal{O}(#1)}
\newcommand{\tildeO}[1]{\widetilde{\mathcal{O}}(#1)}
\newcommand{\fw}{\mathsf{w}}
\newcommand{\lb}{\mathsf{\omega}}

\newcommand{\vars}{\text{\sf vars}}
\newcommand{\atoms}{\text{\sf at}}
\newcommand{\Dom}{\text{\sf Dom}}
\newcommand{\agm}{\text{\sf AGM}}
\newcommand{\canpart}{\text{\sf CP}}
\newcommand{\parent}{\text{\sf parent}}
\newcommand{\children}{\text{\sf children}}
\newcommand{\ancestors}{\text{\sf ancestors}}
\newcommand{\decendants}{\text{\sf decendants}}
\newcommand{\segment}{\text{\sf seg}}

\newcommand{\iv}{\text{\sf IV}}
\newcommand{\polylog}{\text{\sf polylog}}
\newcommand{\td}{\text{\sf TD}}

\newcounter{magicrownumbers}

\newcommand{\applydelta}{\biguplus}
\newcommand{\true}{\text{\sf true}}

\newcommand{\mycount}{\text{\sf count}}
\newcommand{\subw}{\mathrm{subw}}

\newcommand{\univariate}{time\xspace}

\newcommand{\multivariate}{multivariate\xspace}
\newcommand{\Multivariate}{Multivariate\xspace}
\newcommand{\insertdelete}{insert-delete\xspace}
\newcommand{\InsertDelete}{Insert-Delete\xspace}

\newtheorem{thm}{Theorem}[section]
\newtheorem{claim}[thm]{Claim}

\theoremstyle{definition}
\newtheorem{remark}[thm]{Remark}

\usepackage[colorinlistoftodos]{todonotes}

\newtheoremstyle{cited}%
{.5\baselineskip\@plus.2\baselineskip
    \@minus.2\baselineskip}
{.5\baselineskip\@plus.2\baselineskip
    \@minus.2\baselineskip}
{\itshape}
{\parindent}
{}
{.}
{.5em}
{\textsc{\thmname{#1}} \thmnote{\normalfont#3}}

\theoremstyle{cited}
\newtheorem{citedthm}{Theorem}
\newtheorem{citedprop}{Proposition}

\newtheorem{citedlem}{Lemma}

\newcommand{\ivmpm}{\text{IVM}^{\pm}}
\newcommand{\ivmpmd}{\text{IVM}^{\pm}_\delta}
\newcommand{\ivmp}{\text{IVM}^+}
\newcommand{\ivmpd}{\text{IVM}^+_\delta}
\newcommand{\ivmov}{\overline{\text{IVM}}}
\newcommand{\ivmwh}{\widehat{\text{IVM}}}

\newcommand{\eval}{\text{Eval}}

\definecolor{light-gray}{gray}{0.7.2}
\definecolor{goodgreen}{rgb}{0.1, 0.5, 0.1}
\definecolor{burntorange}{rgb}{0.8, 0.33, 0.0}
\definecolor{lightblue}{RGB}{173, 216, 230} %

\newcommand{\calT}{\mathcal{T}}

\newcommand{\calA}{\mathcal{A}}

\newcommand{\calD}{\mathcal{D}}

\newcommand{\calI}{\mathcal{I}}

\newcommand{\calO}{\mathcal{O}}

\newcommand{\TAB}{\makebox[2.5ex][r]{}}%
\newcommand{\OUTPUT}{\textbf{output}\xspace}%
\newcommand{\mycomment}[1]{}

\newcommand{\nop}[1]{}

\newcommand{\change}[1]{#1}

\newcommand{\shortorfull}[2]{#2}

\newcommand{\ov}{\overline}
\newcommand{\wh}{\widehat}

\begin{document}

\title{Insert-Only versus Insert-Delete in Dynamic Query Evaluation}

\author{Mahmoud Abo Khamis}
\affiliation{%
    \institution{RelationalAI}
    \city{Berkeley}
    \state{CA}
    \country{United States}
}
\email{mahmoudabo@gmail.com}

\author{Ahmet Kara}
\affiliation{%
    \institution{OTH Regensburg}
    \city{Regensburg}
    \country{Germany}
}
\email{ahmet.kara@oth-regensburg.de}

\author{Dan Olteanu}
\affiliation{%
    \institution{University of Zurich}
    \city{Zurich}
    \country{Switzerland}
}
\email{olteanu@ifi.uzh.ch}

\author{Dan Suciu}
\affiliation{%
    \institution{University of Washington}
    \city{Seattle}
    \state{WA}
    \country{United States}
}
\email{suciu@cs.washington.edu}

\renewcommand{\shortauthors}{Abo Khamis, Kara, Olteanu, and Suciu}

\begin{abstract}
We study the dynamic query evaluation problem: Given a full conjunctive query $Q$ and a sequence of updates to the input database, we construct a data structure that supports constant-delay enumeration of the tuples in the query output after each update.

We show that a sequence of $N$ insert-only updates to an initially empty database can be executed in total time $\bigO{N^{\fw(Q)}}$, where $\fw(Q)$ is the fractional hypertree width of $Q$. This matches the complexity of the static query evaluation problem for $Q$ and a database of size $N$. One corollary is that the amortized time per single-tuple insert is constant for $\alpha$-acyclic full conjunctive queries.

In contrast, we show that a sequence of $N$ inserts and deletes can be executed in total time $\tildeO{N^{\fw(\wh Q)}}$, where $\wh Q$ is obtained from $Q$ by extending every relational atom with extra variables that represent the ``lifespans'' of tuples in the database. We show that this reduction is optimal in the sense that the static evaluation runtime of $\wh Q$ provides a lower bound on the total update time for the output of $Q$. \change{Our approach achieves amortized optimal update times  for the hierarchical and Loomis-Whitney join queries}.
\end{abstract}

\keywords{incremental view maintenance; optimality; intersection joins}

\maketitle

\section{Introduction}
\label{sec:intro}

Answering queries under updates to the database, also called dynamic query evaluation, is a fundamental problem in data management, with recent work addressing it from both systems and theoretical perspectives.  There have been recent efforts on building systems for dynamic query evaluation, such as \textsf{DBToaster}~\cite{DBLP:journals/vldb/KochAKNNLS14},
\textsf{DynYannakakis}~\cite{DBLP:conf/sigmod/IdrisUV17,DBLP:journals/vldb/IdrisUVVL20},
\textsf{F-IVM}~\cite{DBLP:conf/sigmod/NikolicO18,DBLP:conf/sigmod/NikolicZ0O20,FIVM:VLDBJ:2023}, and \textsf{CROWN}~\cite{DBLP:journals/pvldb/WangHDY23}.
By allowing tuples to carry payloads with elements from rings~\cite{DBLP:conf/pods/Koch10}, such systems can maintain complex analytics over database queries, such as linear algebra computation~\cite{DBLP:conf/sigmod/NikolicEK14}, collection programming~\cite{DBLP:conf/pods/0001LT16}, and machine learning models~\cite{FIVM:VLDBJ:2023}.
There has also been work on dynamic computation for intersection joins~\cite{DBLP:journals/jcss/TaoY22} and for expressive languages such as Datalog~\cite{DBLP:journals/ai/MotikNPH19}, Differential Datalog~\cite{DBLP:conf/fossacs/AbadiMP15}, and DBSP~\cite{DBLP:journals/pvldb/BudiuCMRT23}. The
descriptive complexity of various recursive Datalog queries under updates, such as reachability, has also been investigated~\cite{DBLP:journals/jacm/DattaKMSZ18,DBLP:journals/sigmod/SchwentickVZ20}.
\change{Dynamic query evaluation has also been utilized for complex event recognition~\cite{10.14778/3137765.3137829, grez_et_al:LIPIcs.ICDT.2020.14}.}

While this problem attracted continuous interest in academia and industry over the past decades, it is only relatively recently that the first results on the fine-grained complexity and in particular on the optimality of this problem have emerged. These efforts aim to mirror the breakthrough made by the introduction of worst-case optimal join algorithms~\cite{AtseriasGM13,LeapFrogTrieJoin2014,Ngo:JACM:18}.  Beyond notable yet limited explorations, understanding the optimality of maintenance for the entire language of conjunctive queries remains open. Prime examples of progress towards the optimality of query maintenance are the characterizations of queries that admit (worst-case or amortized) constant time per single-tuple update (insert or delete): the $q$-hierarchical queries~\cite{BerkholzKS17,DBLP:conf/sigmod/IdrisUV17,FIVM:VLDBJ:2023,DBLP:journals/pvldb/WangHDY23}, 
queries that become $q$-hierarchical in the presence of free access patterns~\cite{DBLP:conf/icdt/00020OZ23} or by rewriting under functional dependencies~\cite{FIVM:VLDBJ:2023}, or queries on update sequences whose enclosureness is bounded by a constant~\cite{DBLP:conf/sigmod/WangY20}. The $\delta_1$-hierarchical queries~\cite{DBLP:conf/pods/0002NOZ20,DBLP:journals/lmcs/KaraNOZ23,DBLP:conf/icdt/00020OZ23} 
and the triangle queries~\cite{DBLP:conf/icdt/KaraNNOZ19,DBLP:journals/tods/KaraNNOZ20} admit optimal, albeit non-constant, update time conditioned on the Online Matrix-Vector Multiplication (OMv) conjecture~\cite{HenzingerKNS15}.

\change{In this paper, we introduce a new approach to incremental view maintenance for full conjunctive queries (or queries for short). Our approach complements prior work in four ways.}

\change{First, we give reductions from the dynamic query evaluation problem for a query $Q$ to the static query evaluation problem of a derived query $\wh Q$ (called the multivariate extension of $Q$), and vice versa. Our reductions link the complexities of the two problems and allow to transfer {\em both} algorithms {\em and} lower bounds from one problem to the other.
Specifically, we give a reduction from dynamic query evaluation
to a wide class of algorithms for static query evaluation, namely algorithms that meet the {\em fractional hypertree width}, thus allowing us to use those algorithms for dynamic query evaluation. Moreover, we give a reduction from
static query evaluation to {\em any} algorithm for dynamic query evaluation, thus allowing us to use lower bounds on static query evaluation to infer lower bounds on dynamic query evaluation.}
Both reductions use a translation of the time dimension in the dynamic problem into a spatial dimension in the static problem, so that the maintenance under a stream of updates corresponds to taking the intersection of intervals representing the ``lifespans'' of tuples in the update stream.

Second, we devise new dynamic query evaluation algorithms that use our reductions. 
We call them \mvivm (short for {{\textsf{M}}ulti{\textsf{V}}ariate \textsf{IVM}}).
\change{Their single-tuple update times are {\em amortized}. For  any $\alpha$-acyclic query in the insert-only setting, \mvivm take 
amortized constant single-tuple update time. This recovers a prior result~\cite{DBLP:journals/pvldb/WangHDY23}, when restricted from free-connex to full conjunctive queries.
For Loomis-Whitney queries in the \insertdelete setting, \mvivm take  amortized $\tildeO{|\calD|^{1/2}}$ update  time, where $|\calD|$ is the database size at update time; this is optimal up to a polylogarithmic factor in the database size. This recovers prior work on the triangle query, which is the Loomis-Whitney query with three variables~\cite{DBLP:journals/tods/KaraNNOZ20}.
For hierarchical queries in the \insertdelete setting, \mvivm take amortized constant single-tuple update time; this is weaker than prior work~\cite{BerkholzKS17}, which showed that worst-case constant single-tuple update time can be achieved for $q$-hierarchical queries.}

\change{Third, \mvivm support (worst-case, not amortized) constant delay enumeration} of the tuples in the query output (full enumeration) or in the change to the query output (delta enumeration) {\em after each update}.
In particular,  \mvivm map an input stream of updates to an output stream of updates, so they are closed under updates. In practice, the change to the output after an update is often much smaller than the full output, making delta enumeration more efficient than full enumeration.

Fourth, we pinpoint the complexity gap between the insert-only and the \insertdelete settings. 
\change{In the insert-only setting, we show that one can readily use a worst-case optimal join algorithm like GenericJoin~\cite{Ngo:JACM:18} and  LeapFrogTrieJoin~\cite{LeapFrogTrieJoin2014} to execute a stream of inserts one insert at a time in the same overall time as if the entire stream was executed as one bulk update (which corresponds to static evaluation).}
\change{This implies an upper bound on the amortized update time for the dynamic evaluation of $Q$ in the insert-only setting.}
In contrast, in the \insertdelete setting, a stream of both inserts and deletes  can be executed one update at a time in the same overall time as the static evaluation of a query $\wh Q$, which is obtained by extending every relational atom in $Q$ with extra variables representing the ``lifespans'' of tuples.
This query $\wh Q$ may have a higher complexity than $Q$,
\change{thus leading to an upper bound on the amortized update time for $Q$ in the insert-delete setting that may be higher than in the insert-only setting. Nevertheless, our reduction
from the maintenance of $Q$ in the insert-delete setting to the static evaluation problem for $\wh Q$ is optimal:}
Lower bounds on the latter problem imply lower bounds on the former.
This shows that maintenance in the \insertdelete setting can be more expensive than in the insert-only setting.

\subsection*{Motivating Examples}

We exemplify our results for the following three join queries:
\begin{align}
    Q_{3p}(A,B,C,D) &= R(A,B)\wedge S(B,C)\wedge T(C,D) \nonumber\\
    Q_\triangle(A,B,C) &= R(A,B)\wedge  S(B,C)\wedge T(A,C) \label{eq:triangle}\\
    Q_{4}(A,B,C,D) &= R(B,C,D)\wedge S(A,C,D)\wedge T(A,B,D)\wedge U(A,B,C)\nonumber
\end{align}

We would like to perform single-tuple updates to each query, while ensuring constant-delay enumeration of the tuples in the query output after each update.
These queries are not $q$-hierar\-chi\-cal, hence they cannot admit worst-case constant time per single-tuple update~\cite{BerkholzKS17}:  The OMv conjecture implies that there is no algorithm that takes $\bigO{|\calD|^{1/2-\gamma}}$ time for single-tuple updates, for any $\gamma>0$, to any of these queries. In fact, for a database $\calD$, systems like \textsf{DBToaster}~\cite{DBLP:journals/vldb/KochAKNNLS14}, \textsf{DynYannakakis}~\cite{DBLP:conf/sigmod/IdrisUV17}, \textsf{F-IVM}~\cite{DBLP:conf/sigmod/NikolicO18}, and \textsf{CROWN}~\cite{DBLP:journals/pvldb/WangHDY23} need at least
\change{worst-case} $\bigO{|\calD|}$ time per single-tuple update to each of these queries.
For the triangle query $Q_\triangle$, IVM$^\epsilon$ takes $\bigO{|\calD|^{1/2}}$ amortized time for single-tuple updates~\cite{DBLP:journals/tods/KaraNNOZ20}.
In the insert-only setting, \textsf{CROWN} maintains $Q_{3p}$ in amortized constant time per insert. However, it cannot handle non-acyclic queries like $Q_\triangle$ and $Q_4$
(with update time better than recomputation from scratch).
Appendix~\ref{app:comparison} explains how each prior approach maintains $Q_{3p}$ and $Q_\triangle$ in the insert-delete setting, and compares their update times to our approach, \mvivm.

We show that, for $Q_{3p}$, \mvivm performs single-tuple updates: in amortized constant time in the insert-only setting; and in amortized $\tildeO{|\calD|^{1/2}}$ time in the \insertdelete setting ($\tilde{O}$ notation hides a $\polylog(|\calD|)$ factor).
Both times are optimal, the latter is conditioned on the OMv conjecture. The amortized constant time per insert matches that of \textsf{CROWN}~\cite{DBLP:journals/pvldb/WangHDY23} and does not contradict the non-constant lower bound for non-hierarchical queries without self-joins~\cite{BerkholzKS17}: that lower bound proof relies on a reduction from OMv that requires {\em both} inserts {\em and} deletes.
Achieving this amortized constant time, while also allowing constant-delay enumeration after each insert, is non-trivial.
The lazy maintenance approach, which only updates the input data, is {\em not} sufficient to achieve this result. This is because, while it does indeed take constant insert time, it still requires the computation of the query output  {\em before} the enumeration of the first output tuple. This  computation takes linear time using a factorized approach~\cite{FIVM:VLDBJ:2023} and needs to be repeated for {\em every} enumeration request.

\mvivm takes $\bigO{|\calD|^{1/2}}$ and $\bigO{|\calD|^{1/3}}$ amortized single-tuple insert time for the triangle query $Q_\triangle$ and, respectively, the Loomis-Whitney-4 query $Q_{4}$.  
Given a sequence of $|\calD|$ inserts, the total times of $\bigO{|\calD|^{3/2}}$ and $\bigO{|\calD|^{4/3}}$, respectively, match the worst-case optimal times for the static evaluation of these two queries~\cite{Ngo:JACM:18}. 
In the insert-delete setting, single-tuple updates take $\tildeO{|\calD|^{1/2}}$ amortized time for both queries, and this is optimal based on the OMv conjecture.
Figure~\ref{fig:comparison_prior_our} in Appendix~\ref{app:comparison} shows the update times of \mvivm and prior work for further queries.

\section{Preliminaries}
\label{sec:prelims}
\change{In this section, we review some preliminaries, most of which are standard in database theory.}

We use $\mathbb{N}$ to denote the set of natural numbers including $0$.
For $n\in\mathbb{N}$, we define $[n] \defeq \{1,2,\ldots,n\}$. In case $n=0$, we have $[n]=\emptyset$.
\paragraph{Data Model}
\nop{

}
\change{Following standard terminology, 
a {\em relation} is a finite set of tuples and a 
{\em database} is a finite 
set of relations~\cite{AbiteboulHV95}.
The size $|R|$ of a relation $R$  is the number of its tuples. 
The size $|\calD|$ of a database $\calD$ is given by the sum of the sizes of its relations.  
}

\paragraph{Queries}
We consider natural join queries or full conjunctive queries.
We refer to them as {\em queries} for short. A {\em query} has the form 
\begin{align}
    Q(\bm X) = R_1(\bm X_1) \wedge \cdots \wedge R_k(\bm X_k),
    \label{eq:join-query}
\end{align}
where 
each $R_i$ is a relation symbol, each 
$\bm X_i$ is a tuple of variables, and 
$\bm X = \bigcup_{i \in [k]} \bm X_i$.
\change{We refer to $\bm X_i$ as the {\em schema}
of $R_i$ and treat schemas and sets of variables
interchangeably, assuming a fixed ordering of variables.}
We call each $R_i(\bm X_i)$ an {\em atom} of $Q$.
We denote by 
$\atoms(Q) = \{R_i(\bm X_i) \mid i \in [k]\}$ 
the set of all atoms of $Q$ and by $\atoms(X)$
the set of atoms $R_i(\bm X_i)$ with $X \in \bm X_i$.
The set of all variables of $Q$ is denoted by $\vars(Q) = \bigcup_{i \in [k]} \bm X_i$. Since the set 
$\bm X$ is always equal to $\vars(Q)$ for all queries we consider in the paper, \underline{we drop $\bm X$ from the query notation} and abbreviate $Q(\bm X)$ with $Q$ in Eq.~\eqref{eq:join-query}.
We say that $Q$ is without {\em self-joins} if every relation symbol appears in at most one atom. 

\change{
The domain of a variable $X$ is denoted by  $\Dom(X)$. 
A {value tuple} $\bm t$ over a schema $\bm X = (X_1, \ldots , X_n)$ is an element from $\Dom(\bm X) \defeq \Dom(X_1) \times \cdots \times \Dom(X_n)$.
We denote by $\bm t(X)$ the $X$-value of $\bm t$. 
}

A tuple $\bm t$ over schema $\bm X$
is in the result of $Q$ if there are tuples $\bm t_i \in R_i$ for $i \in [k]$ such that for all $X \in \vars(Q)$, the set $\{\bm t_i(X)\mid i \in [k] \wedge X \in \bm X_i\}$ consists of the single value $\bm t(X)$.
Given a  query $Q$ and  a database $\calD$, we denote by $Q(\calD)$ the result of $Q$ over $\calD$.

A {\em union of queries} is of the form 
$Q = \bigvee_{i \in [n]} Q_i$ where each $Q_i$, called a component of $Q$, is a query over the same set of variables as $Q$. 
The result of $Q$ is the union of the results of its components. 

\paragraph{Query Classes}
\change{
A query is {\em ($\alpha$-)acyclic} if we can construct a tree, called {\em join tree}, such that the nodes of the tree are the atoms of the query ({\em coverage}) and for each variable, it holds: if the variable appears in two atoms, then it appears in all atoms on the path connecting the two atoms ({\em connectivity})~\cite{Brault-Baron16}.
}

{\em Hierarchical queries} are a sub-class of acyclic queries.
A query is called hierarchical if for any two variables $X$ and $Y$, it holds 
$\atoms(X) \subseteq \atoms(Y)$, 
$\atoms(Y) \subseteq \atoms(X)$, or
$\atoms(X) \cap \atoms(Y) = \emptyset$~\cite{2011Suciu}.

{\em Loomis-Whitney queries} generalize
the triangle query from a clique of degree $3$ to higher degrees~\cite{LW:1949}. 
A Loomis-Whitney query of degree $k\geq 3$ has $k$ variables
$X_1, \ldots, X_k$ and has, for each subset $\bm Y \subset 
\{X_1, \ldots, X_k\}$ of size $k-1$, one distinct atom with schema $\bm Y$. 


A union of queries is acyclic (hierarchical, Loomis-Whitney)
if each of its components is acyclic (hierarchical, Loomis-Whitney).
\shortorfull{}{Example~\ref{ex:queries} showcases the query classes mentioned above.}

\paragraph{Single-tuple Updates and Deltas}
We denote an  {\em insert} of a tuple $\bm t$ into a relation $R$ by $+R(\bm t)$ and a {\em delete} of $\bm t$ from $R$ by 
$-R(\bm t)$. 
The insert $+R(\bm t)$ inserts $\bm t$ into $R$ if $\bm t$ {\em is not} contained in $R$, and 
the delete $-R(\bm t)$ deletes $\bm t$ from $R$ if it {\em is} 
contained in $R$.
Each insert or delete has a {\em timestamp} $\tau$, which is a natural number that indicates the specific time at which
the insert or delete occurs.
An {\em update} $\delta_\tau R$ to $R$ at time $\tau$ is the set of all inserts and deletes to $R$ at time $\tau$.
An {\em update} $\delta_\tau R$ is applied to $R$ by first applying all the deletes in 
$\delta_\tau R$ and then applying all the inserts.
We use $R \applydelta \delta_\tau R$ to denote the result of applying an update $\delta_\tau R$ to a relation $R$.
We use $R^{(\tau)}$ to refer to the specific version of a relation $R$ at timestamp $\tau$.
In particular, $R^{(\tau)}$ is the result of applying $\delta_\tau R$ to the previous version $R^{(\tau-1)}$, i.e.~ $R^{(\tau)} = R^{(\tau-1)}\applydelta \delta_\tau R$.
An {\em update} $\delta_\tau\calD$ to a database $\calD$
is the set of all inserts and deletes to the relations in $\calD$ at time $\tau$.
Similarly, we use $\calD^{(\tau)}$ to refer to the specific version of $\calD$ at time $\tau$, i.e.,
$\calD^{(\tau)} = \calD^{(\tau-1)} \applydelta \delta_\tau \calD$.
If $\delta_\tau\calD$ consists of a single insert or delete, i.e.,  $|\delta_\tau\calD| = 1$, it is called a 
{\em single-tuple update}. 
In particular, if it consists of  a single insert (or delete), 
we call it a {\em single-tuple insert} (or {\em single-tuple delete}). 
Given a query $Q$ over a database $\calD$, we use 
$\delta_\tau Q(\calD)$ to denote
the change in $Q(\calD)$ at time $\tau$. Specifically, $\delta_\tau Q(\calD)$ is the difference between $Q(\calD^{(\tau)})$ and $Q(\calD^{(\tau-1)})$:

\begin{align}
    \delta_\tau Q(\calD) \defeq
    &\{+Q(\bm t) \mid \bm t \in Q(\calD^{(\tau)}) \text{ and } \bm t \notin Q(\calD^{(\tau-1)})\}\nonumber\\ \cup
    &\{-Q(\bm t) \mid \bm t \in Q(\calD^{(\tau-1)}) \text{ and } \bm t \notin Q(\calD^{(\tau)})\}
    \label{eq:def:deltaQ}
\end{align}

When the time $\tau$ is clear from the context, we drop the 
annotation $(\tau)$ and write $R, \delta R, \calD, \delta \calD, Q(\calD)$, and $\delta Q(\calD)$.
In the insert-only setting, since all deltas are positive, we ignore the signs
and see a single-tuple update as just a tuple, $\delta R$ as just a relation,
and so on.

\paragraph{Constant-Delay Enumeration}
Given a query $Q$ and a database $\calD$, an enumeration procedure for 
$Q(\calD)$ (or $\delta_{\tau}Q(\calD)$) outputs the elements 
 in $Q(\calD)$ (or $\delta_{\tau}Q(\calD)$) one by one 
 in some order and without repetition. 
The {\em enumeration delay} of the procedure is the maximum of three times:
the time between the start of the enumeration process and the output of the first element, the time between outputting any two consecutive elements, and the time between outputting the last element and the end of enumeration~\cite{DurandFO07}. 
The enumeration delay is {\em constant} if it does not depend
on $|\calD|$. 

\change{
\paragraph{Width Measures}
Consider the following linear program for
a query $Q$ and a database $\calD$:
\nop{Given a query $Q$ and a database $\calD$,
the {\em fractional edge cover number} $\rho^{\ast}(Q,\calD)$ of $Q$ with respect to $\calD$ is the optimal 
objective value of the following linear program:} 
\begin{align*}
\text{min} & \TAB\sum_{R(\bm X) \in\, \atoms(Q)} \log(|R|) \cdot \lambda_{R(\bm X)} && \\[3pt]
\text{s.t.} &\TAB 
\sum_{{R(\bm X)\in\, \atoms(Q) \text{ s.t. } Y \in \bm X}}\hspace{-0.5cm} \lambda_{R(\bm X)} \geq 1 &&\hspace{-0.5em} \text{ for all } Y \in \vars(Q) \text{ and } \\[3pt]
& \TAB\lambda_{R(\bm X)} \in [0,1] &&\hspace{-0.5em} \text{ for all } R(\bm X) \in \atoms(Q)
\end{align*}


Every feasible solution $(\lambda_{R(\bm X)})_{R(\bm X)\in \atoms(Q)}$
to the above program is called a 
{\em fractional edge cover} of $Q$ with respect to $\calD$.
The optimal objective value  of the program is called the  
{\em fractional edge cover number} $\rho^*(Q, \calD)$ of $Q$ with respect to $\calD$. 
The fractional edge cover number gives an upper bound 
on the size of $Q(\calD)$: 
$|Q(\calD)| \leq 2^{\rho^*(Q, \calD)}$~\cite{AtseriasGM13}.
This bound is known as the {\em AGM-bound} and denoted as 
$\agm(Q,D)$.
We denote by $\rho^{\ast}(Q)$
the optimal objective value of the linear program that results from the above program by replacing the objective function with
$\sum_{R(\bm X) \in\, \atoms(Q)} \lambda_{R(\bm X)}$.

\begin{definition}[Tree Decomposition]
    \label{defn:TD}    
    A {\em tree decomposition} of a query $Q$ is a pair 
    $(\calT, \chi)$, where $\calT$ is a tree with vertices $V(\calT)$  
    and $\chi: V(\calT) \rightarrow 2^{\vars(Q)}$ maps each node $t$ of the tree $\calT$ to a subset $\chi(t)$ of variables of $Q$ such that the following properties hold:
    \begin{enumerate}
        \item for every atom $R(\bm X) \in \atoms(Q)$, the schema 
        $\bm X$ is a subset of $\chi(t)$ for some $t \in V(\calT)$, and
        \item for every variable $X \in \vars(Q)$, the set 
        $\{t \mid X \in \chi(t)\}$ is a non-empty connected subtree of $\calT$.
    The sets $\chi(t)$ are called the {\em bags} of the tree decomposition.
    \end{enumerate}
    We use $\td(Q)$ to denote the set of tree decompositions of 
    a query $Q$.
\end{definition}

\begin{definition}[The restriction $Q_{\bm Y}$ of a query $Q$]
    \label{defn:restriction}
    For a query $Q$ and a subset $\bm Y \subseteq \vars(Q)$,
    we define the {\em restriction of $Q$ to $\bm Y$}, denoted by $Q_{\bm Y}$,
to be the query that results
from $Q$ by restricting the schema of each atom to those variables 
that appear in $\bm Y$. Formally:
\begin{align}
    R_{\bm Y}(\bm X \cap \bm Y) &\defeq \pi_{\bm X \cap \bm Y} R(\bm X),\quad\quad
    \text{for all $R(\bm X) \in \atoms(Q)$}
    \nonumber\\
    Q_{\bm Y}(\bm Y) &\defeq \bigwedge_{R(\bm X) \in \atoms(Q)} R_{\bm Y}(\bm X \cap \bm Y)
    \label{eq:bag_query}
\end{align}
\end{definition}

\begin{definition}[Fractional Hypertree Width]
    \label{defn:fhtw}
Given a query $Q$ and a tree decomposition $(\calT,\chi)$ of $Q$, the {\em fractional hypertree width} of $(\calT,\chi)$ and of $Q$ are defined respectively as follows:
\begin{align}
\fw(\calT,\chi) \defeq \max_{t \in V(\calT)} \rho^*(Q_{\chi(t)}),\quad\quad
\fw(Q) \defeq \min_{(\calT,\chi) \in \td(Q)} \fw(\calT,\chi) \label{eq:fhtw}
\end{align}
We call a tree decomposition $(\calT, \chi)$ of $Q$ {\em optimal} if
$\fw(\calT, \chi) = \fw(Q)$.
 The {\em fractional hypertree width} of a union  $Q = \bigcup_{i \in [n]} Q_i$ of join queries 
 is
$\fw(Q) \defeq \max_{i \in [n]} \fw(Q_i)$.
\end{definition}

If a query is a union of join queries,
its AGM-bound (or fractional hypertree width) is the maximum 
AGM-bound (or fractional hypertree width) over its components. 
}

\paragraph{Computational Model}
\change{
We consider the RAM model of computation where 
schemas and 
data values are stored in registers 
of logarithmic size and operations on them can be done in constant time.
We assume that each materialized relation 
with schema $\bm X$ is implemented by a data structure 
of size $\bigO{|R|}$ that can:
(1) look up, insert, and delete tuples in constant time, and
(2) enumerate the tuples in $R$ with constant delay.
Given $\bm Y \subseteq X$ and $\bm t \in \Dom(\bm Y)$, the data structure can:
(1) check $\bm t \in \pi_{\bm Y} R$ in constant time; and 
(2) enumerate the tuples in $\sigma_{\bm Y =\bm t}R$ with constant delay.
}

\section{Problem Setting}
\label{sec:problem}
We consider natural join queries, or equivalently, full conjunctive queries and refer to them as {\em queries}.
We investigate the incremental view maintenance (IVM) of a query $Q$ under
{\em single-tuple updates} to an initially empty database $\calD$.
A {single-tuple update} is an insert or a delete of a tuple into a relation in $\calD$.
We consider four variants of the IVM problem depending on the 
update and enumeration mode.
With regard to updates, we distinguish between the {\em insert-only} setting, where we only allow single-tuple inserts, and the {\em \insertdelete} setting, where we allow both single-tuple inserts and single-tuple deletes.
With regard to enumeration, we distinguish between the {\em full} enumeration and the {\em delta} enumeration, where, after each update, the full query result and respectively the change to the query result at the current time can be enumerated with constant delay.
The four IVM variants are parameterized by a query $Q$ and take as input a database $\calD$, 
whose relations are initially empty, and a stream 
$\delta_1\calD, \delta_2\calD, \ldots, \delta_N\calD$ of $N$ single-tuple updates to $\calD$.
Neither $N$ nor the updates are known in advance.
The updates arrive one by one.
Let $\calD^{(0)} \defeq \emptyset, \calD^{(1)}, \ldots, \calD^{(N)}$
be the sequence of database versions, where
for each $\tau \in [N]$, the database version $\calD^{(\tau)}$ results from applying the update
$\delta_\tau\calD$ to the previous version $\calD^{(\tau-1)}$.

The task of the first IVM variant is to support constant-delay enumeration of the full query result $Q(\calD^{(\tau)})$ after each update:

\noindent
\vspace*{.5em}
\fbox{%
    \parbox{0.96\linewidth}{%
    \begin{tabular}{ll}
   Problem: & $\ivmpm[Q]$ \\
   Parameter: & Query $Q$ \\ 
        Given: & An initially empty database $\calD$ and a stream
        $\delta_1\calD, \ldots,\delta_N\calD$ of $N$ single-tuple updates to $\calD$ \\
        &where $N$ is {\em not} known in advance.\\
        Task: & Support constant-delay enumeration of
        $Q(\calD^{(\tau)})$ after each update $\delta_\tau\calD$ \\
    \end{tabular}
    }}
\vspace*{.5em}

The task of the second IVM variant is to support
constant-delay enumeration of the {\em change} to the query result
after each update, \change{given by Eq.~\eqref{eq:def:deltaQ}:}

\noindent
\fbox{%
    \parbox{0.96\linewidth}{%
    \begin{tabular}{ll}
   Problem: & $\ivmpmd[Q]$ \\
        Parameter: & Query $Q$ \\ 
        Given: & An initially empty database $\calD$ and a stream
        $\delta_1\calD, \ldots,\delta_N\calD$ of $N$ single-tuple updates to $\calD$\\
        &where $N$ is {\em not} known in advance.\\
        Task: & Support constant-delay enumeration of $\delta_\tau Q(\calD)$ after each update $\delta_\tau \calD$ \\
    \end{tabular}
    }}
\vspace*{.5em}


The other two variants, denoted as $\ivmp[Q]$ and $\ivmpd[Q]$, are identical to $\ivmpm[Q]$ and $\ivmpmd[Q]$, respectively, except that the updates $\delta_\tau \calD$ are restricted to be single-tuple {\em inserts} (no deletes are allowed).
In all four variants, we are interested in the {\em update time}, i.e., the time to process a single-tuple update. 

\begin{example}
\label{ex:plain_update_enumeretaion}
Consider the triangle query $Q_\triangle$ from Eq.~\eqref{eq:triangle}.
Suppose we have the update stream of length $8$ given in the second column of  Table~\ref{tab:triangle:full_delta_enumeration}.
The third and fourth column show the full result and the delta result, respectively,  
after each update. \change{We will use this example as a running example throughout the paper.}

\begin{table}[ht!]
    \begin{tabular}{|c|l|l|l|}
        \hline
        $\tau$ & $\delta_\tau\calD$ & $Q(\calD^{(\tau)})$ & $\delta_\tau Q(\calD)$ \\\hline\hline 
        1& $+R(a_1, b_1)$ & $\{\}$ & $\{\}$ \\\hline
        2& $+S(b_1, c_1)$ & $\{\}$ & $\{\}$ \\\hline
        3& $+T(a_1, c_1)$ & $\{(a_1,b_1,c_1)\}$ & $\{+Q(a_1, b_1, c_1)\}$ \\\hline
        4& $+S(b_2, c_1)$ & $\{(a_1,b_1,c_1)\}$ & $\{\}$ \\\hline
        5& $-S(b_1, c_1)$ & $\{\}$ & $\{-Q(a_1, b_1, c_1)\}$\\\hline
        6& $-S(b_2, c_1)$ & $\{\}$ & $\{\}$\\\hline
        7& $-T(a_1, c_1)$ & $\{\}$ & $\{\}$ \\\hline
        8& $-R(a_1, b_1)$ & $\{\}$ & $\{\}$ \\\hline
    \end{tabular}
    \caption{An update sequence for the query $Q_\triangle$ in Eq.~\eqref{eq:triangle}.
The last two columns show the full and delta result 
after each update.}
    \label{tab:triangle:full_delta_enumeration}
    \vskip -0.6cm
\end{table}
\end{example}

We obtain lower bounds on the update time for our IVM problems by reductions from the following static query evaluation problem:
 
\noindent
\vspace*{.5em}
\fbox{%
    \parbox{0.96\linewidth}{%
    \begin{tabular}{ll}
   Problem: & $\eval[Q]$ \\
   Parameter: & Query $Q$ \\
        Given: & Database  $\calD$\\
        Task: & Compute $Q(\calD)$
    \end{tabular}
    }}
\vspace*{.5em}

For this problem, we are interested in the time to compute $Q(\calD)$.
\begin{definition}[$\fw(Q)$ and $\lb(Q)$]
    \label{defn:fw_lb}
    Given a query $Q$, let $\fw(Q)$ denote the {\em fractional hypertree width} of $Q$ (see Section~\ref{sec:prelims}).
    Let $\lb(Q)$ denote the {\em smallest exponent} $\kappa$ such that
    $\eval[Q]$ has an algorithm with runtime $\bigO{|\calD|^{\kappa + o(1)} + |Q(\calD)|}$
    on any input database $\calD$.
    Given a union of queries $Q = \bigvee_i Q_i$, we define $\fw(Q)$ and $\lb(Q)$ to be the maximum of $\fw(Q_i)$ and $\lb(Q_i)$, respectively, over all $i$.
    \footnote{The function $\omega(Q)$ is defined analogously to the constant $\omega$
    which is the best exponent for matrix multiplication.}
\end{definition}
The function $\lb(Q)$ is not known in general. However, we know that $\lb(Q) \leq \fw(Q)$;
see e.g.~\cite{faq}.
More strongly, $\lb(Q)$ is upper bounded by the {\em submodular width} of $Q$~\cite{marx:subw,panda:pods17,2024arXiv240202001A}\footnote{
    \cite{panda:pods17,2024arXiv240202001A} provide a query evaluation algorithm with runtime
    $\bigO{|\calD|^{\subw(Q)}\cdot \polylog{|\calD|} + |Q(\calD)|}=
    \bigO{|\calD|^{\subw(Q)+o(1)} + |Q(\calD)|}$, thanks to the extra $+o(1)$
    in the exponent. $\subw(Q)$ denotes the submodular width of $Q$.
}.
For acyclic queries, we have $\lb(Q) = 1$~\cite{Yannakakis81}.
For the triangle query, $\lb(Q_\triangle) \geq \frac{4}{3}$ modulo
the 3SUM conjecture~\cite{10.1145/1806689.1806772}.
In this paper, we prove lower bounds on the IVM problems in terms of 
the function $\lb(Q)$.
These lower bounds are not conditioned on a fine-grained complexity conjecture, like the OMv conjecture~\cite{HenzingerKNS15}\shortorfull{}{ stated in Appendix~\ref{app:prelims}}.
Instead, they apply no matter what the value
of $\lb(Q)$ turns out to be. In that sense, they can be thought of as unconditional.
\change{However, they are relatively straightforward to prove
and are mostly meant to justify the matching upper bounds.}

\paragraph{Complexity Measures}
For all problems introduced above, we 
consider the query to be fixed.
Moreover, we say that the {\em amortized} update time for the $\tau$-th single-tuple update is $f(N, |\calD^{(\tau)}|)$ for some function $f$, if the total time to process all $N$ updates in the stream is upper bounded by $\sum_{\tau \in [N]}f(N, |\calD^{(\tau)}|)$.
In the insert-only setting,
the database size $|\calD^{(\tau)}|$ at time $\tau$ is always equal to $\tau$.
Hence, for $\ivmp[Q]$ and $\ivmpd[Q]$,
we ignore the database size and only measure the update time as a function of the length $N$ of the update stream.  
In the insert-delete setting, the database size at time $\tau$ is upper bounded by $\tau$
and could be much smaller than $\tau$.
As a result, for $\ivmpm[Q]$ and $\ivmpmd[Q]$, it is more natural to ignore $N$ and only measure the update time as a function of $|\calD^{(\tau)}|$.
For the four IVM problems, the constant enumeration delay is {\em not} allowed to be amortized,
regardless of whether the update time is amortized or not.
For the problem $\eval$, we measure the computation time as a function of the size of the input database
$\calD$.
We use the $\calO$-notation
to state worst-case bounds and $\widetilde{\calO}$-notation to hide a polylogarithmic factor
in $N$, $|\calD^{(\tau)}|$, or $|\calD|$.

\paragraph{\Multivariate Extensions}
For the purpose of stating our results, we introduce the following central concept,
which gives our approach its name, {{\textsf{M}}ulti{\textsf{V}}ariate \textsf{IVM}} (\mvivm for short).
\begin{definition}[\Multivariate extension $\wh Q$ of a query $Q$]
  \label{defn:multivariate}
    Consider a query $Q(\bm X) = R_1(\bm X_1)\wedge \ldots\wedge R_k(\bm X_k)$
    where $\bm X = \bigcup_{i \in [k]} \bm X_i$,
    and fresh variables $Z_1, \ldots , Z_k$ that do {\em not} occur in $Q$.
    (In this paper, \underline{we drop head variables $\bm X$} and write $Q$ instead of $Q(\bm X)$ for brevity.)
    Let $\Sigma_k$ be the set of permutations of the set $[k]$.
    For any permutation $\bm\sigma = (\sigma_1, \ldots , \sigma_k) \in \Sigma_k$, we denote by $\wh Q_{\bm\sigma}$ the query that results from $Q$ by extending the schema of each atom $R_{\sigma_i}(\bm X_{\sigma_i})$ with the variables 
    $Z_1, \ldots , Z_i$:
    \begin{align}
        \wh Q_{\bm\sigma} = \wh R_{\sigma_1}(Z_1, \bm X_{\sigma_1}) 
        \wedge \wh R_{\sigma_2}(Z_1, Z_2, \bm X_{\sigma_2})
        \wedge \cdots\wedge \wh R_{\sigma_k}(Z_1, \ldots, Z_k,\bm X_{\sigma_k}) \label{eq:multivariate_component}
    \end{align}
    The {\em \multivariate extension} $\wh Q$ of $Q$ is the union of $\wh Q_{\bm\sigma}$
    over all permutations $\bm\sigma \in \Sigma_k$:
    \begin{align}
    \label{eq:multivariate}
    \wh{Q} = \bigvee_{\bm\sigma \in \Sigma_k}\wh Q_{\bm\sigma}
    \end{align}
   We call each $\wh Q_{\bm\sigma}$ a {\em component} of $\wh{Q}$. 
\end{definition}

\begin{example}
    \label{ex:transformations}
    The \multivariate extension $\wh Q_{\triangle}$ of $Q_\triangle$ from Eq.~\eqref{eq:triangle}
    consists of 6 components $\wh Q_{123}, \wh Q_{132}, \wh Q_{213}, \wh Q_{231}, \wh Q_{312}$
    and $\wh Q_{321}$. The first component, $\wh Q_{123}$, orders the input relations as $(R, S, T)$
    and applies Eq.~\eqref{eq:multivariate_component} by adding $\{Z_1\}$ to $R$, $\{Z_1, Z_2\}$ to $S$, and $\{Z_1, Z_2, Z_3\}$ to $T$, thus resulting in
    $\wh R(Z_1, A, B)$, $\wh S(Z_1, Z_2, B, C)$, and $\wh T(Z_1, Z_2, Z_3, A, C)$ respectively. Below, we add the subscript $123$ to each relation to indicate the component for later convenience:
    \begin{align}
        \wh Q_{123} =
            \wh R_{123}(Z_1, A, B) \wedge 
            \wh S_{123}(Z_1, Z_2, B, C) \wedge
            \wh T_{123}(Z_1, Z_2, Z_3, A, C)
        \label{eq:triangle-query:multi-123:body}
    \end{align}
    \shortorfull{}{The other components are given in Eq.~\eqref{eq:triangle-query:multi-132-321}.}
\end{example}

\section{Overview of Main Results}
\label{sec:results}
\change{
In this section, we overview our main results and discuss their implications.
}


\subsection{Insert-Only Setting}
\label{sec:results_insert_only}

In the insert-only setting, we show that off-the-shelf worst-case optimal join algorithms, e.g., LeapFrogTrieJoin~\cite{LeapFrogTrieJoin2014}, can be used to achieve the best known amortized update time for the dynamic query evaluation of arbitrary join queries, so including the {\em cyclic} queries, while supporting constant delay for full/delta enumeration.
Specifically, we give an upper bound on the single-tuple update time in terms of the fractional hypertree width $\fw(Q)$ of $Q$:

\begin{theorem}
\label{thm:main_inserts}
For any query $Q$,
both $\ivmp[Q]$ and $\ivmpd[Q]$ can be solved with $\bigO{N^{\fw(Q)-1}}$ amortized update time {\change{and non-amortized constant enumeration delay}}, where $N$ is the number of single-tuple inserts.
\end{theorem}

The upper bound in Theorem~\ref{thm:main_inserts} is met by our algorithm outlined in Sec.~\ref{sec:insert_only}.
It uses worst-case optimal join algorithms and tree decompositions.
Amortization is necessary as some inserts may be costly, while many others are necessarily relatively cheaper, so the average insert time matches the cost of computing a factorized representation of the query result divided by the number of inserts.

Since every acyclic query has a fractional hypertree width of one, Theorem~\ref{thm:main_inserts} implies that every acyclic query can be maintained with amortized constant time per insert:

\begin{corollary}[Theorem~\ref{thm:main_inserts}]
\label{cor:inserts_acyclic}
For any acyclic query $Q$, $\ivmp[Q]$ and
$\ivmpd[Q]$ can be solved with
$\bigO{1}$ amortized update time {\change{and non-amortized constant enumeration delay}}.
\end{corollary}

Corollary~\ref{cor:inserts_acyclic} recovers the amortized constant update time for acyclic joins from prior work~\cite{DBLP:journals/pvldb/WangHDY23}. It also shows that the insert-only setting can be computationally cheaper than the \insertdelete setting investigated in prior work~\cite{BerkholzKS17}: In the \insertdelete setting, the non-hierarchical acyclic queries cannot admit $\bigO{|\calD|^{1/2-\gamma}}$ update time (while keeping the enumeration delay constant) for any database $\calD$ and  
$\gamma > 0$~\cite{BerkholzKS17}, conditioned on the OMv conjecture.
\nop{
Furthermore, the \textsf{F-IVM} system~\cite{FIVM:VLDBJ:2023} reports time linear in the database size per single-tuple insert for non-hierarchical acyclic queries {\em in the worst-case}. Our finer complexity analysis for acyclic queries suggests that \textsf{F-IVM} takes in fact {\em amortized} constant time per single-tuple insert.}

\change{We accompany the upper bound from Theorem~\ref{thm:main_inserts} with a lower bound on the insertion time in terms of the lower bound $\lb(Q)$ on the static evaluation of $Q$ (Definition~\ref{defn:fw_lb}).
Unlike prior lower bounds~\cite{BerkholzKS17,DBLP:conf/icdt/KaraNNOZ19}, this lower bound is not
conditioned on a fine-grained complexity conjecture, such as the OMv conjecture~\cite{HenzingerKNS15}.
However, its proof is straightforward by viewing static evaluation as a stream of inserts.
It is meant to show that Theorem~\ref{thm:main_inserts} is optimal up to the gap between $\fw(Q)$ and $\lb(Q)$.
\footnote{\change{
More broadly, our lower bounds are meant to introduce a new framework to assess the optimality of algorithms in database theory, in terms of the function $\omega(Q)$.
}}
It also applies to {\em both} amortized and non-amortized update time and enumeration delay.
\footnote{\change{Note that a lower bound on the amortized time also implies the same lower bound on the worst-case time. This is because if an amortized algorithm does not
    exist, then a worst-case algorithm does not exist either.}}
}
\change{
\begin{proposition}
    \label{prop:insert_only_lower_bound}
    For any query $Q$ and any constant $\gamma > 0$,
    neither $\ivmp[Q]$ nor $\ivmpd[Q]$ can be solved with $\tildeO{N^{\lb(Q)-1-\gamma}}$ (amortized) update time \change{and (non-amortized) constant enumeration delay}.
\end{proposition}
}


\subsection{\InsertDelete Setting}
\label{sec:results_fully_dynamic}

In the \insertdelete setting, our approach can maintain arbitrary join queries, and in particular any cyclic join query, with update times that can be asymptotically lower than recomputation.
In particular, we give an upper bound on the update time for a query $Q$ in terms of  the fractional hypertree width of the \multivariate extension $\wh Q$ of $Q$ (Definition~\ref{defn:multivariate}):

\begin{theorem}
    \label{thm:main_fully_dynamic}
    For any query $Q$,
    both $\ivmpm[Q]$ and $\ivmpmd[Q]$ can be solved with $\tildeO{|\calD^{(\tau)}|^{\fw(\wh Q)-1}}$ amortized update time \change{and non-amortized constant enumeration delay}, where $\wh Q$ is the \multivariate extension of $Q$, and $|\calD^{(\tau)}|$ is the current database size at update time $\tau$.
\end{theorem}

Sec.~\ref{sec:forward_reduction:body} overviews our algorithm that meets the upper bound in \change{Theorem~\ref{thm:main_fully_dynamic}}. It involves {\em intersection joins} (Sec.~\ref{sec:intersection-joins}).
The following statements shed light on the relationship between a query $Q$ and its 
\multivariate extension $\wh Q$ (their proofs can be found in \shortorfull{\cite{full-version}}{Appendix~\ref{app:results}}):
\begin{itemize} 
\item $Q$ is hierarchical if and only if its \multivariate extension $\wh Q$ is acyclic, or equivalently $\fw(\wh Q) =1$.
\item If $Q$ is non-hierarchical, then $\fw(\wh Q) \geq \frac{3}{2}$. 
\item If $Q$ is Loomis-Whitney of {\em any} degree, then \mbox{$\fw(\wh Q) = \frac{3}{2}$}.
\item For {\em any} query $Q$, we have
$\fw(Q) \leq  \fw(\wh Q) \leq \fw(Q)+ 1$.
\end{itemize}

Immediate corollaries of these statements are that, in the \insertdelete setting, our approach needs: (1) amortized $\tildeO{1}$ update time for hierarchical queries; and (2) amortized $\tildeO{|\calD|^{1/2}}$ update time for the triangle query, which is the Loomis-Whitney query of degree 3. These update times mirror those given in prior work on $q$-hierarchical queries~\cite{BerkholzKS17} and the full triangle query~\cite{DBLP:journals/tods/KaraNNOZ20}, albeit our setting is \change{more restricted}: it considers full conjunctive queries and initially empty databases and the update times are amortized and have a polylog factor.



A close analysis of the lower bound proof for non-hierarchical queries in prior work \cite{BerkholzKS17} reveals that for any non-hierarchical query $Q$ and  any $\gamma >0$, there is no algorithm that solves $\ivmpm[Q]$ or $\ivmpmd[Q]$ with {\em amortized} $\tildeO{N^{1/2-\gamma}}$ update time, unless the OMv-conjecture fails \shortorfull{\cite{full-version}}{(Proposition \ref{thm:fully_dynamic_OMV_lower_bound} in 
Appendix~\ref{app:results})}. Following \change{Theorem~\ref{thm:main_fully_dynamic}} and  the fact that $\fw(\wh Q) = \frac{3}{2}$ for any Loomis-Whitney query $Q$, we conclude that both $\ivmpm[Q]$ and $\ivmpmd[Q]$ can be solved with $\tildeO{|\calD|^{1/2}}$ amortized update time, which is  optimal, unless the OMv conjecture fails\shortorfull{~\cite{full-version}}{}.


First-order IVM, i.e., delta queries, and even higher-order IVM, i.e., delta queries with materialized views, cannot achieve the update time of our approach. It was already discussed in prior work that for the triangle join query, both IVM approaches need linear update time per single-tuple update~\cite{DBLP:conf/icdt/KaraNNOZ19,DBLP:journals/tods/KaraNNOZ20}.
See Appendix~\ref{app:comparison}.
The IVM$^\epsilon$ approach resorts to a heavy/light partitioning argument that is tailored to the triangle query and does not generalize to other cyclic queries~\cite{DBLP:conf/icdt/KaraNNOZ19,DBLP:journals/tods/KaraNNOZ20}. Instead, our approach solves this problem more systematically and for {\em any} query $Q$ by translating the temporal dimension (the tuple lifespan as defined by its insert and possible delete) into spatial attributes (the multivariate encoding of the tuple lifespan), taking a tree decomposition of the multivariate extension $\hat{Q}$ of $Q$, and by materializing and maintaining the bags of this tree decomposition.

We complement the  upper bound in \change{Theorem~\ref{thm:main_fully_dynamic}} with a lower bound in terms of the lower bound $\lb(\wh Q_{\bm\sigma})$ for the static evaluation of any component $\wh Q_{\bm\sigma}$ of $\wh Q$.
It is not conditioned of a complexity conjecture.
However, it only applies to the $\ivmpmd[Q]$ problem.
\change{In particular, it is meant to show that our upper bound from Theorem~\ref{thm:main_fully_dynamic} is tight for $\ivmpmd[Q]$ up to the gap between $\fw(\wh Q)$ and $\lb(\wh Q_{\bm\sigma})$.}
\change{It also applies to both amortized and non-amortized update time and enumeration delay.}
Sec.~\ref{sec:backword_reduction:body} gives the high-level idea:

\begin{theorem}
\label{thm:fully_dynamic_lower_bound}
    Let $Q$ be a query and $\wh Q_{\bm \sigma}$ any component of its \multivariate extension. For any constant $\gamma > 0$, $\ivmpmd[Q]$ cannot be solved with \change{(amortized)} update time
    \change{$\tildeO{|\calD^{(\tau)}|^{\lb(\wh Q_{\bm\sigma})-1-\gamma}}$}
    \change{and (non-amortized) constant enumeration delay}.
\end{theorem}

\section{IVM: Insert-Only Setting}
\label{sec:insert_only}

In this section, we give an overview of our algorithm for $\ivmp[Q]$ that meets the upper bound
of Theorem~\ref{thm:main_inserts}. We leave the details and proofs to \shortorfull{\cite{full-version}}{Appendix~\ref{app:insert_only}}.
We start with the following lemma (Recall notation from Sec.~\ref{sec:prelims}):
\begin{lemma}
    \label{lem:insert_only_in_AGM}
    Given a  query $Q$, an initially empty database $\calD^{(0)}$, and a stream of $N$ single-tuple inserts,
we can compute the {\em new} output tuples $\delta_\tau Q(D)$ after every insert
$\delta_\tau \calD$, where the total computation time over {\em all} inserts is $\bigO{N + \agm(Q, \calD^{(N)})}$.
\end{lemma}
The total computation time above is the same as the AGM bound~\cite{AtseriasGM13} of $Q$ over the final database $\calD^{(N)}$. 
In particular, even in the static setting where $\calD^{(N)}$ is given 
upfront, the output size $|Q(\calD^{(N)})|$ can be as large as $\agm(Q, \calD^{(N)})$ in the worst-case.
Moreover, worst-case optimal join algorithms cannot beat this runtime~\cite{Ngo:JACM:18,LeapFrogTrieJoin2014,WCOJGemsOfPODS2018}.
The above lemma is proved based on the {\em query decomposition lemma}~\cite{SkewStrikesBack2014,WCOJGemsOfPODS2018}, which
says the following.
Let $Q$ be a query and let $\bm Y \subseteq \vars(Q)$.
Then, the AGM bound of $Q$ can be decomposed into a sum of AGM bounds of ``residual'' queries: one query $Q\Join \bm y$ for each tuple $\bm y$ over the variables $\bm Y$. \shortorfull{}{Appendix~\ref{app:insert_only} gives a proof of the query decomposition lemma and an example.}
\change{
\begin{lemma}[Query Decomposition Lemma~\cite{SkewStrikesBack2014,WCOJGemsOfPODS2018}]
    Given a query $Q$ and a subset $\bm Y \subseteq \vars(Q)$, let $\left(\lambda_{R(\bm X)}\right)_{R(\bm X)\in\atoms(Q)}$ be a fractional edge cover of $Q$.
    Then, the following inequality holds:

    \begin{equation}
        \sum_{\bm y\in \Dom(\bm Y)}
        \underbrace{\prod_{R(\bm X) \in \atoms(Q)}
        |R(\bm X)\ltimes \bm y|^{\lambda_{R(\bm X)}}
        }_{\text{AGM-bound of $Q\ltimes \bm y$}}
        \leq
        \underbrace{\prod_{R(\bm X) \in \atoms(Q)}
        |R|^{\lambda_{R(\bm X)}}}_{\text{AGM-bound of $Q$}}
        \label{eq:query-decom-lemma}
    \end{equation}
    \label{lmm:query-decom-lemma}
\end{lemma}
In the above, $\bm y\in \Dom(\bm Y)$ indicates that the tuple $\bm y$ has schema $\bm Y$. Moreover,
$R(\bm X)\ltimes \bm y$ denotes the {\em semijoin} of the atom
$R(\bm X)$ with the tuple $\bm y$.
}


\begin{example}[for Lemma~\ref{lem:insert_only_in_AGM}]
    \label{ex:insert-only:agm}
    Consider the triangle query $Q_\triangle$ from Eq.~\eqref{eq:triangle}.
    Lemma~\ref{lem:insert_only_in_AGM} says that given a stream of $N$ single-tuple inserts
    into $R, S$, and $T$, we can compute the new output tuples after every insert
    in a total time of $\bigO{N^{3/2}}$.
    To achieve this, we need to maintain two indices for $R(A, B)$: one index that sorts $R$ first by $A$ and then by $B$, while the other sorts $R$ first by $B$ and then by $A$.
    For every insert to $R$, we update the two indices simultaneously.
    Similarly, we maintain two indices for each of $S$ and $T$.

    Following notation from Sec.~\ref{sec:prelims}, let $R^{(0)}\defeq \emptyset, R^{(1)}, \ldots, R^{(N)}$ be the sequence of versions of $R$ after each insert, and the same
    goes for $S$ and $T$.
    Suppose that the $\tau$-th insert in the stream is inserting a tuple $(a, b)$ into $R(A, B)$. The new output tuples that are added due to this insert correspond to the output
    of the following query:
    (Below, we drop the database instance $\calD$ from the notation $\delta_\tau Q_\triangle(\calD)$
    since $\calD$ is clear from the context.)
    \begin{align*}
    \delta_\tau Q_\triangle(A, B, C) \defeq \sigma_{B=b}S^{(\tau)}(B,C) \wedge \sigma_{A=a}T^{(\tau)}(A, C)
    \end{align*}
    The output size of this query is upper bounded by:
    \begin{align}
        \min(|\sigma_{B=b}S^{(\tau)}(B,C)|, |\sigma_{A=a}T^{(\tau)}(A, C)|)\leq
        \sqrt{|\sigma_{B=b}S^{(\tau)}(B,C)|\cdot |\sigma_{A=a}T^{(\tau)}(A, C)|}
        \label{eq:lemma_inserts_only_triangle_0}
    \end{align}
    Moreover, the above query can be computed in time within $\bigO{1}$ factor
    from the quantity in Eq.~\eqref{eq:lemma_inserts_only_triangle_0}.
    To achieve this time, we need to use the index for $S(B, C)$ that is indexed
    by $B$ first, and we also need the index for $T(A, C)$ that is indexed by $A$ first.
    Because $S^{(\tau)}\subseteq S^{(N)}$ and $T^{(\tau)}\subseteq T^{(N)}$, the quantity in Eq.~\eqref{eq:lemma_inserts_only_triangle_0} is bounded by:
    \begin{align}
        \sqrt{|\sigma_{B=b}S^{(N)}(B,C)|\cdot |\sigma_{A=a}T^{(N)}(A, C)|}
        \label{eq:lemma_inserts_only_triangle_1}
    \end{align}
    The query decomposition lemma
    \shortorfull{}{(specifically Eq.~\eqref{eq:query-decom-lemma:triangle} in Example~\ref{ex:query-decom-lemma:triangle})}
    says that the sum of the quantity in Eq.~\eqref{eq:lemma_inserts_only_triangle_1} over all $(a, b) \in R^{(N)}$
    is bounded by
    $\sqrt{|R^{(N)}|\cdot |S^{(N)}|\cdot |T^{(N)}|}\leq N^{3/2}$.
    Hence, all inserts into $R$ take time $\bigO{N^{3/2}}$.

    Now suppose that the $\tau$-th insert in the stream is inserting a tuple $(b, c)$ into $S(B, C)$. To handle this insert, we compute the query:
    \begin{align*}
        \delta_\tau Q_\triangle(A, B, C) \defeq  \sigma_{B=b}R^{(\tau)}(A, B) \wedge \sigma_{C=c}T^{(\tau)}(A,C)
    \end{align*}
    To compute this query in the desired time, we need to use the index
    for $T(A, C)$ that is indexed by $C$ first. This is why we need to maintain
    two indices for $T(A, C)$. The same goes for $R$ and $S$.
\end{example}

Instead of the AGM-bound, the upper bound in Theorem~\ref{thm:main_inserts} is given in terms of the fractional
hypertree width of $Q$. To achieve this bound, we take an optimal tree decomposition
of $Q$ and maintain a materialized relation for every bag in the tree decomposition using
Lemma~\ref{lem:insert_only_in_AGM}.
In order to support constant-delay enumeration of the output, we ``calibrate'' the bags
by semijoin reducing adjacent bags with one another.
The following example illustrates this idea.
\begin{example}[for Theorem~\ref{thm:main_inserts}, $\ivmp$]
    \label{ex:insert_only_full}
    Suppose we want to solve $\ivmp[Q]$
    for the following query $Q$ consisting of two adjacent triangles: One triangle over $\{A, B, C\}$ and the other over $\{B, C, D\}$.

    \begin{equation}
        Q(A, B, C, D) = R(A, B) \wedge S(B, C) \wedge T(A, C) \wedge U(B, D) \wedge V(C, D)
        \label{eq:2triangles}
    \end{equation}
    
    The above query $Q$ has a fractional hypertree width of $\frac{3}{2}$. In particular, one optimal tree decomposition of $Q$
    consists of two bags: One child bag $B_1 =\{A, B, C\}$ and another root bag
    $B_2=\{B, C, D\}$.
    Theorem~\ref{thm:main_inserts} says that given a stream of $N$ single-tuple inserts,
    $Q$ can be updated in a total time of $\bigO{N^{3/2}}$ where we can do
    constant-delay enumeration of the output after every insert. To achieve this, we maintain
    the following query plan:

    \begin{align*}
        Q_1(A, B, C) &= R(A, B) \wedge S(B, C) \wedge T(A, C)\\
        P_1(B, C) &= Q_1(A, B, C)\\
        Q_2'(B, C, D) &= P_1(B, C) \wedge S(B, C) \wedge U(B, D) \wedge V(C, D)
    \end{align*}
    
    By Lemma~\ref{lem:insert_only_in_AGM}, $Q_1$ above can be updated in a total time of $\bigO{N^{3/2}}$, as shown in Example~\ref{ex:insert-only:agm}.
    Also, $P_1$ can be updated in the same time because it is just a projection of $Q_1$.
    Regardless of the size of $P_1$, the query $Q_2'$ can be updated in a total time of $\bigO{N^{3/2}}$ because
    its AGM bound is upper bounded by the input relations $S, U$ and $V$. To do constant-delay enumeration of $Q(A, B, C, D)$, we enumerate $(b, c, d)$ from $Q_2'$,
    and for every $(b, c)$, we enumerate the corresponding $A$-values from $Q_1$.
    Note that at least one $A$-value must exist because $Q_2'$ includes $P_1$, which is
    the projection of $Q_1$.
\end{example}
\change{
In the above example, we only calibrate bottom-up from the leaf $Q_1$ to the root $Q_2'$. However, this is not sufficient for $\ivmpd[Q]$. Instead, we also need to calibrate top-down, as the following example demonstrates.
}

\change{
\begin{example}[for Theorem~\ref{thm:main_inserts}, $\ivmpd$]
    \label{ex:insert_only_delta}
    Consider again the query~\eqref{eq:2triangles} from Example~\ref{ex:insert_only_full}.
    Suppose now that we want to extend our solution from Example~\ref{ex:insert_only_full}
    to the $\ivmpd[Q]$ problem.
    In particular, suppose we have an insert of a tuple $(a, b)$ into $R$ above.
    In order to enumerate the new output tuples corresponding to $(a, b)$, we have to start our enumeration from $Q_1$ as the root. For that purpose, $Q_1$ also needs to be calibrated
    with $Q_2'$.
    To that end, in addition to $Q_1, P_1$ and $Q_2'$ that are defined in Example~\ref{ex:insert_only_full},
    we also need to maintain the following relations:
    \begin{align*}
        P_1''(B, C) &= Q_2'(B, C, D)\\
        Q_1''(A, B, C) &= Q_1(A, B, C) \wedge P_1''(B, C)\\
    \end{align*}
    The AGM bound of $Q_1''$ is still $N^{3/2}$, and so is the AGM bound of $P_1''$.
    Hence, we can maintain all inserts into them in a total time of $\bigO{N^{3/2}}$.
    Moreover, note that by definition, $Q_1''$ must be the projection of $Q(A,B,C,D)$
    onto $\{A, B, C\}$.

    Now, suppose we have an insert of a tuple $\bm t$ into either one of the relations $S, U$ or $V$. Then, we can enumerate the new output tuples that are added due to this insert by starting from  $Q_2'$, just like we did in Example~\ref{ex:insert_only_full}.
    However, if we have inserts into $R$ or $T$, then we have to start our enumeration from $Q_1''$. For each output tuple $(a, b, c)$ of $Q_1''$ that joins with $\bm t$, we enumerate the corresponding $D$-values from $Q_2'$. There must be at least one $D$-value because $Q_1''$ is the projection of the output $Q$ onto $\{A, B, C\}$.
\end{example}
}
\section{Technical Tools: Intersection Joins}
\label{sec:intersection-joins}
Before we explain our upper and lower bounds for IVM in the \insertdelete setting in Sec.~\ref{sec:fully_dynamic}, we review in this section some necessary background
on intersection joins in the {\em static} setting.

\begin{definition}[Intersection queries~\cite{KhamisCKO22}]
    \label{defn:intersection-join}
An {\em intersection query} is a (full conjunctive) query
that contains {\em interval variables}, which are variables
whose domains are discrete intervals over the natural numbers. 
For instance, the discrete interval $[3,7]$ consists of the natural 
numbers from $3$ to $7$.
For better distinction, we denote an interval variable
by $[X]$ and call the variables that are not interval variables {\em point variables}.
A value of an interval variable $[X]$ is denoted by $[x]$, which is an interval.
\nop{
}
Two or more intervals {\em join} if they {\em intersect}.
In particular, 
the semantics of an intersection query
$Q = R_1(\bm X_1) \wedge \cdots \wedge R_k(\bm X_k)$ 
is defined as follows.
A tuple $\bm t$ over the schema $\vars(Q)$
is in the result of $Q$ if there are tuples $\bm t_i \in R_i$ for 
$i \in [k]$ such that for all $X \in \vars(Q)$, it holds:
\begin{itemize}[leftmargin=*]
\item if $X$ is a point variable, then the set of points
$\{\bm t_i(X)\mid i \in [k] \wedge X \in \bm X_i\}$ consists of a single value,
which is the point $\bm t(X)$;
\item if $[X]$ is an interval variable, then the intervals in the set
$\{\bm t_i([X]) \mid i \in [k] \wedge [X] \in \bm X_i\}$ have a non-empty intersection,
which is the interval $\bm t([X])$. 
\end{itemize}
\end{definition}

In this paper, we are only interested in a special class of intersection 
queries, defined below:

\begin{definition}[The \univariate extension $\ov Q$ of a query $Q$]
    \label{defn:univariate}
    Let $Q = R_1(\bm X_1)\wedge \cdots\wedge R_k(\bm X_k)$ be a query where $\vars(Q)$
    are all point variables, and let $[Z]$ be a new interval variable.
    The {\em \univariate} extension 
    $\ov{Q}$ of $Q$ is an intersection query that results 
    from $Q$ by adding $[Z]$ to the schema 
    of each atom:

    \begin{align}
    \label{eq:univariate}
    \ov{Q} = \ov{R}_1([Z], \bm X_1)\wedge \cdots\wedge \ov{R}_k([Z],\bm X_k)
    \end{align}
    
    A query $\ov Q$ is called a {\em \univariate extension} if it is the \univariate extension of some query $Q$.
\end{definition}
We call $\ov Q$ the ``\univariate extension'' of $Q$ because in Sec.~\ref{sec:forward_reduction:body}, we use the interval $[Z]$ in $\ov Q$ to represent the ``lifespan'' of tuples in $Q$.

\begin{example}
    \label{ex:triangle-query:uni}
Consider this intersection query, which is the \univariate extension of $Q_\triangle$ from Eq.~\eqref{eq:triangle}:
\begin{align}
    \hspace{-0.4cm}
    \ov Q_\triangle([Z], A, B, C) = \ov R([Z], A, B) \wedge \ov S([Z], B, C) \wedge \ov T([Z], A, C)
    \label{eq:triangle-query:uni}
\end{align}
In the above, $[Z]$ is an interval variable while $A, B, C$ are point variables.
Consider the database instance $\ov \calD$ depicted in Figure~\ref{table:intersection-join-example}. The output of $\ov Q_\triangle$ on this instance
is depicted in Figure~\ref{table:intersection-join-example:Q}.
\end{example}
\begin{figure}[ht!]
    \begin{subtable}[t]{.3\columnwidth}
    \centering
    \begin{tabular}[t]{c|c|c|}
        $[Z]$ & $A$ & $B$ \\\hline
        $[1, 8]$ & $a_1$ & $b_1$\\
        & &
    \end{tabular}
    \caption{\label{table:intersection-join-example:R}$\ov R$}
    \end{subtable}
    \begin{subtable}[t]{.3\columnwidth}
    \centering
    \begin{tabular}[t]{c|c|c|}
        $[Z]$ & $B$ & $C$ \\\hline
        $[2, 5]$ & $b_1$ & $c_1$\\
        $[4, 6]$ & $b_2$ & $c_1$
    \end{tabular}
    \caption{\label{table:intersection-join-example:S}$\ov S$}
    \end{subtable}
    \begin{subtable}[t]{.3\columnwidth}
    \centering
    \begin{tabular}[t]{c|c|c|}
        $[Z]$ & $A$ & $C$ \\\hline
        $[3, 7]$ & $a_1$ & $c_1$\\
        & &
    \end{tabular}
    \caption{\label{table:intersection-join-example:T}$\ov T$}
    \end{subtable}
    \begin{subtable}[t]{\columnwidth}
        \centering
        \begin{tabular}[t]{c|c|c|c|}
            $[Z]$ & $A$ & $B$ & $C$ \\\hline
            $[3, 5]$ & $a_1$ & $b_1$ & $c_1$
        \end{tabular}
    \caption{\label{table:intersection-join-example:Q}$\ov Q_\triangle$}
    \end{subtable}
    \caption{A database instance $\ov \calD$ for the query $\ov Q_\triangle$ in Eq.~\eqref{eq:triangle-query:uni}.}
    \Description{A database instance $\ov \calD$ for the query $\ov Q_\triangle$ in Eq.~\eqref{eq:triangle-query:uni}.}
    \label{table:intersection-join-example}
\end{figure}

Let $\ov Q$ be the \univariate extension of a query $Q$.
Prior work reduces the evaluation of $\ov Q$ 
to the evaluation of $\wh Q$, where 
$\wh Q$ is the multivariate extension of $Q$ (Eq.~\eqref{eq:multivariate}), whose variables are all point variables~\cite{KhamisCKO22}.
Moreover, it is shown in~\cite{KhamisCKO22} that this reduction is {\em optimal} (up to a $\tildeO{1}$ factor). In particular, the query $\ov Q$ is exactly as hard as the hardest component $\wh Q_{\bm\sigma}$ in $\wh Q$. This is shown by a backward reduction from (the static evaluation of) each component $\wh Q_{\bm\sigma}$
to $\ov Q$.
\change{We summarize both reductions below and defer the details to Appendix~\ref{app:intersection-joins}
and~\cite{KhamisCKO22}.}
We rely on both reductions 
to prove upper and lower bounds for IVM in the \insertdelete setting
(Sec.~\ref{sec:results_fully_dynamic}).

\change{
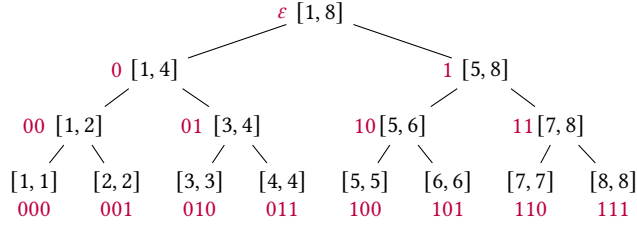
\begin{figure}
    \begin{tikzpicture}
        \node at (0,0)  (root) {$[1,8]$};
        \node at (-2.2,-0.75)  (c1) {$[1,4]$} edge[-] (root);
        \node at (2.2,-0.75)  (c2) {$[5,8]$} edge[-] (root);
        
        \node at (-3.2,-1.5)  (c21) {$[1,2]$} edge[-] (c1);
        \node at (-1.1,-1.5)  (c22) {$[3,4]$} edge[-] (c1);
        
        \node at (1.1,-1.5)  (c23) {$[5,6]$} edge[-] (c2);
        \node at (3.2,-1.5)  (c24) {$[7,8]$} edge[-] (c2);

        \node at (-3.8,-2.25)  (c31) {$[1,1]$} edge[-] (c21);
        \node at (-2.7,-2.25)  (c32) {$[2,2]$} edge[-] (c21);
        \node at (-1.6,-2.25)  (c33) {$[3,3]$} edge[-] (c22);
        \node at (-0.5,-2.25)  (c34) {$[4,4]$} edge[-] (c22);
        \node at (0.6,-2.25)  (c35) {$[5,5]$} edge[-] (c23);
        \node at (1.7,-2.25)  (c36) {$[6,6]$} edge[-] (c23);
        \node at (2.8,-2.25)  (c37) {$[7,7]$} edge[-] (c24);
        \node at (3.9,-2.25)  (c38) {$[8,8]$} edge[-] (c24);
        
         \node at (-0.5,0)  (x) {\color{purple}$\varepsilon$};
        \node at (-2.7,-0.75)  (x) {\color{purple}$0$};
        \node at (1.7,-0.75)  (x) {\color{purple}$1$};
        \node at (-3.8,-1.5)  (x) {\color{purple}$00$};
        \node at (-1.7,-1.5)  (x) {\color{purple}$01$};
        \node at (0.6,-1.5)  (x) {\color{purple}$10$};
        \node at (2.7,-1.5)  (x) {\color{purple}$11$};
        
        \node at (-3.8,-2.6)  (x) {\color{purple} $000$};
        \node at (-2.7,-2.6)  (x) {\color{purple}$001$};
        \node at (-1.6,-2.6)  (x) {\color{purple}$010$};
        \node at (-0.5,-2.6)  (x) {\color{purple}$011$};
        \node at (0.6,-2.6)  (x) {\color{purple}$100$};
        \node at (1.7,-2.6)  (x) {\color{purple}$101$};
        \node at (2.8,-2.6)  (x) {\color{purple}$110$};
        \node at (3.9,-2.6)  (x) {\color{purple}$111$};
    \end{tikzpicture}
    \caption{\change{A segment tree $\calT_8$. Each node is identified by a {\color{purple}bitstring}.}}
    \Description{A segment tree $\calT_8$. Each node is identified by a {\color{purple}bitstring}.}
    \label{fig:segment-tree}
\end{figure}
}

\subsection{Reduction from $\ov Q$ to $\wh Q$}
\label{sec:body:intersection-joins:forward}
Suppose we want to evaluate $\ov Q$ on a given database instance $\ov \calD$.
We construct a {\em segment tree}, $\calT_N$, over the interval variable $[Z]$,
where the parameter $N$ is roughly the maximum number of different $Z$-values.
A segment tree $\calT_N$ is a balanced binary tree with $N$ leaves.
Each node corresponds to an interval. The root corresponds
to the interval $[1, N]$, its left and right child correspond to the intervals $[1, N/2]$ and $[N/2+1, N]$ respectively, and so on.
Each node can be represented by a bitstring of length at most $\log_2 N$.
Figure~\ref{fig:segment-tree} depicts the segment tree $\calT_8$.
Each interval $[z]$ can be broken down into at most
$\bigO{\log N}$ nodes in $\calT_N$. We refer to these nodes as the {\em canonical partition}
of $[z]$, and denote them by $\canpart_N([z])$.
For example, the canonical partitions of $[2, 8]$ and $[2, 5]$ in the segment tree $\calT_8$ 
are:

\begin{align}
    \canpart_8([2, 8]) = \{001, 01, 1\},\quad\quad
    \canpart_8([2, 5]) = \{001, 01, 100\}
    \label{eq:canpart28:25}
\end{align}

The canonical partition of the intersection of some intervals is roughly the intersection of their canonical partitions \shortorfull{\cite{KhamisCKO22}}{(Lemma~\ref{lem:interv_dec})}.
\change{Two (or more) nodes in the segment tree correspond to overlapping intervals
if and only if one node is an ancestor of the other, which can only happen if one of the two corresponding bitstrings is a prefix of the other \shortorfull{~\cite{KhamisCKO22}}{(Lemma~\ref{lem:interv_dec})}.}
For example, the nodes labeled by $1$ and $100$ in Figure~\ref{fig:segment-tree} correspond to two overlapping intervals, namely $[5, 8]$ and $[5, 5]$, because the string $1$ is a prefix of $100$.
Therefore, testing intersections can be reduced to testing whether some bitstrings
form a chain of prefixes.
Armed with this idea, we convert $\ov\calD$ into a database instance $\wh\calD$ for 
$\wh Q$, called the {\em canonical partition} of $\ov\calD$ and denoted by $\canpart_N(\ov\calD)$. There is a mapping between the outputs of $\ov Q(\ov\calD)$ and $\wh Q(\wh \calD)$\change{, as the following example shows.}
\shortorfull{}{See Theorem~\ref{thm:IJ:forward:reduction} for more details.}
\begin{example}
    \label{ex:IJ:forward:reduction:body}
    Suppose we want to reduce $\ov Q_\triangle$ 
    from Eq.~\eqref{eq:triangle-query:uni} to $\wh Q_\triangle$ from Example~\ref{ex:transformations}. In particular, we are given the database instance $\ov \calD$ from Figure~\ref{table:intersection-join-example} and want to compute the corresponding 
    instance $\wh \calD$ for $\wh Q_\triangle$. We  show how to construct
    $\wh R_{123}, \wh S_{123}, \wh T_{123}$ for $\wh Q_{123}$ in Eq.~\eqref{eq:triangle-query:multi-123:body}. The remaining $\wh Q_{132},\ldots, \wh Q_{321}$
    are similar.
    The query $\wh Q_{123}$ is meant to test whether $\ov R, \ov S, \ov T$ contain
    three intervals $[z^R], [z^S], [z^T]$ whose canonical partitions contain
    bitstrings $z^R, z^S, z^T$ where $z^R$ is a prefix of $z^S$ which is a prefix of $z^T$. This can only happen
    if there are bitstrings $z_2, z_3$ where $z^S = z^R\circ z_2$ and $z^T = z^S \circ z_3$, where $\circ$ denotes string concatenation. \change{To that end,} we define:
    \begin{align}
        \wh R_{123} &\defeq \{(z_1, a, b) \mid  \exists [z] : ([z], a, b)\in \ov R \wedge z_1\in \canpart_N([z])\}\label{eq:triangle:canpart:body}\\
        \wh S_{123} &\defeq \{(z_1, z_2, b, c) \mid  \exists [z] : ([z], b, c)\in \ov S \wedge (z_1\circ z_2)\in \canpart_N([z])\}\nonumber\\
        \wh T_{123} &\defeq \{(z_1, z_2, z_3, a, c) \mid  \exists z : ([z], a, c)\in \ov T \wedge (z_1\circ z_2\circ z_3)\in \canpart_N([z])\}\nonumber
    \end{align}
    \change{Note that $|\wh T_{123}| = \bigO{|\ov T|\cdot \polylog(N)}$ and can be constructed
    in time $\bigO{|\ov T|\cdot \polylog(N)}$, because the height of the segment tree is $\bigO{\log N}$.
    The same goes for $\wh R_{123}$ and $\wh S_{123}$.}
    Let $\varepsilon$ denote the empty string.
    Applying the above to $\ov \calD$ from Figure~\ref{table:intersection-join-example},
    $\wh R_{123}, \wh S_{123}, \wh T_{123}$ respectively contain the tuples:
    $(\varepsilon, a_1, b_1)$, $(\varepsilon, 01, b_1, c_1)$, $(\varepsilon, 01, \varepsilon, a_1, c_1)$, which join together producing the output tuple $(\varepsilon, 01, \varepsilon, a_1, b_1, c_1)$
    of $\wh Q_{123}$. This is a witness to the output tuple $([3, 5], a_1, b_1, c_1)$ of $\ov Q_\triangle$.
    \change{The full answer to $\ov Q_\triangle$, given by Figure~\ref{table:intersection-join-example:Q}, can be retrieved from the union of the answers to the six queries $\wh Q_{132},\ldots, \wh Q_{321}$.}
    \shortorfull{}{See Example~\ref{ex:IJ:forward:reduction} for more details.}
\end{example}

\subsection{Reduction from $\wh Q$ to $\ov Q$}
\label{sec:body:intersection-joins:backward}
\change{
We now summarize the backward reduction from the static evaluation of each component $\wh Q_{\bm\sigma}$ of $\wh Q$ back to $\ov Q$, based on~\cite{KhamisCKO22}.
We will utilize this reduction later to establish our lower bound in Sec.~\ref{sec:backword_reduction:body}.
\begin{example}
    \label{ex:IJ:backward:reduction}
Continuing with Example~\ref{ex:triangle-query:uni}, let us consider the query
$\wh Q_{123}$ from~\eqref{eq:triangle-query:multi-123:body}.
Let $\wh\calD_{123}'$ be an {\em arbitrary} database instance for $\wh Q_{123}$.
(Unlike Example~\ref{ex:triangle-query:uni}, here we cannot make any assumption about
how $\wh \calD_{123}'$ was constructed.)
We show how to use an oracle for the intersection query $\ov Q_\triangle$ from~\eqref{eq:triangle-query:uni} in order to compute $\wh Q_{123}(\wh\calD_{123}')$ in the same time.
WLOG we can assume that each value $x$ that appears in $\wh \calD_{123}'$ is a bitstring
of length exactly $\ell$ for some constant $\ell$. (If a bitstring has length less than $\ell$, we can pad it with zeros.) We construct a segment tree $\calT_{N'}$ where $N'\defeq 2^{3\ell}$, i.e. a segment tree of depth $3 \ell$. (We chose $3$ in this example because it is the number of atoms in $\wh Q_{123}$.)
Each node $v$ in the segment tree corresponds to an interval that is contained in $[N']$
and is identified by a bitstring of length at most $3\ell$.
Given a bitstring $b$ of length at most $3\ell$, let $\segment_{N'}(b)$ denote
the corresponding interval in the segment tree $\calT_{N'}$.
We construct the following database instance $\ov\calD'$
for $\ov Q_\triangle$:
\begin{align*}
    \ov R' &= \{(\segment_{N'}(z_1), a, b)\mid (z_1, a, b) \in \wh R_{123}'\}\\
    \ov S' &= \{(\segment_{N'}(z_1\circ z_2), b, c)\mid (z_1, z_2, b, c) \in \wh S_{123}'\}\\
    \ov T' &= \{(\segment_{N'}(z_1\circ z_2\circ z_3), a, c)\mid (z_1, z_2, z_3, a, c) \in \wh T_{123}'\}
\end{align*}
Following~\cite{KhamisCKO22}, we can show that there is a one-to-one mapping 
between the output tuples of $\ov Q_\triangle(\ov\calD')$ and the output tuples of $\wh Q_{123}(\wh\calD_{123}')$.
This follows from the observation that the three intervals
$\segment_{N'}(z_1^R)$, $\segment_{N'}(z_1^S\circ z_2^S)$,
$\segment_{N'}(z_1^T\circ z_2^T\circ z_3^T)$ overlap if and only if
$z_1^R = z_1^S = z_1^T$ and $z_2^S=z_2^T$.
(Recall that the bitstrings $z_1^R, z_1^S, \ldots$ all have the same length $\ell$.)
Moreover note that $\ov R', \ov S', \ov T'$ have the same sizes as $\wh R_{123}', \wh S_{123}', \wh\calT_{123}'$ respectively and they can be constructed in linear time.
\end{example}
}
\section{IVM: \InsertDelete Setting}
\label{sec:fully_dynamic}

In this section, we give a brief overview of how we obtain our upper and lower bounds for IVM
in the insert-delete setting from \change{Theorem~\ref{thm:main_fully_dynamic} and~\ref{thm:fully_dynamic_lower_bound}}, respectively. Details are deferred to Appendix~\ref{app:fully_dynamic}.

\subsection{Upper bound for $\ivmpm[Q]$}
\label{sec:forward_reduction:body}
\change{
In order to prove Theorem~\ref{thm:main_fully_dynamic}, we start by describing an
algorithm that meets a weaker version of the upper bound in Theorem~\ref{thm:main_fully_dynamic}, where the current database size $|\calD^{(\tau)}|$ is replaced by the number of single-tuple updates $N$: (Note that 
$|\calD^{(\tau)}| \leq N$ and could be unboundedly smaller.)
\begin{lemma}
    \label{lmm:main_fully_dynamic}
    For any query $Q$,
    both $\ivmpm[Q]$ and $\ivmpmd[Q]$ can be solved with $\tildeO{N^{\fw(\wh Q)-1}}$ amortized update time and non-amortized constant enumeration delay, where
    $\wh Q$ is the \multivariate extension of $Q$ and
    $N$ is the number of single-tuple updates.
\end{lemma}
In particular, we can show that we can use the algorithm from Lemma~\ref{lmm:main_fully_dynamic} as a black box in order to meet the stronger upper bound in Theorem~\ref{thm:main_fully_dynamic}.
The main reason why the algorithm from Lemma~\ref{lmm:main_fully_dynamic} does not immediately
meet the upper bound from Theorem~\ref{thm:main_fully_dynamic} is that $N$ can grow much larger than the database size $|\calD|$, especially in the scenario where the update stream contains many deletes.
However, whenever that happens, we can ``reset'' the algorithm from Lemma~\ref{lmm:main_fully_dynamic} by restarting from scratch and inserting all the tuples in the current database $\calD$ as a stream of $|\calD|$ inserts.
If this is done carefully, then we can ensure that the total number of updates since the last reset is not significantly larger than $|\calD|$.
\shortorfull{}{Appendix~\ref{sec:forward_reduction:calD} gives the details.}
}

\change{We now focus on demonstrating our algorithm that proves Lemma~\ref{lmm:main_fully_dynamic} for $\ivmpm[Q]$.
\shortorfull{}{Algorithm~\ref{alg:ivmpm} in Appendix~\ref{sec:forward_reduction} 
gives the general algorithm.
Appendix~\ref{sec:forward_reduction:delta_version} shows how to extend it to $\ivmpmd[Q]$.}}
\begin{example}[for Lemma~\ref{lmm:main_fully_dynamic}]
    \label{ex:main_fully_dynamic}
    Suppose we want to solve the problem $\ivmpm[Q_\triangle]$ for $Q_\triangle$ in Eq.~\eqref{eq:triangle}, where we have a stream of $N$ single-tuple inserts/deletes into
    $R, S$ and $T$. For simplicity, assume that $N$ is known in advance.
    \change{(If $N$ is not known in advance, we can initially assume $N$ to be a constant and keep
    doubling $N$ every time we exceed it.\shortorfull{}{ See Remark~\ref{rmk:doubling:N} for details.})}
    We introduce a new interval variable $[Z]$ and use it to represent {\em time}. Specifically,
    $[Z]$ will represent the {\em lifespan} of every tuple in the database, in the spirit of temporal databases~\cite{DBLP:reference/db/JensenS18k}.
    By adding $[Z]$ to every atom, we obtain the \univariate extension $\ov Q_\triangle$ of $Q_\triangle$ in Eq.~\eqref{eq:triangle-query:uni}.
    Suppose that the $\tau$-th update is an insert of a tuple $(a, b)$ into $R$, i.e. $+R(a, b)$.
    Then, we apply the insert $+\ov R([\tau, \infty], a, b)$ into the \univariate extension,
    indicating that the tuple $(a, b)$ lives in $R$ from time $\tau$ on (since we don't
    know yet its future deletion time).
    Now suppose that the $\tau'$-th update (for some $\tau' >\tau$) is a delete of the same tuple $(a, b)$ in $R$.
    Then, we replace the tuple $([\tau, \infty], a, b)$ in $\ov R$ with $([\tau, \tau'], a, b)$.
    We can extract the result of $Q_\triangle$ at any time $\tau$ by selecting output tuples $([z], a, b, c)$ of
    $\ov Q_\triangle$ where the interval $[z]$ contains the current time $\tau$.
    Therefore, if we can efficiently maintain $\ov Q_\triangle$, we can efficiently maintain $Q_\triangle$.
    Table~\ref{tab:triangle:updates} shows the same stream of 8 updates into $Q_\triangle$ as~Table~\ref{tab:triangle:full_delta_enumeration} but adds the corresponding updates to $\ov Q_\triangle$. The final database $\ov \calD$
    after all updates have taken place is the same as the one shown in Figure~\ref{table:intersection-join-example}.
    
    \begin{table}[ht!]
        \begin{tabular}{|c|l|l|}
            \hline
            $\tau$ & $\delta_\tau\calD$ & $\delta_\tau\ov\calD$ \\\hline\hline 
            1& $+R(a_1, b_1)$ & $+\ov R([1, \infty], a_1, b_1)$ \\\hline
            2& $+S(b_1, c_1)$ & $+\ov S([2, \infty], b_1, c_1)$ \\\hline
            3& $+T(a_1, c_1)$ & $+\ov T([3, \infty], a_1, c_1)$ \\\hline
            4& $+S(b_2, c_1)$ & $+\ov S([4, \infty], b_2, c_1)$ \\\hline
            5& $-S(b_1, c_1)$ & $-\ov S([2, \infty], b_1, c_1)$, $+\ov S([2, 5], b_1, c_1)$\\\hline
            6& $-S(b_2, c_1)$ & $-\ov S([4, \infty], b_2, c_1)$, $+\ov S([4, 6], b_2, c_1)$\\\hline
            7& $-T(a_1, c_1)$ & $-\ov T([3, \infty], a_1, c_1)$, $+\ov T([3, 7], a_1, c_1)$ \\\hline
            8& $-R(a_1, b_1)$ & $-\ov R([1, \infty], a_1, b_1)$, $+\ov R([1, 8], a_1, b_1)$ \\\hline
        \end{tabular}
        \caption{Updates to $Q_\triangle$~\eqref{eq:triangle} (same as Table~\ref{tab:triangle:full_delta_enumeration}) along with the corresponding updates to $\ov Q_\triangle$~\eqref{eq:triangle-query:uni}. At the end (i.e.~$\tau = 8$), the database $\ov \calD$
        is the same as the one shown in Figure~\ref{table:intersection-join-example}.}
        \label{tab:triangle:updates}
        \vskip -0.6cm
    \end{table}
    
    To maintain $\ov Q_\triangle$, we use the six components  $\wh Q_{123}, \ldots, \wh Q_{321}$ of the \multivariate extension $\wh Q_\triangle$  of $Q_\triangle$ from Example~\ref{ex:transformations}.
    In particular, we construct and maintain the relations $\wh R_{123}, \wh S_{123}$
    and $\wh T_{123}$ from Eq.~\eqref{eq:triangle:canpart:body} and forth.
    We use them along with Theorem~\ref{thm:main_inserts} to maintain $\wh Q_{123}$ (and the same
    goes for $\wh Q_{132}, \ldots$).
    For example, suppose that the number of updates $N$ is $8$, and consider the insert
    $+S(b_1, c_1)$ at time $2$ in Table~\ref{tab:triangle:updates}.
    This insert corresponds to inserting the following tuples into $\wh S_{123}$:
    (Recall Eq.~\eqref{eq:canpart28:25}. Below, $z_1$ and $z_2$ are bitstrings and $z_1\circ z_2$ is their concatenation.)

    \begin{align*}
        +\{(z_1, z_2, b_1, c_1) \mid (z_1\circ z_2)\in \canpart_8([2, 8])\}
    \end{align*}
    The corresponding delete $-S(b_1, c_1)$ at time 5 corresponds to
    deleting and inserting the following two sets of tuples from/to $\wh S_{123}$:
    \begin{align*}
        &-\{(z_1, z_2, b_1, c_1) \mid (z_1\circ z_2)\in \canpart_8([2, 8])\}
        &+\{(z_1, z_2, b_1, c_1) \mid (z_1\circ z_2)\in \canpart_8([2, 5])\}
    \end{align*}
    
    Out of the box, Theorem~\ref{thm:main_inserts} is only limited to inserts. However, we can
    show that these pairs of inserts/deletes have a special structure that allows us to maintain
    the guarantees of Theorem~\ref{thm:main_inserts}.
    In particular, at time 5, when we truncate a tuple $\bm t$ from $[2, 8]$ to $[2, 5]$,
    even after truncation, $\bm t$ still joins with the same set of tuples that it used to join with before the truncation. This is because the only intervals that overlap with $[2, 8]$
    but not with $[2, 5]$ are the ones that start after 5. But these intervals are in the {\em future}, hence
    they don't exist yet in the database\shortorfull{}{ (see Proposition~\ref{prop:trunc})}.
    We use this idea to amortize the cost of truncations over the inserts.

    For $\ivmpm[Q_\triangle]$, our target is to do
    constant-delay enumeration of the full output $Q_\triangle$ at time $\tau$. To that end, we enumerate tuples
    $([z], a, b, c)$ of $\ov Q_\triangle$ where the interval $[z]$ contains the current time $\tau$. And to enumerate those, we enumerate tuples
    $(z_1, z_2, z_3, a, b, c)$ from $\wh Q_{123}, \ldots, \wh Q_{321}$ where $z_1\circ z_2\circ z_3$
    corresponds to an interval in $\canpart_N([z])$ that contains $\tau$. This is a {\em selection} condition
    over $(z_1, z_2, z_3)$. By construction, the tree decomposition of any component, say $\wh Q_{123}$, must contain a bag $W$ that contains the variables $\{z_1, z_2, z_3\}$.
    This is because, by Definition~\ref{defn:multivariate}, $\wh Q_{123}$ contains an atom $\wh T_{123}$ that contains these variables. We designate $W$ as the root of tree decomposition of $\wh Q_{123}$ and start our constant-delay enumeration from $W$.
    The same goes for $\wh Q_{132}, \ldots, \wh Q_{321}$.
    Moreover, the enumeration outputs of $\wh Q_{123}, \ldots, \wh Q_{321}$ are {\em disjoint}
    in this case.
    \shortorfull{}{See the proof of Lemma~\ref{lem:reduction_ivmtilde_ivmstar}.}

    For $\ivmpmd[Q_\triangle]$, our target is to enumerate the change to the output at time
    $\tau$. To that end, we enumerate tuples $([z], a, b, c)$ of $\ov Q_\triangle$ where the interval $[z]$ has $\tau$ {\em as an endpoint}. Appendix~\ref{sec:forward_reduction:delta_version}
    explains the technical challenges and how to address them.
\end{example}

\subsection{Lower bound for $\ivmpmd[Q]$}
\label{sec:backword_reduction:body}
\change{In order to prove the lower bound from Theorem~\ref{thm:fully_dynamic_lower_bound},
we prove a stronger lower bound where $|\calD^{(\tau)}|$ is replaced by the number of single-tuple updates $N$.
The following example demonstrates our lower bound for $\ivmpmd[Q]$,} which is based on
a reduction from the static evaluation of a component $\wh Q_{\bm\sigma}$ of the multivariate extension of $Q$ to $\ivmpmd[Q]$.
\shortorfull{}{Algorithm~\ref{alg:backward_reduction} in Appendix~\ref{sec:backward_reduction}
details the reduction.}

\begin{example}[for Theorem~\ref{thm:fully_dynamic_lower_bound}]
    \label{ex:backward:reduction}
    Consider $Q_\triangle$ from Eq.~\eqref{eq:triangle} and $\wh Q_{123}$ from Eq.~\eqref{eq:triangle-query:multi-123:body}.
    Suppose that there exists $\gamma > 0$ such that $\ivmpmd[Q_\triangle]$ can be solved with amortized
    update time \change{$\tildeO{|\calD^{(\tau)}|^\kappa}$ where $|\calD^{(\tau)}|$ is the database size, $\kappa \defeq \omega(\wh Q_{123}) - 1 - \gamma$,
    and $\omega$ is given by Definition~\ref{defn:fw_lb}.
    (Note that $\tildeO{|\calD^{(\tau)}|^\kappa} = \tildeO{N^\kappa}$ since $|\calD^{(\tau)}| \leq N$.)}
    We show that this implies that $\eval[\wh Q_{123}]$ can be solved in time $\tildeO{|\wh \calD_{123}|^{\omega(\wh Q_{123})-\gamma} + |\wh Q_{123}(\wh \calD_{123})|}$ on any database
    instance $\wh \calD_{123}$, thus contradicting the definition of $\omega(\wh Q_{123})$.

    Let us take an arbitrary database instance $\wh\calD_{123} = (\wh R_{123}, \wh S_{123}, \wh T_{123})$ and evaluate $\wh Q_{123}$ over $\wh\calD_{123}$ using the oracle that solves $\ivmpmd[Q_\triangle]$.
    Example~\ref{ex:IJ:backward:reduction} shows how to reduce $\eval[\wh Q_{123}]$ over 
    $\wh\calD_{123}$ to the static evaluation of $\ov Q_\triangle$ from Eq.~\eqref{eq:triangle-query:uni} over a database instance $\ov \calD = (\ov R, \ov S, \ov T)$ that satisfies:
    $|\ov\calD| = |\wh \calD_{123}|$ and
    $|\ov Q_{\triangle}(\ov\calD)| = |\wh Q_{123}(\wh \calD_{123})|$.
    Moreover, $\ov\calD$ can be constructed in linear time\shortorfull{}{ (Theorem~\ref{thm:IJ:backward:reduction})}.
    Next, we reduce the evaluation of $\ov Q_\triangle$
    over $\ov\calD$ into $\ivmpmd[Q_\triangle]$.

    We interpret the interval $[Z]$ of a tuple in $\ov Q_\triangle$
    as the ``lifespan'' of a corresponding tuple in $Q_\triangle$, similar to Section~\ref{sec:forward_reduction:body}.
    \change{In particular, we initialize empty relations
    $R, S, T$ to be used as input for $Q_\triangle$, which will be maintained under a sequence of updates for the $\ivmpmd[Q_\triangle]$ problem.}
    We also sort the tuples $([\alpha, \beta], a, b)$ of $\ov R$ based on the end points $\alpha$ and $\beta$ of their $[Z]$-intervals. 
    Each tuple $([\alpha, \beta], a, b)$ appears twice
    in the order: once paired with its beginning $\alpha$ and another with its end $\beta$.
    Similarly, we add the tuples $([\alpha, \beta], b, c)$ of $\ov S$ and
    $([\alpha, \beta], a, c)$ of $\ov T$ to the same sorted list, where each tuple appears twice. Now, we go through the list
    in order. If the next tuple $([\alpha, \beta], a, b)$ is paired with its beginning $\alpha$, then we insert the tuple $(a, b)$ into $R$.
    However, if the next tuple $([\alpha, \beta], a, b)$ is paired with its end $\beta$, then we delete $(a, b)$ from $R$. The same goes for $S$ and $T$.
    The number of updates\change{, $N$,} is $2|\ov\calD|$,
    and the total time needed to process them is $\tildeO{|\ov\calD|^{\kappa + 1}}$.

    To compute the output of $\ov Q_\triangle$, we do the following. After every insert or delete in $Q_\triangle$, we use the oracle for $\ivmpmd[Q_\triangle]$ to enumerate the {\em change} in the output of $Q_\triangle$. Whenever the oracle reports an insert of an output tuple $(a, b, c)$ at time $\tau$
    and the delete of the same tuple at a later time $\tau'>\tau$, then we know that
    $([\tau, \tau'], a, b, c)$ is an output tuple of $\ov Q_\triangle$.
    Therefore, the overall time needed to compute the output of $\ov Q_\triangle$ (including the total update time from before) is $\tildeO{|\ov\calD|^{\kappa + 1} + |\ov Q_\triangle(\ov\calD)|}= $ $\tildeO{|\wh\calD_{123}|^{\kappa + 1} + |\wh Q_{123}(\wh\calD_{123})|}$.
    This contradicts the definition of $\omega(\wh Q_{123})$.

\shortorfull{}{Appendix~\ref{sec:backward_reduction} highlights the challenges in extending the above lower bound from $\ivmpmd$ to $\ivmpm$, which remains an open problem.}
\end{example}

\section{Conclusion}

\change{This paper puts forward two-way reductions between the dynamic and static query evaluation problems that allow us to transfer a wide class of algorithms and lower bounds between the two problems.}
For the dynamic problem, the paper characterizes the complexity gap between the insert-only and the insert-delete settings. The proposed algorithms recover best known \change{amortized update times} and produce new ones.
\change{The matching lower bounds justify why our upper bounds are natural.}
\change{The lower bound function $\omega(Q)$ that we define can be used as a general tool to assess the optimality of some algorithms in database theory.}

The assumptions that the queries have all variables free and the input database be initially empty can be lifted without difficulty. To support arbitrary conjunctive queries, our algorithms need to consider free-connex tree decompositions, which ensure constant-delay enumeration of the tuples in the query result. To support non-empty initial databases, we can use a preprocessing phase to construct our data structure, where all initial tuples get the same starting timestamp. The complexities stated in the paper need to be changed to account for the size of the initial database as part of the size of the update sequence.

\change{
Although the proposed dynamic algorithms \mvivm are designed to maintain base relations and the query output as sets, they can be extended to maintain bags. As in DBToaster~\cite{DBLP:journals/vldb/KochAKNNLS14} and F-IVM~\cite{FIVM:VLDBJ:2023}, \mvivm can be extended to maintain the multiplicity of each tuple in a base relation or view and enumerate each tuple in the query output together with its multiplicity.
}

\begin{acks}
    This work was partially supported by NSF-BSF 2109922, NSF-IIS 2314527, NSF-SHF 2312195,
    and UZH Global Strategy and Partnerships Funding Scheme,
    and was conducted while some of the authors participated in the Simons Program on Logic and Algorithms in Databases and AI.
\end{acks}

\bibliographystyle{ACM-Reference-Format}
\bibliography{bibliography}

\pagebreak

\appendix
\section{Comparison to Prior IVM Approaches}
\label{app:comparison}

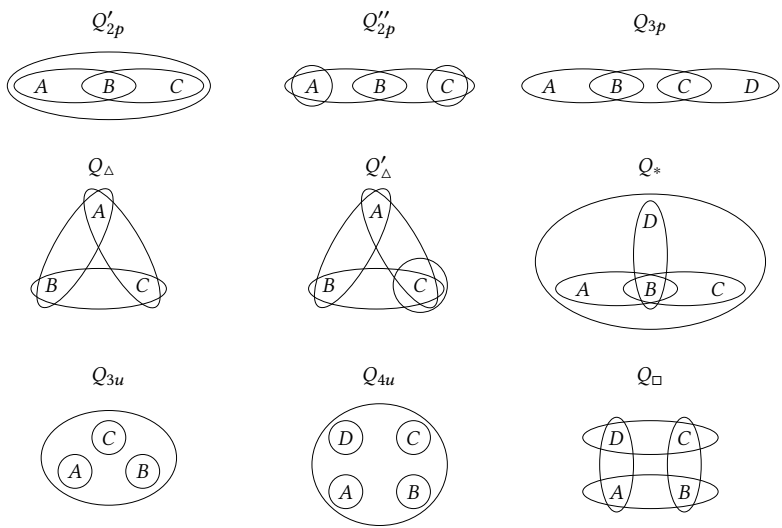
\begin{figure}[ht]
\begin{center}
 \scalebox{0.9}{  
    \begin{tikzpicture}

    \node at (1, 0.9) (A) {$Q_{2p}'$};
    \node at (0, 0) (A) {$A$};
    \node at (1, 0) (B) {$B$};
    \node at (2, 0) (B) {$C$};
    \draw (0.5,0) ellipse (0.9cm and 0.25cm);
    \draw (1.5,0) ellipse (0.9cm and 0.25cm);    
    \draw (1,0) ellipse (1.5cm and 0.5cm);

    \begin{scope}[xshift = 4cm]
        \node at (1, 0.9) (A) {$Q_{2p}''$};
        \node at (0, 0) (A) {$A$};
        \node at (1, 0) (B) {$B$};
        \node at (2, 0) (B) {$C$};
        \draw (0.5,0) ellipse (0.9cm and 0.25cm);
        \draw (1.5,0) ellipse (0.9cm and 0.25cm);    
        \draw (0,0) ellipse (0.3cm and 0.3cm);
        \draw (2,0) ellipse (0.3cm and 0.3cm);
    \end{scope}

    \begin{scope}[xshift = 7.5cm]
        \node at (1.5, 0.9) (A) {$Q_{3p}$};
        \node at (0, 0) (A) {$A$};
        \node at (1, 0) (B) {$B$};
        \node at (2, 0) (B) {$C$};
        \node at (3, 0) (B) {$D$};
        \draw (0.5,0) ellipse (0.9cm and 0.25cm);
        \draw (1.5,0) ellipse (0.9cm and 0.25cm);    
        \draw (2.5,0) ellipse (0.9cm and 0.25cm);
    \end{scope}

    \begin{scope}[xshift=0.5cm, yshift = -2.4cm]
        \node at (0.4, 1.2) (A) {$Q_{\triangle}$};
        \draw[rotate=60] (0,0) ellipse (1cm and 0.3cm);
        \draw[rotate=120] (-0.35,-0.6) ellipse (1cm and 0.3cm);
        \node at (0.35, 0.55) (A) {$A$};
        \node at (-0.35, -0.55) (B) {$B$};
        \node at (1, -0.55) (B) {$C$};
        \draw (0.35,-0.6) ellipse (1cm and 0.3cm);
    \end{scope}

    \begin{scope}[xshift=4.6cm, yshift = -2.4cm]
        \node at (0.4, 1.2) (A) {$Q_{\triangle}'$};
        \draw[rotate=60] (0,0) ellipse (1cm and 0.3cm);
        \draw[rotate=120] (-0.35,-0.6) ellipse (1cm and 0.3cm);
        \node at (0.35, 0.55) (A) {$A$};
        \node at (-0.35, -0.55) (B) {$B$};
        \node at (1, -0.55) (B) {$C$};
        \draw (0.35,-0.6) ellipse (1cm and 0.3cm);

        \draw (1, -0.55) ellipse (0.4cm and 0.4cm);   
    \end{scope}

    \begin{scope}[xshift=8cm, yshift = -3cm]
        \node at (1, 1.8) (A) {$Q_{\ast}$};
        \node at (0, 0) (A) {$A$};
        \node at (1, 0) (B) {$B$};
        \node at (2, 0) (B) {$C$};
        \node at (1, 1) (D) {$D$};    
        \draw (0.5,0) ellipse (0.9cm and 0.25cm);
        \draw (1.5,0) ellipse (0.9cm and 0.25cm);    
        \draw (1,0.5) ellipse (0.25cm and 0.8cm);
        \draw (1,0.4) ellipse (1.7cm and 1cm);
    \end{scope}
    
    \begin{scope}[xshift=0cm, yshift = -5.7cm]
        \node at (1, 1.4) (A) {$Q_{3u}$};
        \node at (0.5, 0) (A) {$A$};
        \node at (1.5, 0) (B) {$B$};
        \node at (1, 0.5) (C) {$C$};    
        \draw (0.5,0) ellipse (0.25cm and 0.25cm);
        \draw (1.5,0) ellipse (0.25cm and 0.25cm);    
        \draw (1,0.5) ellipse (0.25cm and 0.25cm);
        \draw (1,0.2) ellipse (1cm and 0.7cm);
    \end{scope}

    \begin{scope}[xshift=4.0cm, yshift = -6cm]
        \node at (1, 1.7) (A) {$Q_{4u}$};
        \node at (0.5, 0) (A) {$A$};
        \node at (1.5, 0) (B) {$B$};
        \node at (1.5, 0.8) (C) {$C$};    
        \node at (0.5, 0.8) (D) {$D$};        
        \draw (0.5,0) ellipse (0.25cm and 0.25cm);
        \draw (1.5,0) ellipse (0.25cm and 0.25cm);    
        \draw (1.5, 0.8) ellipse (0.25cm and 0.25cm);
        \draw (0.5, 0.8) ellipse (0.25cm and 0.25cm);
        \draw (1,0.4) ellipse (1cm and 0.9cm);
    \end{scope} 

    \begin{scope}[xshift=8cm, yshift = -6cm]
        \node at (1, 1.7) (A) {$Q_{\Box}$};
        \node at (0.5, 0) (A) {$A$};
        \node at (1.5, 0) (B) {$B$};
        \node at (1.5, 0.8) (C) {$C$};    
        \node at (0.5, 0.8) (D) {$D$};        
        
        \draw (1,0) ellipse (1cm and 0.25cm);
        \draw (1,0.8) ellipse (1cm and 0.25cm);
        \draw (1.5, 0.4) ellipse (0.25cm and 0.7cm);
        \draw (0.5, 0.4) ellipse (0.25cm and 0.7cm);
    \end{scope}     
    \end{tikzpicture}
 }   
\end{center}  
\caption{Hypergraphs of the queries analyzed in this section. Atoms are represented by hyperedges. }
    \label{fig:hypergraphs}
\end{figure}

In this section, we compare our \mvivm approach with existing IVM approaches. We later focus on two specific queries, namely $Q_{3p}$ and $Q_\triangle$, and explain in detail how they are maintained by all considered approaches. 
We focus on the problem 
$\ivmpm$, where the task is to process single-tuple inserts {\em and} deletes and enumerate the 
{\em full} query result upon each enumeration request.
We consider the following existing approaches: 
\begin{itemize}
\item \naive: This approach recomputes the query result from scratch after each update.

\item \deltac~\cite{ChirkovaY12}: This is a {\em first-order IVM} approach that materializes the query result and computes the delta or the change to the query result for each update. It updates the query result with the delta result. 

\item \fivm~\cite{DBLP:conf/sigmod/NikolicZ0O20, FIVM:VLDBJ:2023}: 
This is a {\em higher-order IVM} approach that speeds up the delta computation using a view tree whose structure is modelled on a variable order.
Its update times are shown to be at least as good as the prior higher-order IVM approach \textsf{DBToaster}~\cite{DBLP:journals/vldb/KochAKNNLS14}. 
\nop{
~\cite{DBLP:journals/vldb/KochAKNNLS14, DBLP:conf/sigmod/NikolicZ0O20, FIVM:VLDBJ:2023} 
In the following, we consider \fivm
 \fivm~\cite{DBLP:conf/sigmod/NikolicZ0O20, FIVM:VLDBJ:2023}, 
 a recently introduced higher order IVM approach that outperforms in practice prior approaches like  DBToaster~\cite{DBLP:journals/vldb/KochAKNNLS14}. 
}

\item \ivmeps~\cite{DBLP:journals/tods/KaraNNOZ20, DBLP:conf/icdt/KaraNNOZ19}:
This is an {\em  adaptive IVM} methodology that uses delta computation and materialized views and also takes the degrees of data values into account. 
It is adaptive in the sense that  it treats updates referring to values with high degree differently from those with low degree. The specific
attributes whose values are partitioned with regard to their 
degrees and the partitioning threshold that lead to the best update time can vary depending on the query structure. 
So far, there is no algorithm that produces the optimal partitioning strategy for arbitrary queries. Therefore, \ivmeps is not a concrete algorithm for arbitrary queries but rather a methodology.
It was shown in~\cite{DBLP:journals/tods/KaraNNOZ20, DBLP:conf/icdt/KaraNNOZ19} how to manually tailor this methodology to the triangle query and several non-hierarchical queries to achieve the worst-case optimal update time. 

\item \crown ({\em {\bf C}hange p{\bf RO}pagation {\bf W}ithout joi{\bf N}s})~\cite{DBLP:journals/pvldb/WangHDY23}:
This is a  higher-order IVM approach tailored to free-connex ($\alpha$-)acyclic queries. Its maintenance approach is in line with F-IVM, when applied to free-connex acyclic queries, yet it provides a more refined complexity analysis that takes the structure of the update sequence into account. For certain update sequences, such as {\em insert-only} or 
{\em first-in-first-out}, the amortized update time is constant. 
For the purpose of comparing against the other IVM approaches, which do not consider such a fine-grained analysis, we consider the complexity results of \crown for arbitrary update sequences. 
\end{itemize}

Another higher-order IVM approach is 
\textsf{DynYannakakis}~\cite{DBLP:conf/sigmod/IdrisUV17,DBLP:journals/vldb/IdrisUVVL20}, which can maintain acyclic queries. It achieves the same asymptotic update times as \fivm and \crown for acyclic queries.

We consider the following non-hierarchical  join queries 
(all variables are free). They are visualized in Figure~\ref{fig:hypergraphs}: 
\begin{align}
%
Q_{2p}' & = R(A,B,C) \wedge S(A,B) \wedge T(B,C) \label{eq:ex-queries-begin}\\
Q_{3p} &  = R(A,B) \wedge S(B,C) \wedge T(C,D) \\
%
%
%
Q_{\triangle} & = R(A,B)\wedge S(B,C)\wedge T(A,C) \\
Q_{\triangle}' & = R(A,B)\wedge S(B,C)\wedge T(A,C)\wedge U(C) \\
%
%
%
Q_{\ast} &  = R(A,B,C,D)\wedge S(A,B)\wedge T(B,C)\wedge U(B,D) \\
 Q_{3u} &  = R(A,B,C)\wedge S(A)\wedge T(B)\wedge U(C) \\
Q_{4u} &  = R(A,B,C, D)\wedge S(A)\wedge T(B)\wedge U(C)\wedge V(D) \\
Q_{2p}'' & =  R(A) \wedge S(A,B) \wedge T(B,C) \wedge U(C) \\
Q_{\Box} & =  R(A,B) \wedge S(B,C) \wedge T(C,D) \wedge U(A,D) \label{eq:ex-queries-end}
\nop{
Q_{4p} & = R(A,B,C)\wedge S(A,B)\wedge T(B,C)\wedge V(C,D,E)\wedge L(C,D)\wedge M(D,E) & \text{\ (4-path + two ternary atoms)} \\
Q_{2\ast} &  = R(A,B) \wedge S(B,C) \wedge T(C,D) \wedge U(B,E) \wedge V(C,F)
 & \text{\ (edge + stars at ends)} \\
 Q_{2\ast} & = R(A,B,C,D)\wedge S(A)\wedge T(B)\wedge U(C)\wedge V(D) & \text{\ (four unary atoms enclosed by an atom)} \\ 
}
\end{align}

Figure~\ref{fig:comparison_prior_our} summarizes the update times of all aforementioned IVM approaches\change{, give in data complexity}. 
\change{The update times of 
\naive, \deltac, and \fivm are worst-case and  
those of \crown, \ivmeps, and \mvivm are amortized.
The update times of \mvivm have an extra 
$\polylog(|\calD|)$ factor, which is  hidden
by the $\widetilde \calO$-notation.
}
Since \crown is designed for 
free-connex acyclic queries, we skip its complexity for the acyclic ones, which are marked ``NA''.
It is not clear how to manually tailor the \ivmeps approach to all queries, thus some entries are marked ``Open''.
\change{In terms of query complexity, $\mvivm$ has an extra factor of $k!$, where $k$ is the number of input relations, due to the use of \multivariate extensions (Definiton~\ref{defn:multivariate}).}

\begin{figure}[ht]
\begin{center}
\begin{tikzpicture}
\node at(0,0) {    
\begin{tabular}{l|cccccc}
& \naive & \deltac & \fivm & \crown & \ivmeps  & \mvivm\\
&  &  & & & (manual) & \\
\hline
$Q_{2p}'$ & $\bigO{|\calD|}$ & $\bigO{|\calD|}$ & $\bigO{|\calD|}$ & $\bigO{|\calD|}$ & {\em Open} &
 $\tildeO{|\calD|^{1/2}}$\\[0.4em]
$Q_{3p}$ & $\bigO{|\calD|}$ & $\bigO{|\calD|}$ & $\bigO{|\calD|}$ & $\bigO{|\calD|}$ & $\bigO{|\calD|^{1/2}}$ &
 $\tildeO{|\calD|^{1/2}}$\\[0.4em]
$Q_{\triangle}$ & $\bigO{|\calD|^{3/2}}$ & $\bigO{|\calD|}$ & $\bigO{|\calD|}$ & NA & $\bigO{|\calD|^{1/2}}$ &  $\tildeO{|\calD|^{1/2}}$ \\[0.4em]
$Q_{\triangle}'$ & $\bigO{|\calD|^{3/2}}$ & $\bigO{|\calD|}$ & $\bigO{|\calD|}$ & NA & $\bigO{|\calD|^{2/3}}$ &  $\tildeO{|\calD|^{2/3}}$ \\[0.4em]
$Q_{\ast}$ & $\bigO{|\calD|}$ & $\bigO{|\calD|}$ & $\bigO{|\calD|}$ &
 $\bigO{|\calD|}$ & {\em Open} &  $\tildeO{|\calD|^{2/3}}$ \\[0.4em]
$Q_{3u}$ & $\bigO{|\calD|}$ & $\bigO{|\calD|}$ & $\bigO{|\calD|}$ &
$\bigO{|\calD|}$ & {\em Open} &  $\tildeO{|\calD|^{2/3}}$ \\[0.4em]
$Q_{4u}$ & $\bigO{|\calD|}$ & $\bigO{|\calD|}$ & $\bigO{|\calD|}$ &
$\bigO{|\calD|}$ & {\em Open} &  $\tildeO{|\calD|^{3/4}}$ \\[0.4em]
$Q_{2p}''$ & $\bigO{|\calD|}$ & $\bigO{|\calD|}$ & $\bigO{|\calD|}$ & $\bigO{|\calD|}$
 & $\bigO{|\calD|^{1/2}}$ &
 $\tildeO{|\calD|}$\\[0.4em]
 $Q_{\Box}$ & $\bigO{|\calD|^{2}}$ & $\bigO{|\calD|}$ & $\bigO{|\calD|}$ & NA & $\bigO{|\calD|^{2/3}}$ & 
 $\tildeO{|\calD|}$
\end{tabular}
};
\change{
\coordinate (A) at (-4.2, -3.0);
    \coordinate (B) at (0, -3.0);
\draw[decorate,decoration={brace,amplitude=10pt}]
        (B) -- (A) node[midway, yshift=-0.55cm] {worst-case};

\coordinate (C) at (0.5, -3.0);
    \coordinate (D) at (5.2, -3.0);
\draw[decorate,decoration={brace,amplitude=10pt}]
        (D) -- (C) node[midway, yshift=-0.55cm] {amortized};
}
\end{tikzpicture}
\end{center}
\caption{Update times of different IVM approaches for the problem $\ivmpm$\change{, given in data complexity}. 
\change{The update times of the approaches 
\naive, \deltac, and \fivm are worst-case, while 
those of \crown, \ivmeps, and \mvivm are amortized.
$\widetilde \calO$ hides a 
$\polylog(|\calD|)$ factor.}
Queries are depicted in Figure~\ref{fig:hypergraphs} and are given by Eq.~\eqref{eq:ex-queries-begin}-\eqref{eq:ex-queries-end}. \mvivm is our approach, introduced in this paper.
\crown is not applicable to cyclic queries, thus some entries are marked ``NA''. \ivmeps is not a concrete algorithm but rather a general methodology. It is not clear
how to manually set it up for every query, thus some entries are marked ``Open''.
\change{In query complexity, $\mvivm$ has a query-dependant factor of $k!$, where $k$ is the number of input relations.}}
\label{fig:comparison_prior_our}
\end{figure}

We observe that, besides the last two queries, \mvivm is on par or outperforms the 
existing approaches. For the last two queries, it performs only worse than \ivmeps.
However, in contrast to \ivmeps, \mvivm is a systematic algorithm that works for all
queries, whereas \ivmeps is a general methodology where a heavy-light partitioning strategy needs to be manually tailored to a specific query. Such a strategy defines (i) on which tuples of variables in each relation to partition into heavy and light, and (ii) what is the heavy-light threshold for each relation partition. There are infinitely many possible partitioning thresholds and hence partitioning strategies. It is not possible to try them all out in order to find the best one.
Therefore, so far there does not exist an algorithm that produces an optimal \ivmeps strategy for any query.\footnote{At a technical level, the reason why \mvivm has a worse update time for a query like $Q_{\Box}$ than \ivmeps is because \mvivm is limited to the fractional hypertree width, while
for $Q_{\Box}$, it is known that the {\em submodular width} is strictly smaller than the fractional hypertree width~\cite{marx:subw,panda:pods17,2024arXiv240202001A}.
\mvivm does not immediately generalize to the submodular width because the query decomposition lemma (Lemma~\ref{lmm:query-decom-lemma}) does not.
Known algorithms that meet the submodular width~\cite{panda:pods17,2024arXiv240202001A}
are limited to the static setting and do not seem to generalize to the dynamic setting.
We leave this generalization as an open problem.}

In the following, we illustrate how the aforementioned IVM approaches
maintain the 3-path query $Q_{3p}$ and the triangle query $Q_\triangle$. 
\paragraph{Convention}
To simplify the following presentation, we denote a single tuple update to some relation $R(A,B)$ by 
$\delta R(a,b)$, which expresses that the tuple $(a,b)$ is inserted to or deleted from $R$. 
Whenever it simplifies the explanation of an update procedure, we do not distinguish between an insert and a delete and describe the main computation steps common to both cases. 

\nop{
For any materialized relation or view $V(\bm X)$ with schema $\bm X$, we assume that we have indices that allow us to perform the following operations: 
(1) enumerate the tuples in $V$ with constant delay; 
(2) check if $V$ contains a given tuple in constant time; 
(3) insert or delete a tuple in constant time.
Similarly, given $\bm Y \subseteq \bm X$ and a tuple $\bm y$ over $\bm Y$, we assume that we can:
(1) enumerate the tuples in $\pi_{\bm X \setminus \bm Y} \sigma_{\bm Y = \bm y} V$ with constant delay; 
(2) check if $V$ contains a tuple $\bm x$ with $\pi_{\bm Y}\bm x = \bm y$ in constant time.
}

\subsection{\naive}
After each update, \naive computes from scratch a data structure that allows for
the constant-delay enumeration of the query result. The update time is given by the time 
to compute such a data structure. 
Using the Yannakakis algorithm \cite{Yannakakis81}, we can compute in
$\bigO{|\calD|}$ time 
such a data structure for $Q_{3p}$.
For $Q_{\triangle}$, the best known strategy to achieve constant-delay enumeration 
is to compute the query result using a worst-case optimal join algorithm, which takes 
$\bigO{|\calD|^{3/2}}$ time~\cite{Ngo:JACM:18}.

\subsection{\deltac}
\deltac maintains the materialized query result and computes for each update the 
corresponding delta query.
It updates the query result by iterating over the tuples in the result of the delta query and adding them to or deleting them from query result (depending on whether 
the update is an insert or a delete).

\paragraph{3-Path Query}
Given a single-tuple update $\delta R(a,b)$ to relation $R$, the delta of
$Q_{3p}$ is defined as:
\begin{align*}
\delta Q_{3p}(a,b,C,D)&  = \delta R(a,b) \wedge S(b,C) \wedge T(C,D),
\end{align*}
where the lower case letters $a$ and $b$ signalize that the variables $A$ and $B$ are fixed to the values $a$ and respectively $b$.
Using the Yannakakis algorithm \cite{Yannakakis81}, we can enumerate the 
the result of the delta query with constant delay after linear preprocessing time.  
The deltas for updates to the other relations 
can be computed analogously. 
We conclude that the overall update time is $\bigO{|\calD|}$.

\paragraph{Triangle Query}
Given a single-tuple update $\delta R(a,b)$, 
the delta of $Q_{\triangle}$ is defined as:
\begin{align*}
\delta Q_{\triangle}(a,b,C) = \delta R(a,b) \wedge S(b,C) \wedge T(a, C)
\end{align*}
To compute this delta query, we need to intersect the list of $C$-values paired with $b$ in $S$ and the list of $C$-values paired with  $a$ in $T$. Since both lists can be of size $|\calD|$, the time to compute the intersection is $\bigO{|\calD|}$.
Updates to  the other relations are handled analogously. 
We conclude that the update time is $\bigO{|\calD|}$.

\subsection{\fivm}
\label{app:higher_order_IVM}
\fivm uses delta computation and maintains materialized views
besides the query result.  
The set of views is usually organized as a {\em view tree}. 
Given an update, it computes the delta of each view affected by the update. 
\nop{
F-IVM~\cite{DBLP:conf/sigmod/NikolicZ0O20, FIVM:VLDBJ:2023} maintains one view tree for each of the two queries $Q_p$ and $Q_{triangle}$.
The view tree allows to process updates to one relation in constant time.
Updates to the other relations can still require linear time. 
DBToaster~\cite{DBLP:journals/vldb/KochAKNNLS14} maintains for each of the two queries three view trees, one for each input relation. In the following, we illustrate one view tree for each query. 
}

\paragraph{3-Path Query}
Figure~\ref{fig:higher_order_path} (left) shows a view tree 
for the path query $Q_{3p}$. The leaves of the view tree 
are the input relations. 
A view with a single child 
is obtained from its child view by projecting away a variable.
A view with more that one child is obtained by joining the child
views.

\begin{figure}[ht]
\begin{minipage}{13.5cm}
\begin{center}
 \scalebox{0.95}{
    \begin{tikzpicture}
       \node at (0, 0.5) (VRST) {$V_{RST}(C)$};
      \node at (1, -0.5) (VT) {$V_{T}(C)$} edge[-] (VRST);      
      \node at (1, -1.5) (T) {$T(C,D)$} edge[-] (VT);  
      \node at (-1, -0.5) (VRS') {$V_{RS}'(C)$} edge[-] (VRST);      
      \node at (-1, -1.5) (VRS) {$V_{RS}(B,C)$} edge[-] (VRS');      
      \node at (0, -2.5) (S) {$S(B,C)$} edge[-] (VRS);      
      \node at (-2, -2.5) (VR) {$V_R(B)$} edge[-] (VRS);
      \node at (-2, -3.5) (R) {$R(A,B)$} edge[-] (VR);
    
    \begin{scope}[xshift=4.5cm]
       \node at (0, 0.5) (VRST) {\color{red}$\delta V_{RST}(C)$};
      \node at (1, -0.5) (VT) {$V_{T}(C)$} edge[-] (VRST);      
      \node at (1, -1.5) (T) {$T(C,D)$} edge[-] (VT);  
      \node at (-1, -0.5) (VRS') {\color{red}$\delta V_{RS}'(C)$} edge[-] (VRST);      
      \node at (-1, -1.5) (VRS) {\color{red}$\delta V_{RS}(b,C)$} edge[-] (VRS');      
      \node at (0, -2.5) (S) {$S(B,C)$} edge[-] (VRS);      
      \node at (-2, -2.5) (VR) {\color{red}$\delta V_R(b)$} edge[-] (VRS);
      \node at (-2, -3.5) (R) {\color{red}$\delta R(a,b)$} edge[-] (VR);
    \end{scope}

  \begin{scope}[xshift=9cm]
       \node at (0, 0.5) (VRST) {\color{red}$\delta V_{RST}(c)$};
      \node at (1, -0.5) (VT) {$V_{T}(C)$} edge[-] (VRST);      
      \node at (1, -1.5) (T) {$T(C,D)$} edge[-] (VT);  
      \node at (-1, -0.5) (VRS') {\color{red}$\delta V_{RS}'(c)$} edge[-] (VRST);      
      \node at (-1, -1.5) (VRS) {\color{red}$\delta V_{RS}(b,c)$} edge[-] (VRS');      
      \node at (0, -2.5) (S) {\color{red} $\delta S(b,c)$} edge[-] (VRS);      
      \node at (-2, -2.5) (VR) {$V_R(B)$} edge[-] (VRS);
      \node at (-2, -3.5) (R) {$R(A,B)$} edge[-] (VR);
    \end{scope}  
    \end{tikzpicture}
 }   
\end{center}  
  \end{minipage}
\caption{Left: View tree used by \fivm for the 3-path query 
$Q_{3p} = R(A,B)\wedge S(B,C)\wedge T(C,D)$; 
middle: delta view tree under a single-tuple update to relation $R$; 
right: delta view tree under a single-tuple update to relation $S$. 
}
\label{fig:higher_order_path}
\end{figure}
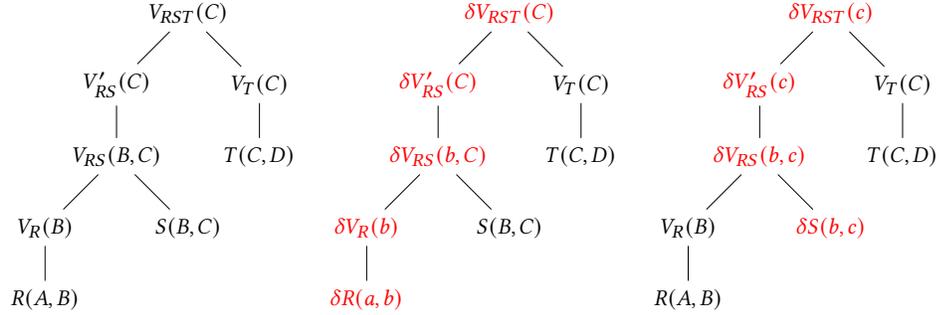

Algorithm~\ref{alg:enumeration_path_query} shows how to enumerate 
from the view tree the result of $Q_{3p}$ with constant delay. 
We use the root view $V_{RST}$ to retrieve the $C$-values that appear in the result tuples and use the view $V_{RS}$ and the relations $R$ and $T$ to fetch the matching $B$-, $A$-, and $D$-values.

\begin{algorithm}[ht]
    \caption{Enumeration of the result of $Q_{3p}$}
    \label{alg:enumeration_path_query}
    \begin{algorithmic}
        \State {\textbf{Input}}
        \begin{itemize}
            \item The view tree for $Q_{3p}$ \Comment{Figure~\ref{fig:higher_order_path} (left)}
        \end{itemize}
        \\\hrulefill
        \State {\textbf{Algorithm}} \\
        For each $C$-value $c \in V_{RST}$ \\
            \TAB for each $B$-value $b$ such that  $(b,c) \in V_{RS}$ \\
            \TAB\TAB for each $A$-value $a$ such that $(a,b) \in R$ \\
            \TAB\TAB\TAB for each $D$-value $d$ such that $(c,d) \in T$ \\
            \TAB\TAB\TAB\TAB \OUTPUT $(a,b,c,d)$
    \end{algorithmic}
\end{algorithm}

Next, we explain how the view tree is maintained under single-tuple updates to relations $R$ and $S$. Updates to $T$ are handled analogously.  
Given a single-tuple update $\delta R(a,b)$ to relation $R$, we need to compute the deltas of all views in the view tree from the leaf $R$ to the root. Figure~\ref{fig:higher_order_path} (middle) shows the {\em delta view tree} where the delta views are highlighted in red. 
We obtain the delta view $\delta V_R$ from $\delta R$ by projecting the tuple $(a,b)$ onto $B$, which requires constant time.
We compute $\delta V_{RS}$ by iterating over all $C$-values paired with $b$ in $S$, which takes at most linear time. 
We obtain $\delta V'_{RS}$ from $\delta V_{RS}$ by projecting away $b$ from all tuples in the latter view, which takes at most linear time. 
The root delta view $\delta V_{RST}$ is obtained
by computing the intersection of the $C$-values 
in $\delta V_{RS}'$ and those in $V_{T}$, which takes at most 
linear time. 

Given a single-tuple update $\delta S(b,c)$ to relation $S$, Figure~\ref{fig:higher_order_path} (right) shows the corresponding 
{\em delta view tree}. 
The computation of the delta view $\delta V_{RS}$ 
requires a constant-time lookup of $b$ in $V_R$.
The delta view $\delta V'_{RS}$ is obtained from 
$\delta V_{RS}$ by projecting the tuple $(b,c)$ onto $C$, which requires constant time. The computation of the root delta view $\delta V_{RST}$ 
requires a constant-time lookup of $c$ in $V_T$. 


We conclude that the view tree in 
Figure~\ref{fig:higher_order_path} (left) can be maintained in 
$\bigO{|\calD|}$ time, which means that the update time is $\bigO{|\calD|}$. 

\paragraph{Triangle Query}
Figure \ref{fig:higher_order_triangle} (left) shows 
a view tree for the triangle query $Q_{\triangle}$.
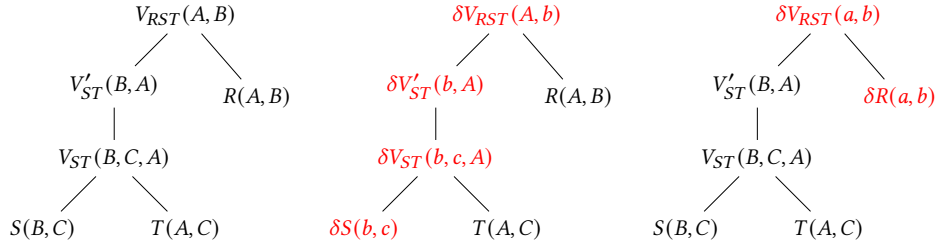
\begin{figure}[ht]
\begin{minipage}{13.5cm}
\begin{center}
\scalebox{0.95}{
    \begin{tikzpicture}
    
       \node at (0, -0.5) (A) {$V_{RST}(A,B)$};
      \node at (-1, -1.5) (C) {$V'_{ST}(B,A)$} edge[-] (A);
      \node at (1, -1.65) (B) {$R(A,B)$} edge[-] (A);      
      \node at (-1, -2.5) (D) {$V_{ST}(B,C,A)$} edge[-] (C);
      \node at (-2, -3.5) (E) {$S(B,C)$} edge[-] (D);
      \node at (0, -3.5) (F) {$T(A,C)$} edge[-] (D);      
    
    \begin{scope}[xshift=4.5cm]
    \node at (0, -0.5) (A) {\color{red} $\delta V_{RST}(A,b)$};
      \node at (-1, -1.5) (C) {\color{red}$\delta V'_{ST}(b,A)$} edge[-] (A);
      \node at (1, -1.65) (B) {$R(A,B)$} edge[-] (A);      
      \node at (-1, -2.5) (D) {\color{red}$\delta V_{ST}(b,c,A)$} edge[-] (C);
      \node at (-2, -3.5) (E) {\color{red}$\delta S(b,c)$} edge[-] (D);
      \node at (0, -3.5) (F) {$T(A,C)$} edge[-] (D);      
    \end{scope}

  \begin{scope}[xshift=9cm]
    \node at (0, -0.5) (A) {\color{red}$\delta V_{RST}(a,b)$};
      \node at (-1, -1.5) (C) {$V'_{ST}(B,A)$} edge[-] (A);
      \node at (1, -1.65) (B) {\color{red}$\delta R(a,b)$} edge[-] (A);      
      \node at (-1, -2.5) (D) {$V_{ST}(B,C,A)$} edge[-] (C);
      \node at (-2, -3.5) (E) {$S(B,C)$} edge[-] (D);
      \node at (0, -3.5) (F) {$T(A,C)$} edge[-] (D);      
    \end{scope}  
    
    \end{tikzpicture}
}    
\end{center}  
  \end{minipage}
\caption{Left: View tree used by \fivm for the triangle query 
$Q_{\triangle}= R(A,B)\wedge S(B,C)\wedge T(A,C)$; 
middle: delta view tree under a single-tuple update to relation $S$; 
right: delta view tree under a single-tuple update to relation $R$.}
\label{fig:higher_order_triangle}
\end{figure}
The view tree can be used as follows to enumerate the query result with constant delay. For each $(A,B)$-tuple $(a,b)$ in $V_{RST}$, we iterate over the $C$-values paired with $(a,b)$ in  $V_{ST}$. For each such $c$, we output the tuple $(a,b,c)$. 

Figure~\ref{fig:higher_order_triangle}  shows the delta view trees for single-tuple updates to $S$ (middle delta view tree) and $R$ (right delta view tree).
The computation of delta views is similar to the case of the 3-path query. Hence, we only give the computation times. 
Given an update to $S$, all delta views on the path from $S$ to the root 
can be computed in linear time.   
Given an update to $R$, the delta view $\delta V_{RST}$ can be computed in constant time. For updates to $T$, the delta views (not shown in 
Figure~\ref{fig:higher_order_triangle}) require at most linear computation time. 

We conclude that \fivm maintains the triangle query with 
$\bigO{|\calD|}$ update time.

\subsection{\ivmeps}
Similar to \fivm, \ivmeps uses  delta computation and materialized views. 
Additionally, it takes degrees of data values into account. 
It partitions relations based on value degrees. 
Since all relations appearing in the two queries $Q_{3p}$ and $Q_{\triangle}$ are binary, we focus here on the partitioning of binary relations.
The partition of a relation $R(A,B)$ on variable $A$ consists of 
 a {\em light} part $R^L$ and a {\em heavy} part $R^H$ defined as follows: $R^{L} = \{(a,b)\in R \mid |\sigma_{A=a}R| < |\calD|^{1/2}\}$ and 
$R^{H} = R \setminus R^{L}$. 
It directly follows from the definition 
of the relation parts: for any $A$-value $a$, it holds $|\sigma_{A=a} R^L| < |\calD|^{1/2}$; $|\pi_{A} R^H| \leq |\calD|^{1/2}$.

\paragraph{3-Path Query}
We partition relation $S$ on variable $B$.
The other relations are not partitioned. 
Given $s \in \{H,L\}$, we denote 
$Q^{s}_{3p} = R(A,B), S^s(B,C), T(C,D)$. The query $Q_{3p}$ can be rewritten as the disjoint union $\bigcup_{s \in \{H,L\}} Q_{3p}^{s}$. 
\ivmeps creates the two view trees visualized in Figure~\ref{fig:adaptive_path} (top row). The first tree represents $Q_{3p}^L$ and the second one represents $Q_{3p}^H$. 

\begin{figure}[ht]
\begin{minipage}{13.5cm}
 \begin{center}   
 \scalebox{0.95}{
    \begin{tikzpicture}
    \begin{scope}
        \node at (0,1) (G) {View tree for $Q_{3p}^{L}$};
       \node at (0, 0.5) (VRST) {$V_{RST}(C)$};
      \node at (1, -0.5) (VT) {$V_{T}(C)$} edge[-] (VRST);      
      \node at (1, -1.5) (T) {$T(C,D)$} edge[-] (VT);  
      \node at (-1, -0.5) (VRS') {$V_{RS}'(C)$} edge[-] (VRST);      
      \node at (-1, -1.5) (VRS) {$V_{RS}(B,C)$} edge[-] (VRS');      
      \node at (0, -2.5) (S) {$S^L(B,C)$} edge[-] (VRS);      
      \node at (-2, -2.5) (VR) {$V_R(B)$} edge[-] (VRS);
      \node at (-2, -3.5) (R) {$R(A,B)$} edge[-] (VR);
    \end{scope}
    
    \begin{scope}[xshift=6cm]
        \node at (0,1) (G) {View tree for $Q_{3p}^{H}$};
       \node at (0, 0.5) (VRST) {$V_{RST}(B)$};
      \node at (1, -0.5) (VR) {$V_{R}(B)$} edge[-] (VRST);      
      \node at (1, -1.5) (R) {$R(A,B)$} edge[-] (VR);  
      \node at (-1, -0.5) (VST') {$V_{ST}'(B)$} edge[-] (VRST);      
      \node at (-1, -1.5) (VST) {$V_{ST}(B,C)$} edge[-] (VST');      
      \node at (-2, -2.5) (S) {$S^H(B,C)$} edge[-] (VST);      
      \node at (0, -2.5) (VT) {$V_T(C)$} edge[-] (VST);
      \node at (0, -3.5) (T) {$T(C,D)$} edge[-] (VT);
    \end{scope}

    \end{tikzpicture}
 }   
 \end{center}   
  \end{minipage}

\vspace{1em}
  \begin{minipage}{13.5cm}
\begin{center}   
\scalebox{0.95}{  
    \begin{tikzpicture}
    \begin{scope}
        \node at (0,1) (G) {Delta view tree for $Q_{3p}^{L}$};
       \node at (0, 0.5) (VRST) {\color{red} $\delta V_{RST}(C)$};
      \node at (1, -0.5) (VT) {$V_{T}(C)$} edge[-] (VRST);      
      \node at (1, -1.5) (T) {$T(C,D)$} edge[-] (VT);  
      \node at (-1, -0.5) (VRS') {\color{red} $\delta V_{RS}'(C)$} edge[-] (VRST);      
      \node at (-1, -1.5) (VRS) {\color{red} $\delta V_{RS}(b,C)$} edge[-] (VRS');      
      \node at (0, -2.5) (S) {$S^L(B,C)$} edge[-] (VRS);      
      \node at (-2, -2.5) (VR) {\color{red} $\delta V_R(b)$} edge[-] (VRS);
      \node at (-2, -3.5) (R) {\color{red} $\delta R(a,b)$} edge[-] (VR);
    \end{scope}
    
    \begin{scope}[xshift=6cm]
        \node at (0,1) (G) {Delta view tree for $Q_{3p}^{H}$};
       \node at (0, 0.5) (VRST) {\color{red}$\delta V_{RST}(B)$};
      \node at (1, -0.5) (VR) {$V_{R}(B)$} edge[-] (VRST);      
      \node at (1, -1.5) (R) {$R(A,B)$} edge[-] (VR);  
      \node at (-1, -0.5) (VST') {\color{red}$\delta V_{ST}'(B)$} edge[-] (VRST);      
      \node at (-1, -1.5) (VST) {\color{red}$\delta V_{ST}(B,c)$} edge[-] (VST');      
      \node at (-2, -2.5) (S) {$S^H(B,C)$} edge[-] (VST);      
      \node at (0, -2.5) (VT) {\color{red}$\delta V_T(c)$} edge[-] (VST);
      \node at (0, -3.5) (T) {\color{red}$\delta T(c,d)$} edge[-] (VT);
    \end{scope}

    \end{tikzpicture}
 }   
 \end{center}   
  \end{minipage}
  \caption{Top row: View trees used by \ivmeps for the 3-path query 
$Q_{3p} = R(A,B)\wedge S(B,C)\wedge T(C,D)$; 
bottom row: delta view trees under a single-tuple update to relation $R$ (left) and relation $T$ (right).}
\label{fig:adaptive_path}
\end{figure}
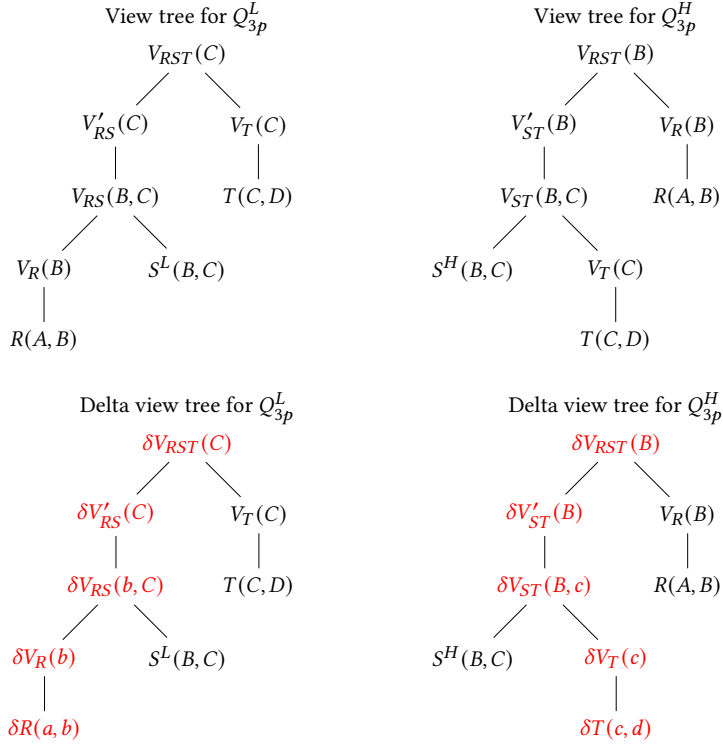

The result tuples of $Q_{3p}^L$ can be enumerated from its view tree
in exactly the same way as from the view tree maintained by \fivm and depicted in Figure~\ref{fig:higher_order_path}.
The enumeration from the view tree of $Q_{3p}^H$ follows the same principle. Upon each enumeration request for $Q_{3p}$, we can first 
enumerate the tuples from the view tree for $Q_{3p}^L$ and then 
enumerate the tuples from the view tree for $Q_{3p}^H$. Since the two queries are disjoint, no tuple will be enumerated twice. 

We analyze the times to maintain the two view trees under updates.
Figure~\ref{fig:adaptive_path} (bottom left) visualizes the delta view tree for $Q_{3p}^L$ under a single-tuple update $\delta R(a,b)$ to $R$.
We compute the delta view $\delta V_R$
in constant time by projecting the tuple $(a,b)$ onto $b$. To compute the delta view $\delta V_{RS}$, we need to iterate over the $C$-values paired with $b$ in $S^L$. Since $S^L$ is light, there are less than $|\calD|^{1/2}$ such $C$-values. Hence, the computation time for $\delta V_{RS}$ is $\bigO{|\calD|^{1/2}}$. The delta view $\delta V'_{RS}$ results from $\delta V_{RS}$ by projecting away from each tuple the 
value $b$, which takes $\bigO{|\calD|^{1/2}}$ time. The computation of 
$\delta V_{RST}$ requires the computation of the intersection of the $C$-values in $\delta V'_{RS}$ and the $C$-values in  $V_T$, which takes $\bigO{|\calD|^{1/2}}$ time. 

Now, we consider the delta view tree for $Q_{3p}^H$ under a single-tuple update $\delta T(c,d)$ to $T$ (bottom right delta view tree in Figure~\ref{fig:adaptive_path}).
The delta view $\delta V_T$
is computed by  projecting the tuple $(c,d)$ onto $c$, which takes constant time. To compute the delta view $\delta V_{ST}$, we iterate over the $B$-values paired with $c$ in $S^H$. Since $S^H$ is heavy, there are at most $|\calD|^{1/2}$ such $B$-values. Thus, the computation time for $\delta V_{ST}$ is $\bigO{|\calD|^{1/2}}$. The remaining delta views are computed within the same time bounds. 

Analogously, we can show that the two view trees in Figure~\ref{fig:adaptive_path}
can be maintained in $\bigO{|\calD|^{1/2}}$ time under updates to the other relations. 
We conclude that \ivmeps maintains the 3-path query with $\bigO{|\calD|^{1/2}}$ update time. 

\paragraph{Triangle Query}
In case of the triangle query, we partition relation $R$ on variable $A$, relation $S$ on $B$, and relation $T$ on $C$. 
Given $r,s,t \in \{H,L\}$, we denote 
$Q_{\triangle}^{rst} = R^r(A,B), S^s(B,C), T^t(A, C)$.
We replace the index $r$ by $*$ if the full relation $R$ participates in the query (and not only its heavy or light part). Same goes for relations $S$ and $T$.
The query $Q_{\triangle}$ can be
rewritten as the disjoint union $\bigcup_{r,s,t \in \{H,L\}} Q^{rst}$. 
\ivmeps creates the three view trees visualized in Figure~\ref{fig:adaptive_triangle}. Observe that each view tree refers at its leaves to two relation parts (e.g., $R^H$ and $S^L$ in the first view tree) and a full relation 
(e.g., $T$ in the first view tree). 
Similar to the case of the 3-path query, we can show that the three view trees can be maintained under single-tuple updates in $\bigO{|\calD|^{1/2}}$ time.

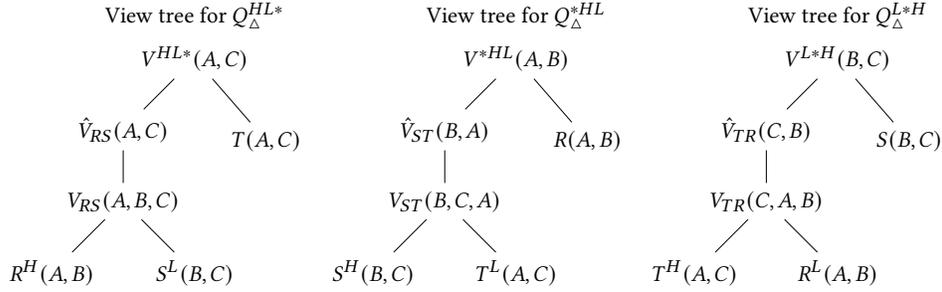
\begin{figure}[ht]
\begin{minipage}{13.5cm}
\scalebox{0.95}{
    \begin{tikzpicture}
    \begin{scope}
     \node at (0,0.1) (G) {View tree for $Q_\triangle^{HL*}$};
      \node at (0, -0.5) (A) {$V^{HL*}(A,C)$};
      \node at (-1, -1.5) (C) {$\hat{V}_{RS}(A,C)$} edge[-] (A);
      \node at (1, -1.65) (B) {$T(A,C)$} edge[-] (A);
      \node at (-1, -2.5) (D) {$V_{RS}(A,B,C)$} edge[-] (C);
      \node at (-2, -3.5) (E) {$R^H(A,B)$} edge[-] (D);
      \node at (0, -3.5) (F) {$S^L(B,C)$} edge[-] (D);            
    \end{scope}

    \begin{scope}[xshift=4.5cm]
      \node at (0,0.1) (G) {View tree for $Q_{\triangle}^{*HL}$};
      \node at (0, -0.5) (A) {$V^{*HL}(A,B)$};
      \node at (-1, -1.5) (C) {$\hat{V}_{ST}(B,A)$} edge[-] (A);
      \node at (1, -1.65) (B) {$R(A,B)$} edge[-] (A);      
      \node at (-1, -2.5) (D) {$V_{ST}(B,C,A)$} edge[-] (C);
      \node at (-2, -3.5) (E) {$S^H(B,C)$} edge[-] (D);
      \node at (0, -3.5) (F) {$T^L(A,C)$} edge[-] (D);      
    \end{scope}

    \begin{scope}[xshift=9cm]
      \node at (0,0.1) (G) {View tree for $Q_\triangle^{L*H}$};
      \node at (0, -0.5) (A) {$V^{L*H}(B,C)$};
      \node at (-1, -1.5) (C) {$\hat{V}_{TR}(C,B)$} edge[-] (A);
      \node at (1, -1.65) (B) {$S(B,C)$} edge[-] (A);
      \node at (-1, -2.5) (D) {$V_{TR}(C,A,B)$} edge[-] (C);
      \node at (-2, -3.5) (E) {$T^H(A,C)$} edge[-] (D);
      \node at (0, -3.5) (F) {$R^L(A,B)$} edge[-] (D);     
    \end{scope}
    \end{tikzpicture}
}    
  \end{minipage}
  \nop{

\vspace{1em}
  \begin{minipage}{13.5cm}
\scalebox{0.95}{ 
    \begin{tikzpicture}
    \begin{scope}
     \node at (0,0.1) (G) {Delta view tree for $Q_\triangle^{HL*}$};
      \node at (0, -0.5) (A) {$\delta V^{HL*}(a,C)$};
      \node at (-1, -1.5) (C) {$\delta \hat{V}_{RS}(a,C)$} edge[-] (A);
      \node at (1, -1.65) (B) {$T(A,C)$} edge[-] (A);
      \node at (-1, -2.5) (D) {$\delta V_{RS}(a,b,C)$} edge[-] (C);
      \node at (-2, -3.5) (E) {$\delta R^H(a,b)$} edge[-] (D);
      \node at (0, -3.5) (F) {$S^L(B,C)$} edge[-] (D);            
    \end{scope}

    \begin{scope}[xshift=4.5cm]
      \node at (0,0.1) (G) { Delta view tree for $Q_\triangle^{*HL}$};
      \node at (0, -0.5) (A) {$\delta V^{*HL}(a,b)$};
      \node at (-1, -1.5) (C) {$\hat{V}_{ST}(B,A)$} edge[-] (A);
      \node at (1, -1.65) (B) {$\delta R(a,b)$} edge[-] (A);      
      \node at (-1, -2.5) (D) {$V_{ST}(B,C,A)$} edge[-] (C);
      \node at (-2, -3.5) (E) {$S^H(B,C)$} edge[-] (D);
      \node at (0, -3.5) (F) {$T^L(A,C)$} edge[-] (D);      
    \end{scope}

    \begin{scope}[xshift=9cm]
      \node at (0,0.1) (G) {Delta view tree for $Q_\triangle^{L*H}$};
      \node at (0, -0.5) (A) {$\delta V^{L*H}(b,C)$};
      \node at (-1, -1.5) (C) {$\delta \hat{V}_{TR}(C,b)$} edge[-] (A);
      \node at (1, -1.65) (B) {$S^s(B,C)$} edge[-] (A);
      \node at (-1, -2.5) (D) {$\delta V_{TR}(C,a,b)$} edge[-] (C);
      \node at (-2, -3.5) (E) {$T^H(A,C)$} edge[-] (D);
      \node at (0, -3.5) (F) {$\delta R^L(a,b)$} edge[-] (D);     
    \end{scope}
    \end{tikzpicture}
}    
  \end{minipage}
 } 
  \caption{View trees used by \ivmeps for the triangle query 
$Q_{\triangle}(A,B,C) = R(A,B)\wedge S(B,C)\wedge T(A,C)$.
Note that $T$ is partitioned into $T^L$ and $T^H$ \underline{based on $C$}, while $R$ and $S$ are partitioned based on $A$ and $B$, respectively.}
\label{fig:adaptive_triangle}
\end{figure}

\paragraph{Rebalancing and Amortization}
Updates can change the degrees of values, which means that a light value can become heavy and vice-versa.
The original \ivmeps algorithm moves tuples between the heavy and the light part of a relation whenever the degrees of the values in the tuples pass certain thresholds~\cite{DBLP:journals/tods/KaraNNOZ20}. This leads to additional time cost that can be amortized over update sequences such that the amortized update time remains $\bigO{|\calD|^{1/2}}$. 
\nop{
Using standard de-amortization techniques, the original 
\ivmeps algorithm can be tweaked so that the update time is $\bigO{|\calD|^{1/2}}$ in the worst case~\cite{DBLP:journals/tods/KaraNNOZ20}. 
}

\subsection{\crown}
\crown is a higher order IVM approach. 
It uses view trees where each view is either a projection of another view or a semi-join
between a view and other views. 
Since \crown is tailored to free-connex acyclic queries, 
we illustrate it for the 3-path query $Q_{3p}$ only.
\nop{
The complexity analysis is more refined in the sense that it takes the structure of the update sequence 
into account and shows that for certain update sequences, such as {\em insert-only} or 
{\em first-in-first-out} the update time is much lower than in the general setting. 
In the following, we illustrate CROWN for the path query $Q_p$ only.
}

\paragraph{3-Path Query} 
The view trees constructed by \crown follow the so-called {\em generalized join trees} which, 
unlike classical join trees, can contain nodes that consist of variables that are a strict subset of the schema of a relation appearing in the query.    
Figure~\ref{fig:gen_join_tree} shows a possible join tree for $Q_{3p}$.
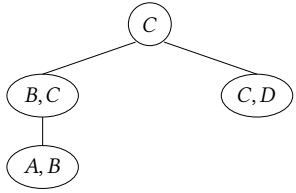
\begin{figure}[ht]
\begin{minipage}{13.5cm}
\begin{center}
 \scalebox{0.95}{ 
    \begin{tikzpicture}

\draw (0,0) ellipse (0.3cm and 0.3cm);
\node at (0,0) {$C$};

\draw (-1.5,-1) ellipse (0.5cm and 0.3cm);
\node at (-1.5,-1) {$B,C$};

\draw (1.5,-1) ellipse (0.5cm and 0.3cm);
\node at (1.5,-1) {$C,D$};

\draw (-1.5,-2) ellipse (0.5cm and 0.3cm);
\node at (-1.5,-2) {$A,B$};

\draw (-0.2,-0.25) -- (-1.5,-0.7);
\draw (0.2,-0.25) -- (1.5,-0.7);

\draw (-1.5,-1.3) -- (-1.5,-1.7);

\end{tikzpicture}
    }
\end{center}  
  \end{minipage}
\caption{A generalized join tree for the 3-path query $Q_{3p} = R(A,B), S(B,C), T(C,D)$.}
\label{fig:gen_join_tree}
\end{figure}
The view tree corresponding to this generalized join tree is the same as the one used by \fivm (Figure~\ref{fig:higher_order_path} left). The views $V_R$, $V'_{RS}$, and $V_T$
are called projection views, since they are obtained by a projection on their child views.  
The views $V_{RS}$ and $V_{RST}$ are called semi-join views, since they
are computed by semi-joining one child view with other child views. 
The enumeration of result tuples is exactly as explained in Section~\ref{app:higher_order_IVM}.

\crown uses special indices that can speed up delta computation by constant factors.
We explain these indices for the view tree in Figure~\ref{fig:higher_order_path} (left).
For each value $b$ in the projection view $V_R$, it  maintains a counter 
$\mycount_R$ that stores the number of tuples in $R$ with $B$-value $b$. 
There are similar counters $\mycount_{RS'}$ and $\mycount_{T}$ for the other projection views $V_{RS'}$ and $V_T$.
Additionally, there is  a counter 
$\mycount_{RS}$ that counts for each $(B,C)$-value $(b,c)$ in $S$, the number of child views that contain matching tuples.
More precisely: 
$$
\mycount_{RS}[b,c] = 
\begin{cases} 
2 & \text{if } b \in V_R \text{ and }  (b,c) \in S \\
0 & \text{if } b \notin V_R \text{ and }  (b,c) \notin S \\
1 & \text{otherwise} 
\end{cases}
$$
Crown maintains also a counter $\mycount_{RST}$ that counts for each $C$-value $c \in V_{RS}' \cup V_T$, 
the number of child views that contain $c$.

We explain how the view tree is maintained under an insert of a tuple $(a,b)$ to relation $R$. Other types of updates are handled analogously. 
If $b$ is included in $V_R$, we increase $\mycount_R[b]$ by $1$. Otherwise, we add $b$ to  $V_R$ and set $\mycount_R[b] = 1$. Then, we update $V_{RS}$ as follows. For each tuple $(b,c')$ in $S$, we increase 
$\mycount_{RS}[b,c']$ by $1$
and check whether $\mycount_{RS}[b,c'] = 2$. If yes, we add the tuple $(b,c')$ to $V_{RS}$.
To propagate the change further up to the root, we use the counters
$\mycount_{RS'}$ and $\mycount_{RST}$ in a similar manner. 
The overall update time is $\bigO{|\calD|}$ for arbitrary update sequences.

\subsection{\mvivm}
Finally, we highlight some of the key insights of our approach, \mvivm, that enables it
to maintain the triangle query $Q_\triangle$ in
\change{amortized time $\tildeO{|\calD|^{1/2}}$}.
Example~\ref{ex:main_fully_dynamic} presents the high-level idea while Appendix~\ref{sec:forward_reduction} explains the technical subtleties that arise.
At a high-level, our approach maintains $Q_\triangle$ by constructing and maintaining the six multivariate extension components $\wh Q_{123},\ldots,\wh Q_{321}$ given by Eq.~\eqref{eq:triangle-query:multi-123:body}
and~\eqref{eq:triangle-query:multi-132-321}.
Each of these component queries is maintained using Theorem~\ref{thm:main_inserts} along with a variety of technical lemmas described in Appendix~\ref{sec:forward_reduction}.
It can be mechanically verified that each of these components has a fractional hypertree width of $\frac{3}{2}$, leading to our amortized update time of $\tildeO{|\calD|^{1/2}}$.
Let's take the first component $\wh Q_{123}$ as an example:
\begin{align*}
    \wh Q_{123} =
        \wh R_{123}(Z_1, A, B) \wedge 
        \wh S_{123}(Z_1, Z_2, B, C) \wedge
        \wh T_{123}(Z_1, Z_2, Z_3, A, C)
\end{align*}
This component has an optimal tree decomposition, depicted in Figure~\ref{fig:triangle-query:multi-123:TD}, that consists of
two bags: 
\begin{itemize}
    \item a root bag containing $\{Z_1, Z_2, A, B, C\}$   ({\em without} $Z_3$), and
    \item a child bag containing $\{Z_1, Z_2, Z_3, A, C\}$ ({\em without} $B$).
\end{itemize}
\begin{figure}[ht]
    \centering
    \begin{tikzpicture}[scale = .4, every node/.style={scale=0.9}]
        \node[draw,ellipse] (root) {$Z_1, Z_2, A, B, C$};
        \node[draw,ellipse, below = 0.7 of root] (child) {$Z_1, Z_2, Z_3, A, C$};
        \draw (root) -- (child);
    \end{tikzpicture}
\caption{An optimal tree decomposition for $\wh Q_{123}$ from Eq.~\eqref{eq:triangle-query:multi-123:body}.}
\label{fig:triangle-query:multi-123:TD}
\end{figure}
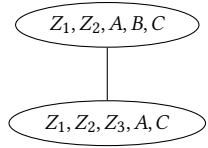
Recall that $Z_1$, $Z_2$, and $Z_3$ are temporal attributes representing the lifespans of tuples in $R$, $S$, and $T$. In particular, the lifespan of a tuple in $T(A, C)$ is represented by the concatenation of $Z_1\circ Z_2\circ Z_3$; see Example~\ref{ex:IJ:forward:reduction:body}.
The above tree decomposition corresponds to a query plan where we are projecting $Z_3$ out of the child bag
before joining with the root bag.
Projecting $Z_3$ out of the child bag corresponds to ``expanding'' the lifespans of tuples of $T(A, C)$ from the lifespan $Z_1\circ Z_2\circ Z_3$ to the {\em longer} lifespan $Z_1\circ Z_2$. What this means algorithmically is that we are {\em lazily} processing deletes from $T$ in order to avoid doing excessive work for each delete, especially in the case where the same tuple from $T$ is inserted and deleted many times. In particular, a sequence of repeated insert-delete of the same tuple of $T$ corresponds to lifespans whose encodings may share the same prefix $Z_1\circ Z_2$. Hence, we don’t need to propagate changes from the child bag to the root bag for every such insert-delete because the lifespan $Z_1\circ Z_2$ remains in the child bag throughout all these insert-deletes. Without this lazy processing, each such repeated insert/delete into $T$ would have required linear update time, thus we would have failed to meet our target update time of $\tildeO{|\calD|^{1/2}}$. This implicit lazy processing is just one example of sophisticated behavior that comes automatically out of our approach.

\section{Missing Details in Preliminaries}
\label{app:prelims}

\paragraph{Query Classes}

We give below some examples of {\em acyclic}, {\em hierarchical}, and {\em Loomis-Whitney queries},
where the latter two are defined in Section~\ref{sec:prelims}.
Recall that for a query $Q$ and a variable $X \in \vars(Q)$, we use $\atoms(X)$ to denote
the set of atoms of $Q$ that contain $X$ in their schema.
\begin{example}[Query Classes]
    \label{ex:queries}
    The query $Q_{\text{h}} = R(A,B) \wedge S(A,C)$ is hierarchical. 
    We have $\atoms(A) = \{R(A,B), S(A,C)\}$, 
    $\atoms(B) = \{R(A,B)\}$, and $\atoms(C) = \{S(A,C)\}$.
    It holds $\atoms(B)\subseteq \atoms(A)$, 
    $\atoms(C)\subseteq \atoms(A)$, and $\atoms(B) \cap \atoms(C) = \emptyset$.
    The query 
    $Q_{\text{nh}} = R(A) \wedge S(A,B) \wedge T(B)$ is acyclic but not hierarchical. Indeed,
    we have $\atoms(A)= \{R(A), S(A,B)\}$ and $\atoms(B)= \{T(B),$ $S(A,B)\}$ and it does not hold $\atoms(A)\subseteq \atoms(B)$ nor $\atoms(B)\subseteq \atoms(A)$ nor $\atoms(A)\cap \atoms(B) = \emptyset$. A possible join tree for 
    $Q_{nh}$ has $R(A,B)$ as root and the other two atoms as the 
    children of the root. The triangle query (or Loomis-Whitney query of degree $3$)
    is of the form $Q_{\triangle} = R(A,B) \wedge S(B,C) \wedge T(A,C)$.
    A Loomis-Whitney query of degree $4$ is of the form 
    $Q_{\text{lw4}} = R(A,B,C) \wedge S(B,C,D) \wedge T(C,D,A) \wedge U(D,A,B)$.
\end{example}

\paragraph{Variable Elimination Orders}
We review the equivalence between tree decompositions and {\em variable elimination orders},
which is a folklore concept in graph and database theory; see e.g.~\cite{faq} and references
thereof.
Given a query $Q$, a {\em variable elimination order} $\bm\mu$ for $Q$ is an order of the variables of $Q$ that we use to {\em eliminate} them one by one.
In order to {\em eliminate} a variable $X$, we take all atoms $R(\bm X)$ where $X \in \bm X$
and replace them with a new atom $R_X^{\bm\mu}(\bm U_{X}^{\bm\mu}-\{X\})$ where $\bm U_{X}^{\bm\mu}$ is the union of schemas $\bm X$ of all these atoms $R(\bm X)$ that contain $X$. It is long known that for a fixed variable elimination order $\bm\mu$, the sets $\bm U_X^{\bm\mu}$ are the bags of some tree decomposition of $Q$.
Moreover, one of these tree decompositions that are obtained from variable elimination orders is {\em optimal}, i.e. has a minimum fractional hypertree width.
Therefore, the fractional hypertree width can be equivalently defined as follows.
Let $\Sigma_{\vars(Q)}$ denote the set of all permutations of the variables of $Q$, and let $\bm\mu \in \Sigma_{\vars(Q)}$ be a variable elimination order for $Q$:
\begin{align}
\fw(\bm\mu) &\defeq \max_{X \in \bm\mu} \rho^*(Q_{\bm U_X^{\bm\mu}})\nonumber\\
\fw(Q) &\defeq \min_{\bm\mu \in \Sigma_{\vars(Q)}} \fw(\bm\mu) \label{eq:fhtw:vo}
\end{align}
In particular, it is known that~\eqref{eq:fhtw:vo} and~\eqref{eq:fhtw} are equivalent; see for example~\cite{faq} and references thereof.
We call $\fw(\bm\mu)$ the {\em fractional hypertree width} of the variable elimination order $\bm\mu$. Moreover, we call $\bm\mu$ {\em optimal} if $\fw(\bm\mu) = \fw(Q)$.

\paragraph{Online Matrix-Vector Multiplication (OMv) Conjecture}
In order to state the OMv conjecture, we first state the OMv problem.
Given a Boolean matrix $\bm M$ and 
$n$ column vectors 
$\bm v_1, \ldots , \bm v_n$ that arrive one-by-one, the task is to output 
the vector $\bm M \bm v_r$ after seeing $\bm v_r$ for each $r \in [n]$:\\
\noindent
\vspace*{.5em}
\fbox{%
    \parbox{0.96\linewidth}{%
    \begin{tabular}{ll}
   Problem: & OMv \\ 
        Given: & $n \times n$ Boolean matrix $\bm M$ and a stream of $n$ $n\times 1$ Boolean column vectors $\bm v_1, \ldots , \bm v_n$  \\
        Task: & Compute $\bm M \bm v_r$ after seeing each $\bm v_r$ \\
    \end{tabular}
    }}


\begin{conjecture}[OMv~\cite{HenzingerKNS15}]
    \label{conj:omv}
    For any $\gamma > 0$, there is no algorithm that solves the OMv problem in $\bigO{n^{3-\gamma}}$ time. 
\end{conjecture}

\paragraph{Online Vector-Matrix-Vector Multiplication (OuMv) Conjecture}
The OuMv problem is defined as follows.
Given a Boolean matrix $\bm M$ and 
$n$ pairs of column vectors 
$(\bm u_1, \bm v_1), \ldots , (\bm u_n, \bm v_n)$ that arrive one-by-one, the task is to output 
the product ${\bm u}^{\top}_r \bm M \bm v_r$ after seeing $(\bm u_r, \bm v_r)$ for each $r \in [n]$:\\
\noindent
\vspace*{.5em}
\fbox{%
    \parbox{0.96\linewidth}{%
    \begin{tabular}{ll}
   Problem: & OuMv \\ 
        Given: & $n \times n$ Boolean matrix $\bm M$ and a stream of $n$ pairs of
        $n\times 1$ Boolean column vectors\\&$(\bm u_1, \bm v_1), \ldots , (\bm u_n, \bm v_n)$  \\
        Task: & Compute $\bm u_r^{\top} \bm M \bm v_r$ after seeing each $(\bm u_r, \bm v_r)$ \\
    \end{tabular}
    }}


\begin{conjecture}[OuMv~\cite{HenzingerKNS15}]
    \label{conj:oumv}
    For any $\gamma > 0$, there is no algorithm that solves the OuMv problem in $\bigO{n^{3-\gamma}}$ time. 
\end{conjecture}

It is known that the OMv conjecture implies the OuMv conjecture~\cite{HenzingerKNS15}.
\section{Missing Details in Section \ref{sec:insert_only}}
\label{app:insert_only}

In this section, we explain in detail our results on IVM in the insert-only setting, where no deletes
are allowed. The main goal is to prove Theorem~\ref{thm:main_inserts}
(as well as \change{Proposition~\ref{prop:insert_only_lower_bound}}).
The results in this section also serve as building blocks for our results on the \insertdelete setting in Section~\ref{app:fully_dynamic}.

\subsection{Preliminaries: The Query Decomposition Lemma}
We prove the query decomposition lemma (Lemma~\ref{lmm:query-decom-lemma})
and illustrate it by an example. This lemma is the core of the proof of Lemma~\ref{lem:insert_only_in_AGM}, which in turn is a key ingredient in the proof of Theorem~\ref{thm:main_inserts}. Recall notation from Section~\ref{app:prelims}.

The query decomposition lemma basically
says the following.
Let $Q$ be a query and let $\bm Y \subseteq \vars(Q)$.
Then, the AGM bound of $Q$ can be decomposed into a sum of AGM bounds of ``residual'' queries: one query $Q\Join \bm y$ for each tuple $\bm y$ over the variables $\bm Y$.
\begin{citedlem}[\ref{lmm:query-decom-lemma} (Query Decomposition Lemma~\cite{SkewStrikesBack2014,WCOJGemsOfPODS2018})]
    Given a query $Q$ and a subset $\bm Y \subseteq \vars(Q)$, let $\left(\lambda_{R(\bm X)}\right)_{R(\bm X)\in\atoms(Q)}$ be a fractional edge cover of $Q$.
    Then, the following inequality holds:

    \begin{equation}
        \sum_{\bm y\in \Dom(\bm Y)}
        \underbrace{\prod_{R(\bm X) \in \atoms(Q)}
        |R(\bm X)\ltimes \bm y|^{\lambda_{R(\bm X)}}
        }_{\text{AGM-bound of $Q\ltimes \bm y$}}
        \leq
        \underbrace{\prod_{R(\bm X) \in \atoms(Q)}
        |R|^{\lambda_{R(\bm X)}}}_{\text{AGM-bound of $Q$}}
        \label{eq:query-decom-lemma:app}
    \end{equation}
\end{citedlem}
In the above, $\bm y\in \Dom(\bm Y)$ indicates that the tuple $\bm y$ has schema $\bm Y$. Moreover,
$R(\bm X)\ltimes \bm y$ denotes the {\em semijoin} of the atom
$R(\bm X)$ with the tuple $\bm y$.
\begin{proof}[Proof of Lemma~\ref{lmm:query-decom-lemma}]
In order to prove inequality~\eqref{eq:query-decom-lemma:app}, we prove the following special case
where $\bm Y$ consists of a single variable $Y$:
\begin{equation}
    \sum_{y \in \Dom(Y)} \prod_{R(\bm X) \in \atoms(Q)}
    |R(\bm X)\ltimes y|^{\lambda_{R(\bm X)}}
    \leq
    \prod_{R(\bm X) \in \atoms(Q)}
    |R|^{\lambda_{R(\bm X)}}
    \label{eq:query-decom-lemma:one-var}
\end{equation}
Inequality~\eqref{eq:query-decom-lemma:app} follows by repeatedly applying~\eqref{eq:query-decom-lemma:one-var} to each variable in $\bm Y$.
To that end, note that for each ``residual'' query, $Q_y \defeq Q\Join y$, the following is a valid fractional edge cover, where $R_y(\bm X)\defeq R(\bm X) \ltimes y$:
\begin{align*}
    \left(\lambda_{R_y(\bm X)}\right)_{R_y(\bm X) \in \atoms(Q_y)}
\end{align*}

Now, we prove~\eqref{eq:query-decom-lemma:one-var}:
\begin{align}
    \sum_{y\in\Dom(Y)} \prod_{R(\bm X) \in \atoms(Q)}
    |R(\bm X)\ltimes y|^{\lambda_{R(\bm X)}} &=\nonumber\\
    \sum_{y\in\Dom(Y)}
    \prod_{\substack{R(\bm X) \in \atoms(Q)\\Y \in \bm X}}
    |\sigma_{Y=y}R(\bm X)|^{\lambda_{R(\bm X)}}
    \prod_{\substack{R(\bm X) \in \atoms(Q)\\Y \notin \bm X}}
    |R|^{\lambda_{R(\bm X)}}&=\nonumber\\
    \prod_{\substack{R(\bm X) \in \atoms(Q)\\Y \notin \bm X}}
    |R|^{\lambda_{R(\bm X)}}
    \sum_{y\in\Dom(Y)}
    \prod_{\substack{R(\bm X) \in \atoms(Q)\\Y \in \bm X}}
    |\sigma_{Y=y}R(\bm X)|^{\lambda_{R(\bm X)}}
    &\leq\label{eq:query-decom:holder}\\
    \prod_{\substack{R(\bm X) \in \atoms(Q)\\Y \notin \bm X}}
    |R|^{\lambda_{R(\bm X)}}
    \prod_{\substack{R(\bm X) \in \atoms(Q)\\Y \in \bm X}}
    \left(\sum_{y\in\Dom(Y)}
    |\sigma_{Y=y}R(\bm X)|\right)^{\lambda_{R(\bm X)}}
    \nonumber&=\\
    \prod_{\substack{R(\bm X) \in \atoms(Q)\\Y \notin \bm X}}
    |R|^{\lambda_{R(\bm X)}}
    \prod_{\substack{R(\bm X) \in \atoms(Q)\\Y \in \bm X}}
    |R|^{\lambda_{R(\bm X)}}\nonumber&=\nonumber\\
    \prod_{R(\bm X) \in \atoms(Q)}
    |R|^{\lambda_{R(\bm X)}}\nonumber
\end{align}
Inequality~\eqref{eq:query-decom:holder} above follows from H\"older's inequality,
which applies specifically because the sum of $\lambda_{R(\bm X)}$ over all atoms
$R(\bm X)$ that contain $Y$ is at least 1.
\end{proof}

\begin{example}[for Lemma~\ref{lmm:query-decom-lemma}]
    Consider the triangle query:
    \begin{equation*}
        Q_\triangle(A, B, C) = R(A, B) \wedge S(B, C) \wedge T(A, C)
    \end{equation*}
    By choosing $\lambda_{R(A, B)} = \lambda_{S(B,C)} = \lambda_{T(A, C)}=\frac{1}{2}$ and $\bm Y =\{A, B\}$,
    inequality~\eqref{eq:query-decom-lemma:app} implies the following:
    \begin{align}
        \sum_{(a, b) \in R} \sqrt{|\sigma_{B=b} S(B, C)| \cdot|\sigma_{A=a} T(A, C)|}
        \leq \sqrt{|R|\cdot|S| \cdot |T|}
        \label{eq:query-decom-lemma:triangle}
    \end{align}
    This inequality can be proved as follows:
    \begin{align*}
        \sqrt{|R|\cdot|S| \cdot |T|} &=\\
        \sqrt{|S|} \cdot\sqrt{\left(\sum_{a}|\sigma_{A = a}R(A, B)|\right)
            \left(\sum_{a}|\sigma_{A = a}T(A, C)|\right)}&\geq\\
        \sqrt{|S|}\cdot\sum_{a}\sqrt{|\sigma_{A = a}R(A, B)|\cdot|\sigma_{A = a}T(A, C)|}&=\\
        \sum_{a}\sqrt{|\sigma_{A = a}R(A, B)|\cdot|\sigma_{A = a}T(A, C)|\cdot|S|}&=\\
        \sum_{a}\sqrt{|\sigma_{A = a}T(A, C)|}\sqrt{
            \left(\sum_{b}|\sigma_{A = a, B = b}R(A, B)|\right)\left(\sum_b{|\sigma_{B=b} S(B, C)|}\right)
        }&\geq\\
        \sum_{a}\sqrt{|\sigma_{A = a}T(A, C)|}\sum_{b}\sqrt{
            |\sigma_{A = a, B = b}R(A, B)|\cdot{|\sigma_{B=b} S(B, C)|}
        }&=\\
        \sum_{a, b}\sqrt{|\sigma_{A = a, B = b}R(A, B)|\cdot|\sigma_{B=b} S(B, C)|\cdot|\sigma_{A = a}T(A, C)|}&=\\
        \sum_{(a, b) \in R} \sqrt{|\sigma_{B=b} S(B, C)| \cdot |\sigma_{A=a} T(A, C)|}
    \end{align*}
    The two inequalities above follow from the Cauchy-Schwarz inequality, which is a special
    case of H\"older's inequality.
    \label{ex:query-decom-lemma:triangle}
\end{example}

\subsection{Proof of Lemma~\ref{lem:insert_only_in_AGM}}
\label{sec:insert_only:ub:agm}
Theorem~\ref{thm:main_inserts} gives an upper bound on the amortized update time
for IVM in the insert-only setting in terms of the fractional hypertree width of the query $Q$.
Before we show how to achieve this bound in Sections~\ref{sec:insert_only:ub:full} and~\ref{sec:insert_only:ub:delta}, we prove in this section a weaker bound that is based 
on the AGM-bound of $Q$.
In particular, we prove Lemma~\ref{lem:insert_only_in_AGM}, which says the following. Given a query $Q$ (whose input relations are initially empty)
and a stream of $N$ single-tuple inserts into $Q$, we can update $Q$ after every insert
{\em and} report the new output tuples $\delta Q$
in a total time for {\em all} updates/reports {\em together} that does not exceed the AGM bound of $Q$.
Note that  here we are referring to the {\em final} AGM bound of $Q$
after all inserts have been performed.

\begin{citedlem}[\ref{lem:insert_only_in_AGM}]
    Given a query $Q$, an initially empty database $\calD^{(0)}$, and a stream of $N$ single-tuple inserts,
    we can compute the {\em new} output tuples, $\delta_\tau Q(D)$, after every insert
    $\delta_\tau \calD$ where the total time over {\em all} inserts combined is $\bigO{N + \agm(Q, \calD^{(N)})}$.
\end{citedlem}

The proof of Lemma~\ref{lem:insert_only_in_AGM} relies on the query decomposition lemma (Lemma~\ref{lmm:query-decom-lemma}~\cite{SkewStrikesBack2014,WCOJGemsOfPODS2018}). Example~\ref{ex:insert-only:agm} demonstrates Lemma~\ref{lem:insert_only_in_AGM} over the triangle query $Q_{\triangle}$.
Lemma~\ref{lem:insert_only_in_AGM} considers only the $\ivmpd[Q]$ problem, but the proof can be straightforwardly extended to $\ivmp[Q]$.

\begin{proof}[Proof of Lemma~\ref{lem:insert_only_in_AGM}]
    For a given relation $R$ in $Q$, we use
    $R^{(0)} \defeq \emptyset,$ $R^{(1)},$ $\ldots,$ $R^{(N)}$
    to denote the sequence of versions of $R$ after each insert.
    Because we only have inserts here, we have
    $0 =|R^{(0)}| \leq |R^{(1)}| \leq \ldots \leq |R^{(N)}| \leq N$.
    Moreover, for each $\tau \in [N]$, let $\delta_\tau R\defeq R^{(\tau)} \setminus R^{(\tau-1)}$. Note that $|\delta_\tau R|$ is either 1 or 0 depending on whether
    the $\tau$-th insert was into the relation $R$ or into another relation of $Q$.
    Let $\left(\lambda_{R(\bm X)}\right)_{R(\bm X)\in\atoms(Q)}$ be a fractional edge cover
    for $Q$. We will show how to perform all $N$ inserts and report $\delta_\tau Q$ after every insert in a total time of $\bigO{\theta}$ (over all inserts/reports combined) where
    \begin{align*}
        \theta\defeq\prod_{R(\bm X)\in\atoms(Q)}|R^{(N)}|^{\lambda_{R(\bm X)}}
    \end{align*}
    In particular, for a fixed atom $S(\bm Y)\in \atoms(Q)$, we will show how to perform all inserts into
    $S$ in a total time of $\bigO{\theta}$.
    For any $\tau \in [N]$ where $\delta_\tau S \neq \emptyset$,
    the insertion of $\delta_\tau S$ into $S$ can be performed in time:
    \begin{equation}
        \prod_{R(\bm X)\in\atoms(Q)}|R^{(\tau)}(\bm X)\ltimes \delta_\tau S(\bm Y)|^{\lambda_{R(\bm X)}}
        \label{eq:lem:insert_only_in_AGM:1}
    \end{equation}
    This can be done by solving the query
    \begin{align*}
        \delta_\tau Q \defeq \bigwedge_{R(\bm X)\in\atoms(Q)}R^{(\tau)}(\bm X)\ltimes \delta_\tau S(\bm Y)
    \end{align*}
    In particular,~\eqref{eq:lem:insert_only_in_AGM:1} is the AGM bound of $\delta_\tau Q$. We can meet this bound if we solve $\delta_\tau Q$ using a {\em worst-case optimal} join algorithm~\cite{Ngo:JACM:18,LeapFrogTrieJoin2014,WCOJGemsOfPODS2018}.
    For that purpose, each relation $R$ needs to be indexed with a {\em trie}~\cite{LeapFrogTrieJoin2014} using an appropriate variable ordering for this query.\footnote{We might need to
    maintain multiple tries (with different variable orderings) for the same relation $R(\bm X)$ in order to handle inserts into different atoms $S(\bm Y)$.}
    The total runtime of all inserts into $S$ is therefore bounded by:~\footnote{In
    order to avoid incurring an extra factor of $\log N$ in the runtime, we can switch from tries to hash tables~\cite{LeapFrogTrieJoin2014}.}
    \begin{align*}
        \sum_{\tau \in [N]}
        \prod_{R(\bm X)\in\atoms(Q)}|R^{(\tau)}(\bm X)\ltimes \delta_\tau S(\bm Y)|^{\lambda_{R(\bm X)}} &\leq
        \sum_{\tau \in [N]}
        \prod_{R(\bm X)\in\atoms(Q)}|R^{(N)}(\bm X)\ltimes \delta_\tau S(\bm Y)|^{\lambda_{R(\bm X)}} \\
        &\leq \prod_{R(\bm X)\in\atoms(Q)}
        |R^{(N)}|^{\lambda_{R(\bm X)}} = \theta
    \end{align*}
    The first inequality above holds because $R^{(\tau)}\subseteq R^{(N)}$,
    while the second inequality follows from~\eqref{eq:query-decom-lemma:app}.
    Moreover, note that~\eqref{eq:lem:insert_only_in_AGM:1} is an upper bound on the output size of $\delta_\tau Q$ (because it is the AGM-bound). From the above, we have
    $\sum_{\tau \in [N]} |\delta_\tau Q| \leq \theta$.
\end{proof}

\subsection{Upper bound for $\ivmp[Q]$}
\label{sec:insert_only:ub:full}
In this section, we prove Theorem~\ref{thm:main_inserts} for $\ivmp[Q]$, which basically says the following:
Given a query $Q$ whose fractional hypertree width is ${\fw(Q)}$ and a stream of $N$ single-tuple inserts
(where all input relations are initially empty),
we can maintain some data structure for $Q$ and update it after every insert where the total update time
over all inserts combined is $\bigO{N^{\fw(Q)}}$. Moreover, we can use this data structure
at any time to do constant-delay enumeration of the output of $Q$.
In particular, fix an optimal tree decomposition of $Q$, i.e.~one whose width is ${\fw(Q)}$. The data structure that we use consists of a collection of intermediate relations: One for each
bag of the chosen tree decomposition.
Moreover, every bag will be ``calibrated'' with its children bags in the tree decomposition.
We use Lemma~\ref{lem:insert_only_in_AGM} to update every bag within a total time of $\bigO{N^{\fw(Q)}}$.
Moreover, because every bag is calibrated with its children bags, we can always do
constant-delay enumeration starting from the root bag. Example~\ref{ex:insert_only_full} demonstrates the idea.

\begin{citedthm}[\ref{thm:main_inserts}]
    For any query $Q$,
    both $\ivmp[Q]$ and $\ivmpd[Q]$ can be solved with $\bigO{N^{\fw(Q)-1}}$ amortized update time, where $N$ is the number of single-tuple inserts.
\end{citedthm}
\begin{proof}[Proof of Theorem~\ref{thm:main_inserts}]
    \underline{The case of $\ivmp[Q]$:}
    Let $(\calT, \chi)$ be a tree decomposition of $Q$ whose fractional hypertree width is ${\fw(Q)}$. (Recall notation from Section~\ref{app:prelims}.)
    Fix an orientation of $\calT$ starting from some root node $r \in V(\calT)$.
    Given a node $t \in V(\calT)$, we will use $\parent(t), \children(t), \ancestors(t)$ and $\decendants(t)$ to refer to the parent, children, ancestors and descendants of $t$ in $\calT$, respectively, under the fixed orientation of $\calT$.
    For every node $t \in V(\calT)$,
    let $\bm X_t \defeq \chi(t)$ be the variables of $t$
    and $Q_t \defeq Q_{\bm X_t}$ be the {\em restriction} of $Q$ to $\bm X_t$, 
    given by Definition~\ref{defn:restriction}.
    We call $Q_t$ a {\em bag query}.
    Let $\bm X \defeq \vars(Q)$.
    By definition, the query $Q$ is equivalent to the join of the outputs of all bag queries:
    \begin{align*}
        Q(\bm X) = \bigwedge_{t \in V(\calT)} Q_t(\bm X_t)
    \end{align*}
    Moreover, for every node $t \in V(\calT)$, we have $\rho^*(Q_t) \leq {\fw(Q)}$.
    For each node $t$, let $\bm Y_t \defeq \bm X_t \cap \bm X_p$ where $p$ is the parent of $t$.
    (If $t$ is the root, then $\bm X_p$ is considered to be empty, hence $\bm Y_t$ is empty.)
    We construct a query plan for $Q$ by defining a collection of intermediate relations
    $Q_t', P_t'$ for every node $t \in V(\calT)$. In particular, $Q_t', P_t'$ are defined
    inductively starting from the leaves of $\calT$ all the way to the root $r$:
    \begin{align}
        Q'_t(\bm X_t) &\defeq Q_t(\bm X_t) \wedge \bigwedge_{c \in \children(t)} P_c'(\bm Y_c)\label{eq:td:bottomup:Q}\\
        P_t'(\bm Y_t) &\defeq \pi_{\bm Y_t} Q'_t(\bm X_t)\label{eq:td:bottomup:P}
    \end{align}
    For any node $t \in V(\calT)$, we have:
    \begin{align}
        Q'_t(\bm X_t) = \pi_{\bm X_t}\bigwedge_{d \in \decendants(t)\cup\{t\}} Q_d(\bm X_d)
        \label{eq:td:bottomup}
    \end{align}
    The proof is by induction starting from the leaves of $\calT$ to the root.
    It is similar to the bottom-up phase of the Yannakakis algorithm~\cite{Yannakakis81}.

    The AGM-bound of each bag query $Q_t$ is at most $N^{{\fw(Q)}}$. Since adding a new relation to a bag does not increase the AGM bound of the bag, the AGM-bound of each $Q'_t$ is also at most $N^{{\fw(Q)}}$. Hence, the  output size of each $Q'_t$ (and by extension $P'_t$) is at most $N^{\fw(Q)}$.
    Each query $Q'_t$ can have two types of inserts: Inserts coming from (the input relations of) $Q_t$ and inserts coming from $P'_c$ for the children $c$ of $t$.
    By Lemma~\ref{lem:insert_only_in_AGM}, all updates to $Q'_t$ combined take time
    $\bigO{N^{\fw(Q)}}$. Moreover, the projection $P'_t$ of $Q'_t$ can also be updated
    in the same time.

    In order to do constant-delay enumeration of the output of $Q(\bm X)$, we start from the root bag $r$ and enumerate the output of $Q'_r$. By~\eqref{eq:td:bottomup},
    $Q'_r(\bm X_r) = \pi_{\bm X_r} Q(\bm X)$. For every tuple $\bm x_r$ in the output of $Q'_r$, we enumerate the corresponding tuples $\bm x_c$ in the output of $Q'_c$ for every child $c$ of $r$. (Note that for every child $c$, there must be at least one tuple $\bm x_c$ that joins with $\bm x_r$, thanks to~\eqref{eq:td:bottomup}.)
    We inductively continue this process until we reach the leaves of $\calT$.
    In particular, for every node $t \in V(\calT)$ and every tuple $\bm x_t$ that we enumerate from $Q'_t$, we enumerate the corresponding tuples $\bm x_c$ from $Q'_c$ for every child $c$ of $t$, and we rely on~\eqref{eq:td:bottomup} to ensure that there is at least
    one tuple $\bm x_c$ for every child $c$ of $t$.
\end{proof}

\subsection{Upper bound for $\ivmpd[Q]$}
\label{sec:insert_only:ub:delta}
In this section, we prove Theorem~\ref{thm:main_inserts} for the $\ivmpd[Q]$ problem.
In particular, we explain how to modify the algorithm and proof from Section~\ref{sec:insert_only:ub:full}, which was originally meant for the $\ivmp[Q]$ problem.
Just like we did in Section~\ref{sec:insert_only:ub:full}, here also we will construct an optimal
tree decomposition of $Q$ that minimizes the fractional hypertree width, and we will maintain
materialized relations for every bag of the tree decomposition.
However, suppose that we have an insert of a tuple $\bm t$ into a relation $R$ and we want to enumerate
the new output tuples $\delta Q$ that are added due to this particular insert.
We can no longer start our constant-delay enumeration from the root of the tree decomposition, as we did in Section~\ref{sec:insert_only:ub:full}. Instead,
we have to find the bag $B$ that contains $R$, select $\bm t$, and start our enumeration
from $B$, as if $B$ were the root of the tree decomposition.
In order to be able to enumerate the output of $Q$ starting from any bag $B$ of the tree decomposition, all bags need to be {\em fully calibrated} with each other.
In particular, in Section~\ref{sec:insert_only:ub:full}, we only calibrated bags in a bottom-up
fashion where each bag was calibrated with its children bags. But here, we have to do full calibration consisting of two phases: One bottom-up phase where we calibrate each bag with its children, followed by a top-down phase where we calibrate each bag with its parent.
These two calibrations are similar to the two phases of the Yannakakis algorithm~\cite{Yannakakis81}.

\begin{citedthm}[\ref{thm:main_inserts}]
    For any query $Q$,
    both $\ivmp[Q]$ and $\ivmpd[Q]$ can be solved with $\bigO{N^{\fw(Q)-1}}$ amortized update time, where $N$ is the number of single-tuple inserts.
\end{citedthm}
\change{Example~\ref{ex:insert_only_delta} illustrates the proof idea.}
\begin{proof}[Proof of Theorem~\ref{thm:main_inserts}]
    \underline{The case of $\ivmpd[Q]$:}
    We will use the same notation as in the proof of Theorem~\ref{thm:main_inserts} for $\ivmp[Q]$ in Section~\ref{sec:insert_only:ub:full}.
    In particular, here also we will fix an optimal tree decomposition $(\calT, \chi)$ of $Q$ and we will fix an orientation of $\calT$ starting from some root node $r \in V(\calT)$. In addition to $Q'_t$ and $P'_t$ given by~\eqref{eq:td:bottomup:Q} and~\eqref{eq:td:bottomup:P}, we will maintain two additional materialized relations $Q_t''$ and $P_t''$ for every node $t \in V(\calT)$. However, this time, we will define them inductively starting from the root $r$ all the way to the leaves of $\calT$: (Recall that for the root $r$, we have $\bm Y_r = \emptyset$.)
    \begin{align*}
        Q''_t(\bm X_t) &\defeq Q'_t(\bm X_t) \wedge P''_t(\bm Y_t)\\
        P''_t(\bm Y_t) &\defeq
        \begin{cases}
            \true&\text{if $t = r$}\\
            \pi_{\bm Y_t} Q''_p(\bm X_p),\text{\quad where $p = \parent(t)$} &\text{otherwise}
        \end{cases}
    \end{align*}
    By induction, we can now prove that for any node $t$,
    \begin{align}
        Q''_t(\bm X_t) = \pi_{\bm X_t}Q(\bm X) \label{eq:td:topdown}
    \end{align}
    (Compare the above to~\eqref{eq:td:bottomup}.)
    The induction now starts from the root $r$ and goes all the way to the leaves of $\calT$. It is similar to the top-down phase of the Yannakakis algorithm~\cite{Yannakakis81}.
    The AGM bound for each $Q''_t$ is still upper bounded by $N^{\fw(Q)}$. Hence, they can be maintained in a total time of $\bigO{N^{\fw(Q)}}$. The same goes for $P''_t$.
    
    Now suppose we have an insert of a tuple $\bm t$ into a relation $R(\bm Z)$ and we want to enumerate the new output tuples $\delta Q$ that are added due to this insert.
    We find the node $t$ in $\calT$ where $\bm Z \subseteq\chi(t)$. The existence of such a node $t$ is guaranteed by the definition of a tree decomposition (Definition~\ref{defn:TD}). Now, we re-orient the tree $\calT$ so that $t$ becomes the root.
    We enumerate output tuples $\bm x_t$ of $Q''_t$ that join with $\bm t$. For each such tuple $\bm x_t$, we enumerate the corresponding tuples $\bm x_c$ from $Q''_c$ for every child $c$ of $t$. (These are the children of $t$ under the new orientation of the tree $\calT$). By~\eqref{eq:td:topdown}, we guarantee that for each tuple $\bm x_t$
    and each child $c$, there must be at least one tuple $x_c$ in $Q''_c$ that joins with $\bm x_t$. We continue this process inductively all the way to the leaves of $\calT$.
\end{proof}

\subsection{Lower bound for $\ivmp[Q]$ and $\ivmpd[Q]$}
\label{sec:insert_only:lb}
In this section, we prove Proposition~\ref{prop:insert_only_lower_bound}, which gives a lower bound on the amortized updated time for both $\ivmp[Q]$ and $\ivmpd[Q]$ in terms of a lower bound on the static evaluation of $Q$.

\change{
\begin{citedprop}[\ref{prop:insert_only_lower_bound}]
    For any query $Q$ and any constant $\gamma > 0$,
    neither $\ivmp[Q]$ nor $\ivmpd[Q]$ can be solved with $\tildeO{N^{\lb(Q)-1-\gamma}}$ (amortized) update time.
\end{citedprop}
}
\begin{proof}[Proof of Proposition~\ref{prop:insert_only_lower_bound}]
    Fix a query $Q$. We first consider $\ivmp[Q]$.
    Suppose, for sake of contradiction, that there is a constant $\gamma > 0$
    where $\ivmp[Q]$ can be solved with amortized update time $\tildeO{N^\kappa}$
    where $\kappa \defeq \omega(Q) - 1 - \gamma$. We will show how to use an oracle
    for $\ivmp[Q]$ to solve $\eval[Q]$ in time $\tildeO{|\calD|^{\omega(Q)-\gamma} + |Q(\calD)|}$
    on any database instance $\calD$. This contradicts the definition of $\omega(Q)$
    (Definition~\ref{defn:fw_lb}).

    We construct a stream of inserts of length $N = |\calD|$ where we insert
    one tuple of $\calD$ at a time. We apply these inserts one by one in a total time of $\tildeO{N^{1+\kappa}}$.
    After the last insert, we invoke constant-delay enumeration of $Q(\calD)$.
    Thus, the overall time is $\tildeO{N^{1+\kappa} + |Q(\calD)|}$.

    For $\ivmpd[Q]$, the proof is very similar except that we enumerate $\delta Q(\calD)$
    after every insert.
\end{proof}

\section{Missing Details in Section~\ref{sec:intersection-joins}}

\label{app:intersection-joins}

In this section, we review some necessary background on intersection joins in the {\em static} setting.
Let $Q$ be a query, $\ov Q$ be its \univariate extension, and $\wh Q$ be its \multivariate extension.
Suppose we want to evaluate $\ov Q$ over a given database instance $\ov\calD$.
This evaluation problem can be reduced~\cite{KhamisCKO22} to the evaluation of 
$\wh Q$ over some database instance, $\wh \calD$, that satisfies $|\wh\calD|=\tildeO{|\ov\calD|}$ and can be constructed in time $\tildeO{|\ov\calD|}$.
Moreover, this reduction is {\em optimal} (up to a $\tildeO{1}$ factor). In particular, the query $\ov Q$ is exactly as hard as the hardest component $\wh Q_{\bm\sigma}$ in $\wh Q$. This is shown~\cite{KhamisCKO22} by a backward reduction from each component $\wh Q_{\bm\sigma}$
to $\ov Q$. We explain both reductions in detail below.
Recall that both reductions are in the context of {\em static evaluation}.
However later on in Section~\ref{app:fully_dynamic}, we will rely on both reductions  in order
to prove upper and lower bounds on the IVM problem in the \insertdelete setting.

\subsection{Reduction from $\ov Q$ to $\wh Q$}
\label{sec:intersection-joins:forward}
The reduction in~\cite{KhamisCKO22} from a \univariate extension query $\ov Q$ to the corresponding \multivariate extension query $\wh Q$ relies on constructing a {\em segment tree}
for the interval variable $[Z]$ in $\ov Q$.

\paragraph{Segment Trees}
Given $N = 2^n$ for some $n \in \mathbb{N}$, a {\em (discrete) segment tree} $\calT_N$ over $N$ is a complete 
binary tree whose nodes are discrete intervals that are subsets of $[N]$. A segment tree
has the following properties:

\begin{itemize}[leftmargin=*]
    \item The leaves from left-to-right are the intervals 
    $[1,1]$, \ldots , $[N,N]$ consisting of single values.

    \item Each inner node is the union of its child nodes. 
\end{itemize}

Figure~\ref{fig:segment-tree} depicts a segment tree $\calT_8$ for $N = 8$.
Following \cite{KhamisCKO22}, we identify each node of a segment tree by a bitstring:
The root is identified by the empty string $\varepsilon$;
for each inner node that is identified by a bitstring $b$,
the left and right child nodes are identified by the bitstrings $b\circ 0$ and respectively 
$b\circ 1$ (where $\circ$ denotes string concatenation). We denote by $V(\calT_N)$ the set of nodes of the segment tree $\calT_N$ and 
by $\parent(v)$ the parent node of a node $v$.
Given an interval $x \subseteq [N]$, its {\em canonical partition} with respect to 
$N$ is defined as 
the set of {\em maximal} nodes $y$ of $\calT_N$
whose corresponding intervals are contained in $x$, i.e.
\begin{align*}\canpart_N(x) \defeq \{y \in V(\calT_N) \mid  y \subseteq x \text{ but } 
\parent(y) \not\subseteq x\}
\end{align*}
For any interval $x \subseteq [N]$, 
$\canpart_N(x)$ is partition of $x$; see~\cite{KhamisCKO22}.
Moreover $|\canpart_N(x)| = \bigO{\log N}$.
For example, in the segment tree $\calT_8$ from Figure~\ref{fig:segment-tree}, we have,
\begin{align*}
    \canpart_8([1,8]) &= \{\varepsilon\} \\
    \canpart_8([2, 5]) &= \{001, 01, 100\}
\end{align*}

\begin{lemma}[\cite{KhamisCKO22}] 
\label{lem:interv_dec}
Consider a set of discrete intervals $\calI = \{x_1, \ldots, x_k\}$ with $x_i \subseteq [N]$ for some
$N \in \mathbb{N}$ for $i \in [k]$. It holds that $\bigcap_{i \in [k]} x_i \neq \emptyset$ if and only if 
there is a permutation $\bm\sigma$ of $[k]$ and a tuple $(b_1, \ldots , b_k)$ of bitstrings such that
\begin{align}
(b_1 \circ\cdots \circ b_j) \in \canpart_N(x_{\sigma_j}) \text{ for }j \in [k].
\label{eq:interv_dec}
\end{align}
Moreover, the set of bitstrings $b_1\circ\cdots\circ b_k$ for all tuples $(b_1, \ldots , b_k)$ that satisfy~\eqref{eq:interv_dec} for some permutation 
$\bm \sigma$ of $[k]$ is exactly the canonical
partition of the intersection $\bigcap_{i \in [k]} x_i$:
\begin{align*}
    \canpart_N(\bigcap_{i \in [k]} x_i) = \{(b_1 \circ \ldots \circ b_k) \mid
    \text{there is a permutation $\bm\sigma$ of $[k]$ that satisfies~\eqref{eq:interv_dec}}\}
\end{align*}
\end{lemma}
(Recall that $\circ$ is the concatenation operator on strings.)
The above lemma is the core of the reduction.
It basically says that a set of $k$ intervals $\calI$ intersect if and only if
their canonical partitions contain $k$ nodes that fall on the same root-to-leaf path in the segment tree. The nodes on this path can appear in any order, and the order is captured by the permutation $\bm\sigma$. For example, consider again the segment tree from Figure~\ref{fig:segment-tree}.
The intervals $([1, 8], [2, 5], [3, 7])$ overlap. Indeed, the canonical partitions of these intervals
contain the nodes $(\varepsilon, 01, 01)$ respectively, and these nodes fall on the path
from the root $\varepsilon$ to the leaf $010$, thus witnessing the intersection.
Alternatively, the intersection is also witnessed by the nodes $(\varepsilon, 100, 10)$,
which fall on the path from the root $\varepsilon$ to the leaf $100$.
These are the only two witness to the intersection, and they correspond to the two bitstrings $01$ and $100$. Note that $\{01, 100\}$ is indeed the canonical partition of $[3, 5]$, which is the intersection of the three intervals $([1, 8], [2, 5], [3, 7])$. This is basically what the second part of the lemma says.

\begin{example}
    \label{ex:IJ:forward:reduction}
Suppose we want to solve the query $\ov Q_\triangle$ from Example~\ref{ex:triangle-query:uni}. We can start by constructing
a segment tree for the interval variable $[Z]$ and replacing each value $[z]$ of
$[Z]$ by the components in its canonical partition $\canpart_N([z])$. 
We can then use Lemma~\ref{lem:interv_dec} to look for intersections. In particular, let's first try to find triples of segment-tree nodes $(z^R, z^S, z^T)$ from
$R, S, T$ respectively that appear on the same root-to-leaf path in the specific order $(z^R, z^S, z^T)$.
This can only happen if the bitstring $z^R$ is a prefix of $z^S$ which in turn is a prefix of $z^T$. To test this prefix pattern, we can break each $z^S$ into a concatenation two strings
$z^S = z^S_1\circ z^S_1$ in all possible ways and similarly break $z^T$ into a concatenation
of three string $z^T = z^T_1\circ z^T_2\circ z^T_3$ and then check whether
$z^R = z^S_1 = z^T_1$ and $z^S_2 = z^T_2$. In particular,
we construct the following input relations:
\begin{align}
    \wh R_{123} &= \{(z_1, a, b) \mid  \exists [z] : ([z], a, b)\in \ov R \wedge z_1\in \canpart_N([z])\}\label{eq:triangle:canpart}\\
    \wh S_{123} &= \{(z_1, z_2, b, c) \mid  \exists [z] : ([z], b, c)\in \ov S \wedge (z_1\circ z_2)\in \canpart_N([z])\}\nonumber\\
    \wh T_{123} &= \{(z_1, z_2, z_3, a, c) \mid  \exists z : ([z], a, c)\in \ov T \wedge (z_1\circ z_2\circ z_3)\in \canpart_N([z])\}\nonumber
\end{align}
Then we solve the query
\begin{align}
    \wh Q_{123}(Z_1, Z_2, Z_3, A, B, C) =
        \wh R_{123}(Z_1, A, B) \wedge 
        \wh S_{123}(Z_1, Z_2, B, C) \wedge
        \wh T_{123}(Z_1, Z_2, Z_3, A, C)
    \label{eq:triangle-query:multi-123}
\end{align}
Note that $|\wh T_{123}| = \bigO{|\ov T|\cdot \polylog(N)}$ and the same goes for $\wh R_{123}$ and $\wh S_{123}$.
This is because the height of the segment tree is $\bigO{\log N}$.
Consider now the database instance depicted in Figure~\ref{table:intersection-join-example}.
Note that the corresponding relation $\wh R_{123}$ above contains (among others) the tuple
$(\varepsilon, a_1, b_1)$  while $\wh S_{123}$ contains the tuples
$(\varepsilon, 01, b_1, c_1)$ and $\wh T_{123}$
contains the tuple $(\varepsilon, 01, \varepsilon, a_1, c_1)$.
Those tuples join together producing the output tuple $(\varepsilon, 01, \varepsilon, a_1, b_1, c_1)$
of $\wh Q_{123}$, which is a witness to the output tuple $([3, 5], a_1, b_1, c_1)$ of $\ov Q_\triangle$.

The query $\wh Q_{123}$ above only considers the specific permutation $(R, S, T)$ of the input relations $\{R, S, T\}$. We have five other permutations corresponding to the following five queries:
\begin{align}
    \wh Q_{132} &=
        \wh R_{132}(Z_1, A, B) \wedge
        \wh S_{132}(Z_1, Z_2, Z_3, B, C) \wedge
        \wh T_{132}(Z_1, Z_2, A, C) \nonumber\\
    \wh Q_{213} &=
        \wh R_{213}(Z_1, Z_2, A, B) \wedge
        \wh S_{213}(Z_1, B, C) \wedge
        \wh T_{213}(Z_1, Z_2, Z_3, A, C) \nonumber\\
    \wh Q_{231} &=
        \wh R_{231}(Z_1, Z_2, Z_3, A, B) \wedge
        \wh S_{231}(Z_1, B, C) \wedge
        \wh T_{231}(Z_1, Z_2, A, C) \nonumber\\
    \wh Q_{312} &=
        \wh R_{312}(Z_1, Z_2, A, B) \wedge
        \wh S_{312}(Z_1, Z_2, Z_3, B, C) \wedge
        \wh T_{312}(Z_1, A, C) \nonumber\\
    \wh Q_{321} &=
        \wh R_{321}(Z_1, Z_2, Z_3, A, B) \wedge
        \wh S_{321}(Z_1, Z_2, B, C) \wedge
        \wh T_{321}(Z_1, A, C) \label{eq:triangle-query:multi-132-321}
\end{align}
Over the database instance depicted in Figure~\ref{table:intersection-join-example},
note that $\wh R_{132}$ contains the tuple $(\varepsilon, a_1, b_1)$ while
$\wh T_{132}$ contains $(\varepsilon, 10, a_1, c_1)$ and
$\wh S_{132}$ contains $(\varepsilon, 10, 0, b_1, c_1)$.
These tuples join together producing the output tuple $(\varepsilon, 10, 0, a_1, b_1, c_1)$
of $\wh Q_{132}$, which is a yet another witness to the output tuple $([3, 5], a_1, b_1, c_1)$ of $\ov Q_\triangle$. In particular, the output tuple $([3, 5], a_1, b_1, c_1)$
has two different witnesses in $\wh Q_\triangle$ specifically  because the canonical partition
of $[3, 5]$ consists of two nodes, namely $01$ and $100$.
More generally, any output tuple $([z], a, b, c)$ of $\ov Q_\triangle$ has as many witnesses in $\wh Q_\triangle$
as the number of nodes in the canonical partition of $[z]$.
(Recall that for any interval $[z]$, we have $|\canpart_N([z])| = \bigO{\log N}$.)

The answer to $\ov Q_\triangle$ can be retrieved from the union of the answers to the six queries above. The new relations $\wh R_{123}, \wh S_{123}, \ldots$ can be constructed in time
$\bigO{N\cdot\polylog N}$ where $N$ is the input size to $\ov Q_\triangle$.
\end{example}

Extrapolating from the above example, we introduce the following definition:
\begin{definition}[Canonical partition of a relation $\canpart_N^{(i)}(R)$]
    \label{defn:canpart:tuple}
    Let $\bm t = ([z], x_1, \ldots, x_m)$ be a tuple where $[z]$ is an interval and
    $x_1, \ldots, x_m$ are points. Given natural numbers $N$ and $i$, the
    {\em canonical partition} of $\bm t$ is defined as the following set of tuples
    \begin{align}
        \canpart_N^{(i)}([z], x_1, \ldots, x_m) \defeq
        \{(z_1, \ldots, z_i, x_1, \ldots, x_m) \mid (z_1\circ \cdots \circ z_i) \in \canpart_N([z])\}
    \end{align}
    We lift the definition of a canonical partition $\canpart_N^{(i)}$ to relations and define $\canpart_N^{(i)}(R)$ as
    the union of the canonical partitions of the tuples in the given relation $R$.
    We also lift the definition to deltas and define $\canpart_N^{(i)}(\delta R)$ as the corresponding delta for $\canpart_N^{(i)}(R)$.
\end{definition}

\begin{definition}[Canonical partition of a database instance $\canpart_N(\ov \calD)$]
    Let $Q = R_1(\bm X_1)\wedge \cdots\wedge R_k(\bm X_k)$ be a query and $\ov Q = \ov R_1([Z], \bm X_1)\wedge \cdots\wedge \ov R_k([Z],\bm X_k)$ its \univariate extension.
    Given a permutation $\bm\sigma\in\Sigma_k$, let $\wh Q_{\bm\sigma}$ be the corresponding component
    of its \multivariate extension $\wh Q$:
    \begin{align*}
        \wh Q_{\bm\sigma} = \wh R_{\sigma_1}(Z_1, \bm X_{\sigma_1}) \wedge \wh R_{\sigma_2}(Z_1, Z_2, \bm X_{\sigma_2}) \wedge \cdots\wedge \wh R_{\sigma_k}(Z_1, \ldots, Z_k,\bm X_{\sigma_k})
    \end{align*}
    Given a database instance $\ov \calD$ for $\ov Q$, the {\em canonical partition} of $\ov \calD$ using $\bm\sigma$, denoted by $\canpart_N^{\bm\sigma}(\ov \calD)$, is a database instance $\wh \calD_{\bm\sigma}$
    for $\wh Q_{\bm\sigma}$ that is defined as:
    \begin{align}
        \wh R_{\sigma_i} \defeq \canpart_N^{(i)}(\ov R_{\sigma_i}),
        \quad\quad\text{ for }i \in [k]
        \label{eq:canpart:instance}
    \end{align}
    Moreover, define $\canpart_N(\ov \calD)$ to be a database instance for $\wh Q$
    which is a combination of $\canpart_N^{\bm \sigma}(\ov \calD)$ for every $\bm\sigma \in
    \Sigma_k$:
    \begin{align*}
        \canpart_N(\ov \calD) \defeq \left(\canpart_N^{\bm\sigma}(\ov \calD)\right)_{\bm \sigma \in \Sigma_k}
    \end{align*}
    We lift the definition of $\canpart_N^{\bm\sigma}$ to deltas and define $\canpart_N^{\bm\sigma}(\delta\ov\calD)$
    to be the corresponding delta for $\canpart_N^{\bm\sigma}(\ov\calD)$.
    We define $\canpart_N(\delta\ov\calD)$ similarly.
    \label{defn:canpart:instance}
\end{definition}
\eqref{eq:triangle:canpart} is a special case of~\eqref{eq:canpart:instance} above.
The following theorem basically says that evaluating $\ov Q(\ov \calD)$ is
equivalent to evaluating $\wh Q$ on the canonical partition of $\ov \calD$.
\begin{theorem}[Implicit in~\cite{KhamisCKO22}]
    \label{thm:IJ:forward:reduction}
    Given a tuple $(z_1, \ldots, z_k, x_1, \ldots, x_m)$, define
    \begin{align}
        G_k(z_1, \ldots, z_k, x_1, \ldots, x_m) \defeq
        (z_1 \circ \cdots\circ z_k, x_1, \ldots, x_m)
    \end{align}
    We lift the definition of $G_k$ to relations and define $G_k(R)$ to be the relation resulting from
    applying $G$ to every tuple in $R$. Then,
    \begin{align}
        \canpart_N^{(1)}(\ov Q(\ov \calD)) = 
        G_k(\wh Q(\canpart_N(\ov \calD)))
        \label{eq:IJ:forward:reduction}
    \end{align}
\end{theorem}

\subsection{Reduction from $\wh Q$ to $\ov Q$}
\label{sec:intersection-joins:backward}
In the previous Section~\ref{sec:intersection-joins:forward}, we saw how to reduce the {\em static}
evaluation of $\ov Q$
to $\wh Q$.
We now summarize the backward reduction from the static evaluation of each component $\wh Q_{\bm\sigma}$ of $\wh Q$ back to $\ov Q$, based on~\cite{KhamisCKO22}. Later on, we will use this backward reduction to prove {\em lower bounds} on the IVM problem in Section~\ref{sec:backward_reduction}.

\change{Example~\ref{ex:IJ:backward:reduction} introduces the high-level idea of the backward reduction.}
Extrapolating from \change{Example~\ref{ex:IJ:backward:reduction}}, we introduce the following concept:
\begin{definition}[Interval version of a database instance $\iv(\wh\calD_{\bm\sigma})$]
    \label{defn:interval-version}
    Let $Q = R_1(\bm X_1)\wedge \cdots\wedge R_k(\bm X_k)$ be a query,
    $\ov Q$ its \univariate extension, $\bm\sigma$ be a permutation of $[k]$, and $\wh Q_{\bm\sigma}$ be a component of the \multivariate extension $\wh Q$ of $Q$.
    Let $\wh\calD_{\bm\sigma}$ be an {\em arbitrary} database instance for $\wh Q_{\bm\sigma}$. Define the {\em interval version}  of $\wh\calD_{\bm\sigma}$, denoted by $\iv(\wh\calD_{\bm\sigma})$, to be a database instance $\ov D$ for $\ov Q$
    that is constructed as follows:
    WLOG assume each value $x$ that appears in $\wh\calD_{\bm\sigma}$ is a bitstring
    of length exactly $\ell$ for some constant $\ell$ (pad with zeros if needed) and let $N\defeq 2^{k\ell}$.
    Let $\segment_{N}(b)$ map a bitstring $b$ of length at most $k\ell$ to the corresponding interval in the segment tree $\calT_{N}$.
    Then $\iv(\wh\calD_{\bm\sigma})$ consists of the following relations for $i \in [k]$:
    \begin{align*}
        \ov R_{\sigma_i} \defeq \{(\segment_{N}(z_1\circ \cdots \circ z_i), x_1, \ldots, x_{m_i})
        \mid (z_1, \ldots, z_i, x_1, \ldots, x_{m_i}) \in \wh R_{\sigma_i}\}
    \end{align*}
    We lift the definition of $\iv$ to deltas and define $\iv(\delta\wh\calD_{\bm\sigma})$
    to be the corresponding delta for $\iv(\wh\calD_{\bm\sigma})$.
\end{definition}

\begin{proposition}
For any database instance $\wh\calD_{\bm\sigma}$ for $\wh Q_{\bm\sigma}$, we have
\begin{align}
    |\iv(\wh\calD_{\bm\sigma})| = |\wh\calD_{\bm\sigma}|
    \label{eq:iv_same_input_size}
\end{align}
Moreover, $\iv(\wh\calD_{\bm\sigma})$ can be computed in linear time in $|\wh\calD_{\bm\sigma}|$.
\label{prop:iv_same_input_size}
\end{proposition}

The following theorem basically says that evaluating $\wh Q_{\bm\sigma}$ on
$\wh\calD_{\bm\sigma}$ is equivalent to evaluating $\ov Q$ on the interval version
of $\wh\calD_{\bm\sigma}$.
\begin{theorem}[Implicit in~\cite{KhamisCKO22}]
    \label{thm:IJ:backward:reduction}
    Given a tuple $(z_1, \ldots, z_k, x_1, \ldots, x_m)$, define
    \begin{multline}
        H_k([z], x_1, \ldots, x_m) \defeq
        \{(z_1, \ldots, z_k, x_1, \ldots, x_m)\mid
        (z_1 \circ \cdots \circ z_k) \in \canpart_N([z]) \wedge |z_1| = \cdots = |z_k|\}
    \end{multline}
    Moreover, we lift $H_k$ to relations and define $H_k$ of a relation $R$ to be the relation resulting from applying $H_k$ to every tuple in $R$.
    Then,
    \begin{align}
        \wh Q_{\bm\sigma}(\wh\calD_{\bm\sigma}) &= 
        H_k(\ov Q(\iv(\wh\calD_{\bm\sigma})))
        \label{eq:IJ:backward:reduction}\\
        |\wh Q_{\bm\sigma}(\wh\calD_{\bm\sigma})| &= 
        |\ov Q(\iv(\wh\calD_{\bm\sigma}))|
        \label{eq:iv_same_output_size}
    \end{align}
\end{theorem}

\section{Missing Details from Section~\ref{sec:fully_dynamic}}
\label{app:fully_dynamic}
\change{
Theorem~\ref{thm:main_fully_dynamic} gives an upper bound on the
amortized update time for both the $\ivmpm$ and $\ivmpmd$
problems in terms of the current database size $|\calD|$.
In order to prove it, we first prove Lemma~\ref{lmm:main_fully_dynamic}
which gives a weaker upper bound where $|\calD|$ is replaced by the number of single-tuple
updates $N$. We prove Lemma~\ref{lmm:main_fully_dynamic} for $\ivmpm$ in Section~\ref{sec:forward_reduction} and for $\ivmpmd$ in Section~\ref{sec:forward_reduction:delta_version}.
In Section~\ref{sec:forward_reduction:calD}, we show how to use Lemma~\ref{lmm:main_fully_dynamic} as a black box in order to prove Theorem~\ref{thm:main_fully_dynamic}.

In contrast to Theorem~\ref{thm:main_fully_dynamic}, Theorem~\ref{thm:fully_dynamic_lower_bound} gives a lower bound
on the amortized update time only for $\ivmpmd$ in terms of the current database size $|\calD|$.
We prove it in Section~\ref{sec:backward_reduction}, by proving a stronger lower bound where $|\calD|$ is replaced by the number of single-tuple updates $N$.
}

\subsection{Upper Bound for $\ivmpm[Q]$}
\label{sec:forward_reduction}

Example~\ref{ex:main_fully_dynamic} introduces the high-level idea of our algorithm for $\ivmpm[Q]$.
Algorithm~\ref{alg:ivmpm} summarizes our algorithmic
framework for $\ivmpm[Q]$ that meets Lemma~\ref{lmm:main_fully_dynamic}.
Algorithm~\ref{alg:ivmpm} assumes that the length $N$ of the update stream is known in advance.
However, the following remark explains how to remove this assumption.
\begin{remark}
    \label{rmk:doubling:N}
If the number of updates $N$ is not known in advance, we can initialize $N$ to $1$ and keep doubling $N$ every time we reach $N$ updates.
We show in the proof of Lemma~\ref{lem:reduction_ivmtilde_ivmstar} that amortized analysis maintains the same overall runtime.
\end{remark}

\begin{algorithm}[th!]
    \caption{Algorithm of $\ivmpm$[Q]}
    \label{alg:ivmpm}
    \begin{algorithmic}
        \State {\textbf{Inputs}}
        \begin{itemize}
            \item $Q = R_1(\bm X_1)\wedge \cdots\wedge R_k(\bm X_k)$ \Comment{Input query to $\ivmpm$}
            \item An initially empty database $\calD$ \Comment{Database instance for $Q$}
            \item $N$: The length of the update stream \Comment{See Remark~\ref{rmk:doubling:N}}
        \end{itemize}
        \\\hrulefill
        \State {\textbf{Initialization}}
        \begin{itemize}
            \item Construct the \univariate extension $\ov Q$ \Comment{Definition~\ref{defn:univariate}}
            \item Initialize $\ov \calD$ to be empty \Comment{Database instance for $\ov Q$}
            \item Construct the \multivariate extension $\wh Q$ 
            \Comment{Eq.~\eqref{eq:multivariate}}
            \item Initialize $\wh \calD$ to be empty \Comment{Database instance for $\wh Q$}
            \item For each component $\wh Q_{\bm\sigma}$ of $\wh Q$
            \begin{itemize}
                \item Initialize a tree decomposition using Theorem~\ref{thm:main_inserts}
            \end{itemize}
        \end{itemize}
        \\\hrulefill
        \State {\textbf{Invariants}}
        \begin{itemize}
            \item $\wh \calD= \canpart_N(\ov \calD)$\Comment{Definition~\ref{defn:canpart:instance}}
        \end{itemize}
        \\\hrulefill
        \State {\textbf{Insertion $\delta_\tau\calD = \{+R_j(\bm t)\}$}}\Comment{The $\tau$-th update inserts a tuple $\bm t$ into $R_j$}
        \begin{itemize}
            \item $\delta_\tau \ov\calD \gets \{+\ov R_j([\tau, \infty], \bm t)\}$
            \Comment{Insert $([\tau, \infty], \bm t)$ into $\ov R_j$}
            \item $\ov\calD \gets \ov\calD \applydelta \delta_\tau \ov D$\Comment{Apply $\delta_\tau\ov \calD$ to $\ov\calD$}
            \item $\delta_\tau\wh \calD \gets \canpart_N(\delta_\tau \ov \calD)$ \Comment{Definition~\ref{defn:canpart:instance}}
            \item $\wh\calD \gets \wh\calD \applydelta \delta_\tau\wh \calD$ \Comment{Use Theorem~\ref{thm:main_inserts}}
        \end{itemize}
        \\\hrulefill
        \State {\textbf{Deletion $\delta_\tau\calD=\{-R_j(\bm t)\}$}}\Comment{The $\tau$-th update deletes a tuple $\bm t$ from $R_j$}
        \begin{itemize}
            \item Find $([\tau', \infty], \bm t) \in \ov R_j$
            \item $\delta_\tau \ov \calD \gets \{- \ov R_j([\tau', \infty], \bm t), + \ov R_j([\tau', \tau], \bm t)\}$
            \Comment{Replace $([\tau', \infty], \bm t)$ with $([\tau', \tau], \bm t)$ in $\ov R_j$}
            \item $\ov\calD \gets \ov\calD \applydelta \delta_\tau \ov D$
            \item $\delta_\tau\wh \calD \gets \canpart_N(\delta_\tau \ov \calD)$
            \item $\wh\calD \gets \wh\calD \applydelta \delta_\tau\wh \calD$ \Comment{Proposition~\ref{prop:trunc}}
        \end{itemize}
        \\\hrulefill
        \State{\textbf{Enumeration}}
        \begin{itemize}
            \item Enumerate (using Theorem~\ref{thm:main_inserts})
            \begin{align*}
                \pi_{\vars(Q)}\bigcup_{\substack{z_1, \ldots , z_{k+1} \in \{0,1\}^*\\ z_1 \circ \cdots \circ z_{k+1} = \canpart([\tau,\tau])}} \sigma_{Z_1 = z_1, \ldots , Z_k = z_k}\wh{Q}(\wh\calD)
            \end{align*}
        \end{itemize}
    \end{algorithmic}
\end{algorithm}

For the rest of this section, we prove that Algorithm~\ref{alg:ivmpm} meets Lemma~\ref{lmm:main_fully_dynamic}.
Fix a query $Q$ together with its 
\univariate extension $\ov{Q}$ (Definition~\ref{defn:univariate}) and
its \multivariate extension $\wh{Q}$ (Eq.~\eqref{eq:multivariate}).
We split the proof of  
Lemma~\ref{lmm:main_fully_dynamic} for $\ivmpm[Q]$
into three parts. 
First, we reduce the maintenance of $Q$ to the maintenance of $\ov{Q}$
(Lemma~\ref{lem:reduction_ivmpm_ivmtilde}). 
Then, we reduce 
the latter problem to the maintenance 
of $\wh{Q}$ (Lemma~\ref{lem:reduction_ivmtilde_ivmstar}). 
Finally, we show that 
$\wh{Q}$ can be maintained with $\tildeO{N^{{\fw(\wh Q)}-1}}$ amortized update time,
where ${\fw(\wh Q)}$ is the fractional hypertree width of $\wh{Q}$
(Lemma~\ref{lem:upper_bound_maintenance_multivariate}).
Lemma~\ref{lmm:main_fully_dynamic} for $\ivmpm[Q]$ then follows from Lemmas
\ref{lem:reduction_ivmpm_ivmtilde}, \ref{lem:reduction_ivmtilde_ivmstar}, and \ref{lem:upper_bound_maintenance_multivariate}.

\subsubsection{From the Maintenance of $Q$ to the Maintenance of $\ov{Q}$}
\label{sec:q_to_uni_q}
We reduce the maintenance of $Q$ to the maintenance of $\ov{Q}$.
Let $\calD$ be the database instance for $Q$ and $\ov \calD$ be the database instance for $\ov Q$.
As mentioned before,
the idea is to use the interval variable $[Z]$
in $\ov{Q}$ to describe the ``lifespan'' of the 
tuples in the database $\calD$ for $Q$. 
Given an atom $\ov{R}([Z],\bm X)$ in $\ov{Q}$
and a tuple $\bm t$ over $\bm X$, we
introduce the following types of updates to 
relation $\ov{R}$:
\begin{itemize}[leftmargin=*]
\item $+\ov{R}([\tau,\infty],\bm t)$:
Intuitively, this single-tuple insert 
indicates that the tuple $\bm t$ lives in the database 
from time $\tau$ on.

\item $(-\ov{R}([\tau', \infty], \bm t), +\ov{R}([\tau', \tau], \bm t)$) for some $\tau' < \tau$.
This compound update is called {\em truncation} 
and replaces $([\tau',\infty],\bm t)$ in $\ov{R}$ with $([\tau',\tau],\bm t)$. 
Intuitively, it 
signals that the lifespan of tuple $\bm t$ ends at time $\tau$. 
\end{itemize}
We call an update stream $\delta_1 \ov \calD, \delta_2 \ov \calD, \ldots$ for $\ov \calD$ {\em restricted}
if each update $\delta_\tau \ov\calD$ is an insert of the form $+\ov{R}([\tau,\infty],\bm t)$ or a truncation 
of the form 
$(-\ov{R}([\tau', \infty], \bm t), +\ov{R}([\tau', \tau], \bm t))$.
The {\em restricted result} of $\ov{Q}$ after update $\delta_\tau \ov\calD$
is defined as all output tuples of $\ov Q$ whose $[Z]$-interval contains $\tau$:
\begin{align}
\label{eq:univariate_restricted_output}
\ov{Q}_\tau(\ov\calD^{(\tau)}) \defeq \sigma_{\tau \in [Z]}\ov{Q}(\ov\calD^{(\tau)}),
\end{align}
Note that by construction, we have
\begin{align}
    Q(\calD^{(\tau)}) = \pi_{\vars(Q)} \ov{Q}_\tau(\ov\calD^{(\tau)})
    \label{eq:from_ovQ_to_Q}
\end{align}
We introduce the maintenance problem for the intersection query $\ov{Q}$:

\noindent
\vspace*{.5em}
\fbox{%
    \parbox{0.96\linewidth}{%
    \begin{tabular}{ll}
   Problem: & $\ivmov[\ov{Q}]$ \\
   Parameter: & A \univariate extension query $\ov{Q}$ \\ 
        Given: & Restricted update stream
        $\delta_1 \ov \calD, \delta_2 \ov \calD, \ldots$\\
        Task: & Ensure constant-delay enumeration of
        $\ov{Q}_\tau(\ov\calD^{(\tau)})$ after each update $\delta_\tau\ov\calD$ \\
    \end{tabular}
    }}
\vspace*{.5em}

We make the reduction from $\ivmpm[Q]$ to 
$\ivmov[\ov{Q}]$ precise:
\begin{lemma}
\label{lem:reduction_ivmpm_ivmtilde}
If $\ivmov[\ov{Q}]$ can be solved with (amortized) $f(N)$ update time for any update stream of length $N$ and some function $f$,
then $\ivmpm[Q]$ can be solved with (amortized) $\bigO{f(N)}$ update time. 
\end{lemma}

\begin{proof}
The reduction from $\ivmpm[Q]$ to 
$\ivmov[\ov{Q}]$ works as follows. 
Consider an update stream 
$\delta_1\calD, \delta_2\calD, \ldots$ for $Q$.
For each $\tau \geq 1$, we pursue the following strategy. 
If $\delta_\tau\calD$ is an insert $+R(\bm t)$,
we execute the insert
$+\ov{R}([\tau,\infty],\bm t)$ on the database $\ov \calD$ of $\ov{Q}$.
If it is a delete $-R(\bm t)$, we apply the truncation 
$(-\ov{R}([\tau', \infty], \bm t), +\ov{R}([\tau', \tau], \bm t))$.
For any enumeration 
request to $Q$ that comes immediately after 
some update $\delta_\tau\calD$, 
we trigger the enumeration procedure for $\ov{Q}$,
which enumerates the tuples in 
$\ov{Q}_\tau(\ov\calD^{(\tau)})$
with constant delay. 
(Recall~\eqref{eq:from_ovQ_to_Q}.)
By construction, the tuples in $\ov{Q}_\tau(\ov\calD^{(\tau)})$ must be distinct 
with respect to their projections onto their point variables (i.e. $\vars(Q)$). In particular,
\begin{align}
    |Q(\calD^{(\tau)})| = |\ov{Q}_\tau(\ov\calD^{(\tau)})|
    \label{eq:from_ovQ_to_Q_sizes}
\end{align}
Hence, by reporting the projections of the tuples in $\ov{Q}_\tau(\ov\calD^{(\tau)})$ onto $\vars(Q)$, we enumerate the output 
$Q(\calD^{(\tau)})$ with constant delay.
\end{proof}

The following proposition will be useful later for the proof of Lemma~\ref{lmm:main_fully_dynamic}. It basically says that when we truncate
a tuple
from $\ov{R}([\tau', \infty], \bm t)$ to $\ov{R}([\tau', \tau], \bm t)$,
the truncated tuple is still going to join with all other tuples that it used to join with before the truncation.
\begin{proposition}
    \label{prop:trunc}
    Consider a truncation update
    $\delta_\tau\ov\calD =$ $\{-\ov R_j([\tau', \infty],\bm t_j),$ $+\ov R_j([\tau', \tau], \bm t_j)\}$.
    Suppose that the tuple before truncation $\ov R_j([\tau', \infty], \bm t_j)$
    used to join with a set of tuples in $\{\ov R_i([\alpha_i, \beta_i], \bm t_i)\mid i \in S\}$ for some $S \subseteq [k]$, i.e.
    \begin{align*}
        \Join_{i \in S} \ov R_i([\alpha_i, \beta_i], \bm t_i) \Join
        \ov R_j([\tau', \infty], \bm t_j) = ([\alpha, \beta], \bm t)
    \end{align*}
    Then, the following hold:
    \begin{enumerate}
        \item $\beta_i = \infty$ for all $i \in S$ and $\beta = \infty$ \label{prop:trunc:1}
        \item The truncated tuple $\ov R_j([\tau', \tau], \bm t_j)$ still joins with the same set of tuples. Namely \label{prop:trunc:2}
        \begin{align*}
            \Join_{i \in S} \ov R_i([\alpha_i, \beta_i], \bm t_i) \Join
            \ov R_j([\tau', \tau], \bm t_j) = ([\alpha, \tau], \bm t)
        \end{align*}
    \end{enumerate}
\end{proposition}
\begin{proof}
    Right before time $\tau$ when we are about to truncate, the database $\ov \calD^{(\tau-1)}$
    does not yet contain any interval $[\alpha_i, \beta_i]$ where $\alpha_i \geq \tau$.
    This is because $\tau$ is the current time, and those intervals $[\alpha_i, \beta_i]$
    where $\alpha_i \geq \tau$ start in the future. This proves part~\ref{prop:trunc:2}.
    Moreover, every interval $[\alpha_i, \beta_i]$ 
    in the database $\ov \calD^{(\tau-1)}$ where $\alpha_i < \tau \leq \beta_i$
    must satisfy $\beta_i = \infty$. This is because such an interval ends in the future, thus we could
    not have encountered its end yet. This proves part~\ref{prop:trunc:1}.
\end{proof}

\subsubsection{From the Maintenance of $\ov{Q}$ to the Maintenance of $\wh{Q}$}
\label{sec:uni_q_to_multi_q}
Now, we give a reduction from the maintenance of $\ov{Q}$
to the maintenance of $\wh{Q}$.
Let $\wh \calD$ be the database instance for $\wh Q$.
In order to translate an update $\delta_\tau \ov\calD$ for $\ov{Q}$
into an update $\delta_\tau \wh\calD$ for $\wh Q$, we take the {\em canonical partition} of $\delta_\tau\ov\calD$, as given by Definition~\ref{defn:canpart:instance}.
In particular, suppose that $N$, the length of the update stream, is a power of 2 and is known in advance.
(See Remark~\ref{rmk:doubling:N} about removing this assumption.)
The corresponding update to $\wh Q$ is
\begin{align}
    \delta_\tau\wh \calD \gets \canpart_N(\delta_\tau \ov \calD)
    \label{eq:canonical_partition_tuple_update}
\end{align}

 We say that a stream $\delta_1\wh\calD, \delta_2\wh\calD, \ldots$
 is a restricted canonical update stream for $\wh{Q}$ if there is 
 a restricted update stream $\delta_1\ov\calD, \delta_2\ov\calD, \ldots$ for $\ov{Q}$
 such that each $\delta_\tau \wh \calD$ is the canonical partition of $\delta_\tau \ov \calD$ with respect to $N$.
 The {\em restricted result} of $\wh{Q}$ after update $\delta_\tau\wh\calD$
 is defined as the set of all output tuples of $\wh{Q}$ where
 the bitstring $Z_1 \circ\cdots\circ Z_k$ (when interpreted as an interval in the segment tree $\calT_N$)
 contains $\tau$. Formally,
\begin{align}
\label{eq:restricted_canonical_result}
     \wh{Q}_{\tau}(\wh\calD^{(\tau)}) \defeq \bigcup_{\substack{z_1, \ldots , z_{k+1} \in \{0,1\}^*\\ z_1 \circ \cdots \circ z_{k+1} = \canpart([\tau,\tau])}} \sigma_{Z_1 = z_1, \ldots , Z_k = z_k}\wh{Q}(\wh\calD^{(\tau)})
\end{align}
We are ready to introduce the maintenance problem for $\wh{Q}$:

\noindent
\vspace*{.5em}
\fbox{%
    \parbox{0.96\linewidth}{%
    \begin{tabular}{ll}
   Problem: & $\ivmwh[\wh{Q}]$ \\
   Parameter: & A \multivariate extension query $\wh{Q}$ \\ 
        Given: & Restricted canonical update stream 
        $\delta_1\wh\calD, \delta_2\wh\calD, \ldots$\\
        Task: & Ensure constant-delay enumeration of
        $\wh{Q}_\tau(\wh\calD^{(\tau)})$ after each update $\delta_\tau\wh\calD$ \\
    \end{tabular}
    }}
\vspace*{.5em}

We formalize the reduction from $\ivmov[\ov{Q}]$
to $\ivmwh[\wh{Q}]$:

\begin{lemma}
\label{lem:reduction_ivmtilde_ivmstar}
If $\ivmwh[\wh{Q}]$
can be solved with (amortized) $f(N)$ update time for any update stream of length $N$
and some monotonic function $f$, then $\ivmov[\ov{Q}]$ can be solved with
$\tildeO{f(N)}$ amortized update time.
\end{lemma}

\begin{proof}
The reduction expressed in Lemma~\ref{lem:reduction_ivmtilde_ivmstar}
works as follows.
Consider an update stream $\delta_1\ov\calD,\delta_2\ov\calD \ldots$
for $\ov{Q}$. We transform each update $\delta_\tau\ov\calD$ into its canonical partition 
$\delta_\tau\wh\calD$ with respect to $N$, and apply 
$\delta_\tau\wh\calD$ on the database $\wh \calD$
for $\wh{Q}$.
As a result, the algorithm maintains the invariant that the database $\wh \calD$ is always
the canonical partition of the database $\ov \calD$ with respect to $N$. Formally,
\begin{align*}
    \wh \calD^{(\tau)} = \canpart_N(\ov \calD^{(\tau)}),\quad\quad\text{for all $\tau \in [N]$}.
\end{align*}
Moreover, by Definition~\ref{defn:canpart:tuple}, the number of single-tuple updates in $\delta_\tau\wh\calD$ is at most polylogarithmic in $N$ (which is $\tildeO{1}$).

Theorem~\ref{thm:IJ:forward:reduction} connects the full outputs of $\ov Q(\ov \calD^{(\tau)})$ and $\wh Q(\wh\calD^{(\tau)})$ using the following relation:
\begin{align}
    \canpart_{N}^{(1)}(\ov Q(\ov \calD^{(\tau)})) = G_k(\wh Q(\wh \calD^{(\tau)}))
\end{align}
From the above equation along with~\eqref{eq:restricted_canonical_result} and~\eqref{eq:univariate_restricted_output}, we conclude that:
\begin{align}
    \pi_{\vars(Q)}\wh{Q}_{\tau}(\wh\calD^{(\tau)}) &=
    \pi_{\vars(Q)}\ov{Q}_\tau(\ov\calD^{(\tau)})\label{eq:from_whQ_to_ovQ}\\
    |\wh{Q}_{\tau}(\wh\calD^{(\tau)})| &=
    |\ov{Q}_\tau(\ov\calD^{(\tau)})|\label{eq:from_whQ_to_ovQ_sizes}
\end{align}
Equation~\eqref{eq:from_whQ_to_ovQ_sizes} follows from the fact that for any segment
$[Z]$ where $\tau \in [Z]$, there is exactly one segment in the canonical partition of $[Z]$ that contains $\tau$.
Hence, we can use constant-delay enumeration for
$\wh{Q}_{\tau}(\wh\calD^{(\tau)})$
in order to enumerate the output of
$\ov{Q}_\tau(\ov\calD^{(\tau)})$.

Finally, if the length of the update stream $N$ is not known in advance, then we use the doubling trick described in Remark~\ref{rmk:doubling:N}.
In particular, we initialize $N$ to 1.
Before we process $\tau$-th update for $\tau = N+1$,
we replace the segment tree $\calT_N$ by $\calT_{2N}$ and replace $N$ by $2N$.
Replacing the segment tree corresponds to creating a copy of the old tree and adding a new root
node whose two children are the old segment tree and the copy.
After updating the segment tree, we can redo the previous inserts into the new tree.
The total update time would be
\begin{align*}
    \tildeO{1}\sum_{i \in [\lceil\log_2 N\rceil]} 2^i \cdot f(2^i) = \tildeO{N\cdot f(N)}
\end{align*}
(Recall that $f$ is monotonic.)
Hence, the amortized update time is still $\tildeO{f(N)}$.
\end{proof}

\subsubsection{Update Time for $\wh{Q}$}
\label{sec:upper_bound_maintenance_multivariate}
Note that~\eqref{eq:from_ovQ_to_Q},
~\eqref{eq:from_ovQ_to_Q_sizes},~\eqref{eq:from_whQ_to_ovQ}, and~\eqref{eq:from_whQ_to_ovQ_sizes} imply:
\begin{align*}
    Q(\calD^{(\tau)}) &= \pi_{\vars(Q)} \wh{Q}_\tau(\wh\calD^{(\tau)})\\
    |Q(\calD^{(\tau)})| &= |\wh{Q}_\tau(\wh\calD^{(\tau)})|
\end{align*}
This basically means that we can do constant-delay enumeration for $Q(\calD^{(\tau)})$
using constant-delay enumeration for $\wh{Q}_\tau(\wh\calD^{(\tau)})$.
Now, we prove the following statement:

\begin{lemma}
\label{lem:upper_bound_maintenance_multivariate}
For any restricted canonical update stream of length $N$, the problem 
$\ivmwh[\wh{Q}]$ can be solved with $\tildeO{N^{{\fw(\wh Q)}-1}}$
amortized update time, where ${\fw(\wh Q)}$ is the fractional hypertree width of $\wh{Q}$. 
\end{lemma}
Before we prove the above lemma, we need some preliminaries:
\begin{definition}
    A subset $\bm Y \subseteq \vars(\wh Q)$ is called {\em $Z$-prefix-closed} if it satisfies the following
    property: For every $Z_i \in \bm Y$, we must have $\{Z_1, \ldots, Z_{i-1}\} \subseteq \bm Y$.
    An atom $R(\bm Y)$ is {\em $Z$-prefix-closed} if its schema $\bm Y$ is $Z$-prefix-closed.
    \label{defn:prefix-closed}
\end{definition}
Recall from Appendix~\ref{app:prelims} that a tree decomposition $(\calT, \chi)$ of a query $Q$ is called {\em optimal} if it minimizes the fractional
hypertree width of $Q$.
\begin{proposition}
    Let $\wh Q_{\bm\sigma}$ be a component of $\wh Q$.
    There must exist an optimal tree decomposition $(\calT, \chi)$ for $\wh Q_{\bm\sigma}$
    where every bag of $(\calT, \chi)$ is $Z$-prefix-closed.
    \label{prop:prefix:td}
\end{proposition}

\begin{proof}[Proof of Proposition~\ref{prop:prefix:td}]
    The proof is based on the equivalence of tree decompositions and {\em variable elimination orders}.  See Appendix~\ref{app:prelims} for some background.
    In particular, the proof works as follows.
    Pick some {\em optimal} variable elimination order $\bm\mu$. We inductively prove the following:
    \begin{claim}
        Throughout variable elimination, every atom is $Z$-prefix-closed.
    \end{claim}
    \begin{claim}
        For every variable $X\in\bm\mu$, the set $\bm U_X^{\bm\mu}$ is prefix-closed.
    \end{claim}
    Initially, both claims above hold. Now, suppose we pick a variable $X$ to eliminate.
    By induction, all atoms that contain $X$ are $Z$-prefix-closed, hence $\bm U_X^{\bm\mu}$ is $Z$-prefix-closed.
    Now in order to show that the newly added atom $R_X^{\bm\mu}(\bm U_X^{\bm\mu}-\{X\})$ can be assumed to be $Z$-prefix closed, we recognize two cases.
    If $X$ is not in $\{Z_1, \ldots, Z_k\}$, then $\bm U_X^{\bm\mu}-\{X\}$ is obviously $Z$-prefix-closed. On the other hand, suppose $X = Z_i$ for some $i \in [k]$.
    Let $j$ be the largest index where $Z_j \in \bm U_X^{\bm\mu}$.
    If $Z_j = Z_i$, then $\bm U_X^{\bm\mu}-\{X\}$ is $Z$-prefix closed.
    Otherwise, $\bm U_X^{\bm\mu}$ must contain $Z_i, Z_{i+1}, \ldots, Z_{j}$. Moreover, none of these
    variables can occur in any atom other than $R_X^{\bm\mu}(\bm U_X^{\bm\mu}-\{X\})$. Hence, there is an optimal variable
    elimination order $\bm\mu'$ that eliminates them in any order. We choose the order $Z_j, Z_{j-1}, \ldots, Z_i$.
\end{proof}

\begin{proof}[Proof of Lemma~\ref{lem:upper_bound_maintenance_multivariate}]
The proof of Lemma~\ref{lem:upper_bound_maintenance_multivariate}
relies on the proof of Theorem~\ref{thm:main_inserts} that is given in Section~\ref{sec:insert_only:ub:full}. 
Consider an optimal tree decomposition $(\calT, \chi)$ for $\wh{Q}$, i.e.~where
$\fw(\calT,\chi) = \fw(\wh Q)$.
Let $\delta_1 \wh\calD, \delta_2 \wh\calD, \ldots$ 
be a restricted canonical update stream for $\wh{Q}$.
Recall that each $\delta_\tau \wh\calD$ is the canonical partition of a corresponding update $\delta_\tau \ov\calD$ to $\ov \calD$.
Consider an update
$\delta_\tau\wh\calD = \canpart_N(\delta_\tau\ov\calD)$ in this stream.
First, assume that 
$\delta_\tau\ov\calD$ is an insert, i.e.~$\delta_\tau\ov\calD=\{+\ov R_j([\tau, \infty], \bm t_j)\}$ for some $j \in [k]$ and tuple $\bm t_j$ with schema $\bm X_j$.
In this case, the corresponding 
$\delta_\tau\wh\calD = \canpart_N(\delta_\tau\ov\calD)$ is a set of single-tuple inserts,
and we can apply these inserts to the materialized bags of the tree decomposition $(\calT, \chi)$,  just  
as explained in the proof of Theorem~\ref{thm:main_inserts}.
Truncation is more tricky.
In particular, suppose that $\delta_\tau\ov\calD$ is a truncation,
i.e.~$\delta_\tau\ov\calD = \{- \ov R_j([\tau', \infty], \bm t_j), + \ov R_j([\tau', \tau], \bm t_j)\}$ for some $\tau' < \tau$, $j \in [m]$ and tuple $\bm t_j$ over schema $\bm X_j$.
In this case, we update the materialized bags of $(\calT, \chi)$
using the updates given by $\canpart_N(\delta_\tau\ov\calD)$.
Some of these updates are inserts, namely $\canpart_N(\{+\ov R_j([\tau', \tau], \bm t_j)\})$,
and these are accounted for by the analysis
from the proof of Theorem~\ref{thm:main_inserts}.
In particular, they are covered by the $\tildeO{N^{{\fw(\wh Q)}}}$ total update time over all inserts
combined.
However, some of the updates in $\canpart_N(\delta_\tau\ov\calD)$ are deletes,
namely $\canpart_N(\{-\ov R_j([\tau', \infty], \bm t_j)\})$,
and in order to account for those, we will use amortized analysis.
In particular, we will deposit extra credits with every insert so that we can use them later to pay for the deletes during truncation.
We explain this argument in detail below.

We start with the following claim, which is very similar in spirit to Proposition~\ref{prop:trunc}:
\begin{claim}
    \label{clm:insert_indistinguishable}
Let $\delta_\tau \ov \calD = \{+\ov R_j([\tau, \infty], \bm t_j)\}$ be an insert into $\ov\calD$, for some $j \in [m]$ and tuple $\bm t_j$ over the schema $\bm X_j$.
For some $i \in [k] -\{j\}$, let $([\alpha, \beta], \bm t_i)$ be a tuple of another relation $\ov R_i$
that joins with the inserted tuple $([\tau, \infty], \bm t_j)$ of $\ov R_j$.
Then, we must have $\alpha < \tau$ and $\beta = \infty$.
Namely, the interval $[\alpha, \beta]$ strictly contains $[\tau, \infty]$.
\end{claim}
\begin{proof}[Proof of Claim~\ref{clm:insert_indistinguishable}]
The above claim follows from the fact that $\tau$ is the current time, and any interval $[\alpha, \beta]$ where $\alpha \geq \tau$ would start in the future, hence it cannot exist yet in $\ov\calD$.
Moreover, in order for $[\alpha, \beta]$ to overlap with $[\tau, \infty]$, it must be the case
that $\beta \geq \tau$, which implies that $\beta = \infty$ because it is in the future.
This proves the claim.
\end{proof}

For the purpose of analyzing the algorithm, we distinguish between two types of tuples in $\ov \calD$: {\em finalized} tuples and {\em provisional} tuples. {\em Finalized} tuples cannot be deleted once they are inserted, while {\em provisional} tuples can be deleted later.
Also for the purpose of analyzing the algorithm, we slightly modify the way the algorithm
handles inserts. In particular, whenever we have an insert 
$+\ov R_j([\tau, \infty], \bm t_j)$ into $\ov\calD$,
we insert this tuple as {\em provisional} (since it might be deleted later
due to truncation) and we insert an additional tuple $([\tau, \tau], \bm t_j)$ into $\ov R_j$ that is {\em finalized}.
In particular, we know for sure that the tuple $([\tau, \tau], \bm t_j)$ will never be deleted later because the lifespan of the tuple $\bm t_j$ extends at least
until the current time $\tau$.
We also update $\wh\calD$ by applying the two corresponding inserts, namely $\canpart_N(\{+\ov R_j([\tau, \infty], \bm t_j)\})$
and $\canpart_N(\{+\ov R_j([\tau, \tau], \bm t_j)\})$.
Tuples in $\wh\calD$ inherit the provisional/finalized property from the underlying tuples in $\ov\calD$.
We make the following claim:

\begin{claim}
    \label{clm:provisional_finalized}
    The total time needed to apply the insert $\canpart_N(\{+\ov R_j([\tau, \infty], \bm t_j)\})$
    into $\wh \calD$ is within a $\tildeO{1}$ factor of the total time needed to apply the insert
    $\canpart_N(\{+\ov R_j([\tau, \tau], \bm t_j)\})$.
\end{claim}
The above claim follows from Claim~\ref{clm:insert_indistinguishable}
along with the fact that every bag in the tree decomposition $(\calT, \chi)$
is $Z$-prefix closed (Proposition~\ref{prop:prefix:td}).
In particular, Claim~\ref{clm:insert_indistinguishable} basically says that the two intervals
$[\tau, \infty]$ and $[\tau, \tau]$ are ``indistinguishable'' as far as joins are concerned.
Proposition~\ref{prop:prefix:td} ensures that whenever we take a projection in the query plan corresponding to the tree decomposition
$(\calT, \chi)$, we are basically replacing these two intervals
with two new intervals that {\em contain} them.
Hence, the new intervals remain indistinguishable.

Now, we can analyze the total update time
as follows. Claim~\ref{clm:provisional_finalized} implies that the total
time needed to apply provisional inserts is upper bounded by the total time needed
to apply finalized inserts. Hence, the overall runtime is dominated by finalized inserts.
Moreover, every time we apply a provisional insert, we can deposit one credit
with every tuple that we insert into any bag of $(\calT, \chi)$, in order to pay for the potential delete
of this tuple in the future.
These credits cover the cost of all deletes of provisional tuples.
(Recall that finalized tuples cannot be deleted.)
Finally, the total time needed to apply all finalized inserts is bounded by
$\tildeO{N^{\fw(\wh Q)}}$, by Theorem~\ref{thm:main_inserts}.

We are left with explaining how to support constant-delay enumeration of $\wh{Q}_{\tau}(\wh\calD^{(\tau)})$.
To that end, we enumerate selections $\sigma_{Z_1 = z_1, \ldots , Z_k = z_k}\wh{Q}(\wh\calD^{(\tau)})$ for various assignments $Z_1 = z_1,$ $\ldots ,$ $Z_k = z_k$, as defined by~\eqref{eq:restricted_canonical_result}.
Luckily, for every component $\wh Q_{\bm\sigma}$ of $\wh Q$, there is one input relation
$\wh R_{\sigma_k}$ that contains the variables $\{Z_1, \ldots, Z_k\}$; see~\eqref{eq:multivariate_component}. Hence, every tree decomposition of $\wh Q_{\bm\sigma}$
must contain a bag $B$ that contains all the variables $Z_1, \ldots , Z_k$.
In order to enumerate a selection $\sigma_{Z_1 = z_1, \ldots , Z_k = z_k}\wh{Q}(\wh\calD^{(\tau)})$ with constant delay, we can use the bag $B$ as the root of the tree decomposition. 
\end{proof}

\subsection{Upper Bound for $\ivmpmd[Q]$}
\label{sec:forward_reduction:delta_version}
\change{
In the previous section, we proved Lemma~\ref{lmm:main_fully_dynamic} for the 
$\ivmpm[Q]$ problem for a given query $Q$. In this section, we extend the proof and the corresponding algorithm from the previous section in order to prove Lemma~\ref{lmm:main_fully_dynamic} for
the $\ivmpmd[Q]$ problem.}

Here are the main changes to Algorithm~\ref{alg:ivmpm} in order to support $\ivmpmd[Q]$, i.e.~support constant-delay enumeration of $\delta_\tau Q(\calD)$ after every update $\delta_\tau \calD$:
\begin{itemize}
    \item \textbf{Insertion} $\delta_\tau\calD = \{+R_j(\bm t)\}$:
    In addition to applying the insert $\{+\ov R_j([\tau, \infty], \bm t)\}$ into $\ov\calD$,
    we now also insert $+\ov R_j([\tau, \tau], \bm t)$.
    \item \textbf{Deletion} $\delta_\tau\calD=\{-R_j(\bm t)\}$:
    In addition to applying the update $\{- \ov R_j([\tau', \infty], \bm t), + \ov R_j([\tau', \tau], \bm t)\}$ to $\ov\calD$, we now also apply the insert
    $+\ov R_j([\tau, \tau], \bm t)$.
    \item \textbf{Enumeration} Suppose we want to enumerate $\delta_\tau Q(\calD)$
    after an update $\delta_\tau \calD = s R_j(\bm t)$ where $s \in \{+, -\}$. Then, we enumerate
    the following:
    \begin{align}
            \delta_\tau Q(\calD) =
            s
            \left(\pi_{\vars(Q)}\bigcup_{\substack{z_1, \ldots , z_{k} \in \{0,1\}^*\\ z_1 \circ \cdots \circ z_{k} = \canpart([\tau,\tau])}} \sigma_{Z_1 = z_1, \ldots , Z_k = z_k}\wh{Q}(\wh\calD)\right)
            \label{eq:delta_version:enumeration}
    \end{align}
\end{itemize}
We call the above modified version of Algorithm~\ref{alg:ivmpm} the {\em delta version} of Algorithm~\ref{alg:ivmpm}.
Before we prove that the delta version of Algorithm~\ref{alg:ivmpm} proves Lemma~\ref{lmm:main_fully_dynamic} for $\ivmpmd[Q]$, we need some preliminaries.
\begin{proposition}
    Given a discrete interval $[\alpha, \beta]$ for some $\alpha, \beta \in \mathbb{N}$,
    let $N$ be a power of 2 such that $N \geq \beta$.
    If $\canpart_N([\alpha, \beta])$ contains a singleton interval $[x, x]$ for some $x \in \mathbb{N}$, then $x$ is an endpoint of $[\alpha, \beta]$. In particular, either $x = \alpha$ or $x = \beta$.
    \label{prop:canpart:singleton}
\end{proposition}
Proposition~\ref{prop:canpart:singleton} follows from the definition of a segment tree $\calT_N$ and canonical partition $\canpart_N$. For example, consider the segment tree $\calT_8$ depicted in Figure~\ref{fig:segment-tree}. Let's take the segment $[2, 5]$. Its canonical partition
consists of $\{[2, 2], [3, 4], [5, 5]\}$. Note that only $[2, 2]$ and $[5, 5]$ are singleton
intervals. But $2$ and $5$ are the endpoints of the interval $[2, 5]$. The above proposition
says that we cannot find in the canonical partition of any segment $[\alpha, \beta]$ any singleton intervals that are not end points of $[\alpha, \beta]$.

The following theorem implies Lemma~\ref{lmm:main_fully_dynamic} for $\ivmpmd[Q]$:
\begin{theorem}
    The delta version of Algorithm~\ref{alg:ivmpm}
    solves the $\ivmpmd[Q]$ with amortized update time $\tildeO{N^{{\fw(\wh Q)} - 1}}$
    where ${\fw(\wh Q)}$ is the fractional hypertree width of $\wh{Q}$.
\end{theorem}
\begin{proof}
    The proof is very similar to the proof of Lemma~\ref{lmm:main_fully_dynamic}
    from Section~\ref{sec:forward_reduction} for $\ivmpm[Q]$.
    For every insert or delete $\delta_\tau\calD = \{s R_j(\bm t)\}$ where $s \in \{+, -\}$, we now have to insert
    an extra tuple $+\ov R_j([\tau, \tau], \bm t)$ into $\ov\calD$. This only increases the size of 
    $\ov\calD$ by a constant factor. Hence, it doesn't affect the overall runtime.

    Next, we prove correctness. In particular, our target is to prove that~\eqref{eq:delta_version:enumeration} does indeed compute $\delta_\tau Q(\calD)$, as defined
    by~\eqref{eq:def:deltaQ}. To that end, we define the following:
    \begin{align}
        \label{eq:univariate_restricted_output:delta}
\ov{Q}_{\delta_\tau}(\ov\calD^{(\tau)}) &\defeq \sigma_{[Z] = [\tau, \tau]}\ov{Q}(\ov\calD^{(\tau)})\\
        \label{eq:restricted_canonical_result:delta}
                \wh{Q}_{\delta_\tau}(\wh\calD^{(\tau)}) &\defeq \bigcup_{\substack{z_1, \ldots , z_{k} \in \{0,1\}^*\\ z_1 \circ \cdots \circ z_{k} = \canpart([\tau,\tau])}} \sigma_{Z_1 = z_1, \ldots , Z_k = z_k}\wh{Q}(\wh\calD^{(\tau)})
    \end{align}
    Compare the above to~\eqref{eq:univariate_restricted_output} and~\eqref{eq:restricted_canonical_result} respectively.
    The correctness proof goes through two steps. First, we prove that
    \begin{align*}
        \delta_\tau Q(\calD) = s\left(\pi_{\vars(Q)}\ov Q_{\delta_\tau}(\ov \calD^{(\tau)})\right)
    \end{align*}
    where recall that $s \in\{+, -\}$ depending on the update $\delta_\tau \calD = \{s R_j(\bm t)\}$. This follows from the fact that $\tau$ is the unique timestamp for the update 
    $\delta_\tau \calD$. In particular, all corresponding updates $\ov s \ov R_j([\alpha, \beta], \bm t)$ in $\delta_\tau\ov\calD$ (where $\ov s\in\{+, -\}$) must have $\tau \in [\alpha, \beta]$. Moreover, no other tuple
    $([\alpha', \beta'], \bm t')$ in $\ov R_j$ can have $\tau \in \{\alpha, \beta\}$ (because the timestamp $\tau$ is unique to this update).
    
    The second step of the correctness proof is to prove that
    \begin{align}
        \pi_{\vars(Q)} \ov{Q}_{\delta_\tau}(\ov\calD^{(\tau)}) = 
        \pi_{\vars(Q)} \wh{Q}_{\delta_\tau}(\wh\calD^{(\tau)})
        \label{eq:from_ovQ_to_whQ:delta}
    \end{align}
    The containment $\subseteq$ in~\eqref{eq:from_ovQ_to_whQ:delta} holds
    thanks to the extra intervals $+\ov R_j([\tau, \tau], \bm t)$ that we are inserting
    to $\ov\calD$ in the delta version of Algorithm~\ref{alg:ivmpm}.
    The other containment $\supseteq$ in~\eqref{eq:from_ovQ_to_whQ:delta}
    follows from Proposition~\ref{prop:canpart:singleton}.
    In particular,~\eqref{eq:restricted_canonical_result:delta} only selects the leaf corresponding to $[\tau, \tau]$ in
    the segment tree, but Proposition~\ref{prop:canpart:singleton} tells us that this
    leaf could not belong to the canonical partition of any interval $[\alpha, \beta]$ unless
    $\tau \in \{\alpha, \beta\}$. Such intervals could only have been added to $\ov\calD$
    during the update $\delta_\tau\calD$ because the timestamp $\tau$ is unique to this update.

    Finally, we show that constant-delay enumeration for $\wh{Q}_{\delta_\tau}(\wh\calD^{(\tau)})$
    implies constant-delay enumeration for $\delta_\tau Q(\calD)$. To that end, we prove:
    \begin{align*}
        |\wh{Q}_{\delta_\tau}(\wh\calD^{(\tau)})| = |\delta_\tau Q(\calD)|
    \end{align*}
    We also prove that in two steps, namely
    \begin{align}
        |\ov{Q}_{\delta_\tau}(\ov\calD^{(\tau)})| &= |\delta_\tau Q(\calD)|\label{eq:delta_enumeration_1}\\
        |\wh{Q}_{\delta_\tau}(\wh\calD^{(\tau)})| &= |\ov Q_{\delta_\tau }(\ov\calD^{(\tau)})|\label{eq:delta_enumeration_2}
    \end{align}
    Equation~\eqref{eq:delta_enumeration_1} follows from~\eqref{eq:univariate_restricted_output:delta} along with the definition of $\delta_\tau Q(\calD)$.
    Equation~\eqref{eq:delta_enumeration_2} follows from the fact that the delta version
    of Algorithm~\ref{alg:ivmpm} ensures the following:
    for every timestamp $\tau$, there can be at most two tuples in $\ov\calD$ where $\tau$
    appears as an end point of an interval. The two tuples must have the form
    $([\tau, \tau], \bm t)$ and $([\alpha, \beta], \bm t)$ for the {\em same} $\bm t$
    where either $\alpha = \tau$ or $\beta = \tau$.
    When we take the canonical partition, both tuples are going to produce the {\em same}
    tuple in $\wh{Q}_{\delta_\tau}(\wh\calD^{(\tau)})$.
\end{proof}

\subsection{Upper bound in terms of $|\calD|$ instead of $N$}
\label{sec:forward_reduction:calD}
\change{
The upper bound in Lemma~\ref{lmm:main_fully_dynamic} is written in terms of the length $N$ of the update stream. In this section, we show how to use Lemma~\ref{lmm:main_fully_dynamic} as a black box in order to prove Theorem~\ref{thm:main_fully_dynamic}, which gives a stronger bound
where $N$ is replaced by the current database size $|\calD^{(\tau)}|$ at the $\tau$-th update.
}
Note that $|\calD^{(\tau)}| \leq N$ and $|\calD^{(\tau)}|$ could be unboundedly smaller than $N$ (specifically
when the stream contains many deletes before the timestamp $\tau$).
\begin{citedthm}[\ref{thm:main_fully_dynamic}]
    For any query $Q$,
    both $\ivmpm[Q]$ and $\ivmpmd[Q]$ can be solved with $\tildeO{|\calD^{(\tau)}|^{\fw(\wh Q)-1}}$ amortized update time \change{and non-amortized constant enumeration delay}, where $\wh Q$ is the \multivariate extension of $Q$, and $|\calD^{(\tau)}|$ is the current database size at update time $\tau$.
\end{citedthm}
Having an amortized update time of $\tildeO{|\calD^{(\tau)}|^{\fw(\wh Q)-1}}$ means that all $N$ updates
combined take time \linebreak $\tildeO{\sum_{\tau \in [N]}|\calD^{(\tau)}|^{\fw(\wh Q)-1}}$.
\begin{proof}
The proof uses Lemma~\ref{lmm:main_fully_dynamic} as a black box. In particular, we will use the 
same maintenance algorithm from the proof of Lemma~\ref{lmm:main_fully_dynamic}.
However, every now and then, we will ``reset'' the algorithm by taking the current database $\calD^{(\tau)}$,
creating a new empty database from scratch, and inserting every tuple in $\calD^{(\tau)}$ into the new database as a single-tuple update.
After that, we process the next $\frac{1}{2}|\calD^{(\tau)}|$ updates\footnote{For simplicity, we assume that
$|\calD^{(\tau)}|$ is a positive even integer.} from the original update sequence,
and then we apply another reset.
The total number of updates between the last two resets is thus $\frac{3}{2}|\calD^{(\tau)}|$.
By Lemma~\ref{lmm:main_fully_dynamic}, these $\frac{3}{2}|\calD^{(\tau)}|$ updates together
take time $\tildeO{|\calD^{(\tau)}|^{\fw(\wh Q)}}$.
Moreover, note that between the two resets, the database size cannot change by more than a factor of 2.
This is because we only have $\frac{1}{2}|\calD^{(\tau)}|$ updates from the original stream that
are applied to $\calD^{(\tau)}$.
In the two extreme cases, all of these updates
are inserts or all of them are deletes.

Consider the original update stream of length $N$. Initially, $|\calD^{(0)}| = 0$.
While using the above scheme to process the stream, suppose that we end up performing $k$
resets in total at time stamps $\tau_1 \defeq 0 < \tau_2 < \cdots < \tau_k \leq N$.
WLOG assume that $\tau_k = N$.
From the above construction, for any $i \in [k-1]$, we have the following:
\begin{align}
    \tau_{i+1}-\tau_i &= \frac{1}{2}|\calD^{(\tau_i)}|\label{eq:thm:main_fully_dynamic:1}\\
    \frac{1}{2}|\calD^{(\tau_i)}|&\leq |\calD^{(\tau)}| \leq\frac{3}{2}|\calD^{(\tau_i)}|, \quad \forall \tau \in [\tau_i, \tau_{i+1})
    \label{eq:thm:main_fully_dynamic:2}
\end{align}
Eq.~\eqref{eq:thm:main_fully_dynamic:2} holds because between $\tau_i$ and $\tau_{i+1}$,
we are making $\frac{1}{2}|\calD^{(\tau_i)}|$ updates to $\calD^{(\tau_i)}$, which in the two extreme cases
are either all inserts or all deletes.
As noted before, from the beginning of the $\tau_i$-reset until the beginning of the $\tau_{i+1}$-reset,
we process $\frac{3}{2}|\calD^{(\tau_i)}|$ updates, which together take time $\tildeO{|\calD^{(\tau_i)}|^{\fw(\wh Q)}}$, thanks to Lemma~\ref{lmm:main_fully_dynamic}.
To bound the total runtime, we sum up over $i \in [k-1]$:
(We don't actually need to perform the last reset at time $\tau_k = N$ since there are no subsequent updates anyhow.)
\begin{align}
    \sum_{i \in [k-1]}|\calD^{(\tau_i)}|^{\fw(\wh Q)}
    =&\sum_{i \in [k-1]}\sum_{\tau \in [\tau_i, \tau_{i+1})}
    \frac{1}{\tau_{i+1}-\tau_i}\cdot
    |\calD^{(\tau_i)}|^{\fw(\wh Q)}\\
    =&2\cdot\sum_{i \in [k-1]}\sum_{\tau \in [\tau_i, \tau_{i+1})}
    |\calD^{(\tau_i)}|^{\fw(\wh Q)-1}\label{eq:thm:main_fully_dynamic:3}\\
    \leq&2^{\fw(\wh Q)}\cdot\sum_{i \in [k-1]}\sum_{\tau \in [\tau_i, \tau_{i+1})}
    |\calD^{(\tau)}|^{\fw(\wh Q)-1}\label{eq:thm:main_fully_dynamic:4}\\
    =&\bigO{1}\cdot \sum_{\tau \in [N]}|\calD^{(\tau)}|^{\fw(\wh Q)-1}
\end{align}
Equality~\eqref{eq:thm:main_fully_dynamic:3} follows from Eq.~\eqref{eq:thm:main_fully_dynamic:1},
while inequality~\eqref{eq:thm:main_fully_dynamic:4} follows from Eq.~\eqref{eq:thm:main_fully_dynamic:2}.
In particular, the total time needed to process all $N$ updates in the stream is
$\tildeO{\sum_{\tau \in [N]}|\calD^{(\tau)}|^{\fw(\wh Q)-1}}$.
This implies that the amortized time per update is $\tildeO{|\calD^{(\tau)}|^{\fw(\wh Q)-1}}$, as desired.
\end{proof}

\subsection{Lower Bound for $\ivmpmd[Q]$}
\label{sec:backward_reduction}

The goal of this section is to prove Theorem~\ref{thm:fully_dynamic_lower_bound},
which gives a lower bound on the $\ivmpmd$ update time of a query $Q$ in terms of a lower bound on the
static evaluation time for any component $\wh Q_{\bm\sigma}$ in its \multivariate extension.
\change{
Specifically, we will prove the following stronger lower bound:
\begin{lemma}
    \label{lmm:fully_dynamic_lower_bound}
    Let $Q$ be a query and $\wh Q_{\bm \sigma}$ any component of its \multivariate extension. For any constant $\gamma > 0$, $\ivmpmd[Q]$ cannot be solved with amortized update time
    \change{$\tildeO{N^{\lb(\wh Q_{\bm\sigma})-1-\gamma}}$}.
\end{lemma}
The above lemma immediately implies Theorem~\ref{thm:fully_dynamic_lower_bound}
since the number of single-tuple updates $N$ is larger than the size of the database $|\calD|$, and could be unboundedly larger.
}
In order to prove Lemma~\ref{lmm:fully_dynamic_lower_bound}, we give a reduction from the static evaluation of $\wh Q_{\bm\sigma}$ to the $\ivmpmd$ problem for
$Q$. Example~\ref{ex:backward:reduction} introduces the main idea behind this reduction.

For the rest of this section,
we fix a query $Q$, its 
\univariate extension $\ov{Q}$ (Definition~\ref{defn:univariate}), and
one component $\wh Q_{\bm\sigma}$ in its \multivariate extension $\wh{Q}$ (Definition~\ref{defn:multivariate}).
Moreover, we fix an arbitrary database instance $\wh \calD_{\bm\sigma}$ for $\wh Q_{\bm\sigma}$.
Algorithm~\ref{alg:backward_reduction} describes our reduction from $\eval[\wh Q_{\bm \sigma}]$ to $\ivmpmd[Q]$. Theorem~\ref{thm:IJ:backward:reduction} already gives a reduction
from $\eval[\wh Q_{\bm \sigma}]$ to the (static) evaluation of $\ov Q$ over some  database instance
$\ov \calD\defeq \iv(\wh\calD_{\bm\sigma})$ that is given by Definition~\ref{defn:interval-version} and satisfies the following properties:
\begin{align*}
    |\ov \calD| &= |\wh \calD_{\bm\sigma}|\\
    H_k(\ov Q(\ov \calD)) &= \wh Q_{\bm\sigma}(\wh \calD_{\bm\sigma})\\
    |\ov Q(\ov \calD)| &= |\wh Q_{\bm\sigma}(\wh \calD_{\bm\sigma})|
\end{align*}
(Recall the definition of $H_k$ from Theorem~\ref{thm:IJ:backward:reduction}.)
The above properties follow from Proposition~\ref{prop:iv_same_input_size}
and Theorem~\ref{thm:IJ:backward:reduction}, which also imply that
$\ov \calD$ can be constructed in linear time in $|\wh \calD_{\bm\sigma}|$,
and $H_k(\ov Q(\ov \calD))$ can be computed in linear time in $|\ov Q(\ov \calD)| = |\wh Q_{\bm\sigma}(\wh \calD_{\bm\sigma})|$.
It remains to show how to reduce the (static) evaluation of $\ov Q$ over $\ov\calD$ into $\ivmpmd[Q]$. This is shown by the following Lemma.

\begin{algorithm}[th!]
    \caption{Reduction from $\eval[\wh Q_{\bm \sigma}]$ to $\ivmpmd[Q]$}
    \label{alg:backward_reduction}
    \begin{algorithmic}
        \State {\textbf{Inputs}}
        \begin{itemize}
            \item A query $\wh Q_{\bm \sigma}$ \Comment{Equation~\eqref{eq:multivariate_component}}
            \item An {\em arbitrary} database instance $\wh \calD_{\sigma}$ for $\wh Q_{\bm \sigma}$
        \end{itemize}
        \\\hrulefill
        \State {\textbf{Algorithm}}
        \begin{itemize}
            \item Construct $\ov Q$ \Comment{Equation~\eqref{eq:univariate}}
            \item $\ov \calD \gets \iv(\wh \calD_{\bm \sigma})$
            \Comment{Definition~\ref{defn:interval-version}}
            \item Initialize $\calD \defeq \left(R_j\right)_{j\in[k]}$ to be empty
            \Comment{Database for $Q$}
            \item \Comment{Sort tuples in $\ov \calD$ by their end points}\begin{align*}L \gets &\{(\alpha, ([\alpha, \beta], \bm t_j)) \mid
                ([\alpha, \beta], \bm t_j) \in \ov R_j\text{ for some $j \in [k]$}\}\\
                \cup
                &\{(\beta, ([\alpha, \beta], \bm t_j)) \mid
                ([\alpha, \beta], \bm t_j) \in \ov R_j\text{ for some $j \in [k]$}\}
            \end{align*}
            \item $\Delta \gets \emptyset$
            \item For each $(\tau, ([\alpha, \beta], \bm t_j)) \in L$ ordered by $\tau$:
            \begin{itemize}
                \item if $\tau = \alpha$,
                \begin{itemize}[label={}]
                    \item insert $\bm t_j$ into $R_j$
                \end{itemize}
                \item Otherwise,
                \begin{itemize}[label={}]
                    \item delete $\bm t_j$ from $R_j$
                \end{itemize}
                \item $\Delta \gets \Delta \cup (\{\tau\}\times\delta_\tau Q(\calD))$\Comment{Invoke constant-delay enumeration from $\ivmpmd[Q]$ to get $\delta_\tau Q(\calD)$}
            \end{itemize}
            \item $\ov Q(\ov \calD)\gets\{([\tau, \tau'], \bm t) \mid (\tau, +Q(\bm t)) \in \Delta
            \wedge (\tau', -Q(\bm t)) \in \Delta\}$
            \item \Return $\wh Q_{\bm\sigma}(\wh \calD_{\bm\sigma}) = H_k(\ov Q(\ov \calD))$
            \Comment{Theorem~\ref{thm:IJ:backward:reduction}}
        \end{itemize}
    \end{algorithmic}
\end{algorithm}

\begin{lemma}
    \label{lem:eval_univariate_to_maintain_join}
    If $\ivmpmd[Q]$ can be solved with $f(N)$ amortized update time 
    for any update stream of length $N$ and some function $f$, 
    then $\eval[\ov{Q}]$ can be solved in 
    $\tildeO{|\ov\calD|\cdot f(2|\ov\calD|) + |\ov Q(\ov \calD)|}$ time for any database $\ov\calD$.
\end{lemma}
Assuming Lemma~\ref{lem:eval_univariate_to_maintain_join} holds, Theorem~\ref{thm:fully_dynamic_lower_bound} follows by taking $f(N)=N^{\omega(\wh Q_{\bm\sigma})-1-\gamma}$ for some constant $\gamma > 0$.
(Recall Definition~\ref{defn:fw_lb}.)
By Lemma~\ref{lem:eval_univariate_to_maintain_join}, this would lead into an algorithm
for $\eval[\ov Q]$ with runtime $\tildeO{|\ov \calD|^{\omega(\wh Q_{\bm\sigma})-\gamma} + |\ov Q(\ov \calD)|}$, and by extension, an algorithm for $\eval[\wh Q_{\bm\sigma}]$ with runtime
$\tildeO{|\wh \calD_{\bm\sigma}|^{\omega(\wh Q_{\bm\sigma})-\gamma} + |\wh Q_{\bm\sigma}(\wh \calD_{\bm\sigma})|}$. This would contradict the definition of $\omega(\wh Q_{\bm\sigma})$, thus proving that our choice of $f(N)$ is not possible.
Next, we prove Lemma~\ref{lem:eval_univariate_to_maintain_join}.
\begin{proof}[Proof of Lemma~\ref{lem:eval_univariate_to_maintain_join}]
    Following Algorithm~\ref{alg:backward_reduction}, the reduction works as follows.
    We collect a list $L$ of all pairs $(\alpha, ([\alpha, \beta], \bm t_j))$
    and $(\beta, ([\alpha, \beta], \bm t_j))$ for every tuple $([\alpha, \beta], \bm t_j)$
    that occurs in any relation $\ov R_j$ in $\ov \calD$.
    We go through pairs $(\tau, ([\alpha, \beta], \bm t_j))$ in the list $L$ in order of $\tau$.
    If $\tau = \alpha$, then we insert $\bm t_j$ into $R_j$.
    Otherwise, we delete $\bm t_j$ from $R_j$.
    In total, we have $2|\ov\calD|$ updates and the amortized time per update is $f(2|\ov\calD|)$.
    
    After every insert or delete, we invoke the constant-delay enumeration procedure
    for $\ivmpmd[Q]$ in order to retrieve $\delta_\tau Q(\calD)$, i.e. the output tuples of $Q$
    that are inserted or deleted due to the insert or delete in $R_j$.
    Whenever constant-delay enumeration reports a delete $-Q(\bm t)$ at time $\tau'$
    where the the same tuple $\bm t$ was reportedly inserted $+Q(\bm t)$ at a previous time $\tau$,
    we report the tuple $([\tau, \tau'], \bm t)$ as an output tuple of $\ov Q$.
    The overall time we spend on enumeration is $\tildeO{|\ov Q(\ov\calD)|}$.
\end{proof}

We conclude this section with the following remark about the prospects of
extending Theorem~\ref{thm:fully_dynamic_lower_bound}
from $\ivmpmd[Q]$ to $\ivmpm[Q]$.
\begin{remark}
    \label{rmk:no_full_lower_bound}
    It is not clear how to extend the proof of Theorem~\ref{thm:fully_dynamic_lower_bound},
    which holds for the $\ivmpmd[Q]$ problem, so that it
    also works  for the $\ivmpm[Q]$ problem.
    The main issue lies in the enumeration step.
    In particular, suppose that the IVM oracle that we have in the proof is for $\ivmpm[Q]$.
    Per the above reduction, we invoke constant-delay enumeration after every update
    in $Q$ in order to extract the output of $\ov Q$. However, now the oracle enumerates the {\em full} output of $Q$ after every update.
    As a result, the same output tuple of $Q$ might be reported many times throughout its lifetime.
    This in turn would prevent us from bounding the total enumeration time with $\tildeO{|\ov Q(\ov\calD)|}$, as we did in the original proof.
    Proving an analog of Theorem~\ref{thm:fully_dynamic_lower_bound} for $\ivmpm[Q]$
    remains an open problem.
\end{remark}

\section{Missing Details in Section~\ref{sec:results}}
\label{app:results}
In this section, we first prove the following statements:

\begin{proposition}
\label{prop:q-hierarchical_multivariate}
A query $Q$ is hierarchical if and only if its \multivariate extension $\wh Q$ is acyclic.  
\end{proposition}

\begin{proposition}
\label{prop:nonq-hierarchical_multivariate_lower_bound}
For any non-hierarchical query $Q$, $\fw(\wh Q) \geq \frac{3}{2}$.
\end{proposition}

\begin{proposition}
\label{prop:fhtw_multi_var_LW}
For any Loomis-Whitney query $Q$, $\fw(\wh Q) = \frac{3}{2}$.
\end{proposition}

\begin{proposition}
\label{prop:upper_bound_fhtw_multi}
For any query $Q$, $\fw(Q) \leq  \fw(\wh Q) \leq \fw(Q)+ 1.$
\end{proposition}

\begin{proposition}
\label{thm:fully_dynamic_OMV_lower_bound}
    For any non-hierarchical query $Q$ without self-joins, update stream of length $N$, and $\gamma >0$, there is no algorithm that solves
    $\ivmpm[Q]$ or $\ivmpmd[Q]$ with $\tildeO{N^{1/2-\gamma}}$
    amortized update time, unless the OMv-Conjecture fails.
\end{proposition}

\nop{At the end of this section, we prove Proposition~\ref{cor:hierarchical_fully_dynamic}.}

\subsection{Proof of Proposition~\ref{prop:q-hierarchical_multivariate}}
\begin{citedprop}[\ref{prop:q-hierarchical_multivariate}]
A query is hierarchical if and only if its \multivariate extension is acyclic.
\end{citedprop}

It follows from Proposition~\ref{prop:nonq-hierarchical_multivariate_lower_bound}
that the \multivariate extension $\wh Q$ of any non-hierarchical query $Q$ is 
{\em not} acyclic.
This is because acyclic queries have a fractional hypertree width of $1$.
It remains to prove:

\begin{lemma}
\label{lem:hierarchical_to_acyclic_multivariate}
The \multivariate extension of any hierarchical query is 
acyclic.
\end{lemma}

First, we introduce some simple notions.
Consider a query 
$$Q = R_1(\bm X_1) \wedge \cdots \wedge R_k(\bm X_k).$$ 
A variable $X$ is called {\em root variable}
if it is contained in every atom, i.e., 
$X \in \bigcap_{i \in [k]} \bm X_i$.
It is a {\em leaf variable} if it is contained in exactly one atom, i.e, $X \in \bm X_i$ for some $i \in [k]$ and $X \notin \bm X_j$ for all 
$j \in [k]\setminus\{i\}$. 
We say that an atom $R_i(\bm X_i)$ is {\em subsumed} by another atom $R_j(\bm X_j)$ if $\bm X_i \subseteq \bm X_j$. 
A query $Q'$ results from $Q$ by the {\em elimination of variable
$X$} if it is of the form
$$Q' = R_1(\bm X_1 \setminus \{X\}) \wedge \cdots \wedge 
R_k(\bm X_k \setminus \{X\}).$$
It results by the {\em elimination of an atom} 
$R_i(\bm X_i)$ if it has the form
$$Q' = R_1(\bm X_1) \wedge \cdots \wedge R_{i-1}(\bm X_{i-1}) \wedge R_{i+1}(\bm X_{i+1}) \wedge \cdots \wedge R_k(\bm X_k).$$
Observe that the elimination of a leaf variable and the elimination of a 
subsumed atom are exactly the two steps of the GYO-reduction~\cite{YuO79}. The following claim follows directly from the definition of acyclic queries and  the correctness of the GYO-reduction:

\begin{claim}
\label{claim:GYO_reduction_steps}
Let $Q$ and $Q'$ be join queries such that $Q'$ results from $Q$ by the elimination of a  root variable, a leaf variable, or a subsumed atom. The query $Q$ is acyclic if and only if the query $Q'$ is acyclic. 
\end{claim}

\begin{proof}[Proof of Claim~\ref{claim:GYO_reduction_steps}]
We distinguish on how $Q'$ is obtained from $Q$. 

\paragraph{$Q'$ is obtained by the elimination of a leaf variable or a subsumed atom} 
In this case, the claim follows directly from the correctness of the GYO-reduction~\cite{YuO79}, which says that a query is 
acyclic if and only if we obtain the trivial query consisting of  a single atom with empty schema by repeatedly applying leaf variable and subsumed atom elimination. Indeed, $Q$ is GYO-reducible if and only if $Q'$ is GYO-reducible.

\paragraph{$Q'$ is obtained by root variable elimination} 
Assume that $Q'$ is obtained from $Q$ by removing a root variable $X$.
Consider a join tree $\calT$ for $Q$, and let $\calT'$ be the tree obtained by removing the variable $X$ from every node in $\calT$. 
There is a one-to-one mapping 
between the nodes in $\calT'$ and the atoms in $Q'$. 
Moreover, removing $X$ from all nodes in $\calT$ does not affect the connectivity property in the join tree. Hence, $\calT'$ must be a join tree for $Q'$.

Analogously, any join tree for $Q'$ turns into a join tree
for $Q$ by adding the same variable $X$ to all nodes. 
We conclude that $Q$ is acyclic if and only if $Q'$ is acyclic. 
\end{proof}

We are ready to prove Lemma~\ref{lem:hierarchical_to_acyclic_multivariate}.
Consider a hierarchical query
$Q = R_1(\bm X_1) \wedge \cdots \wedge R_k(\bm X_k)$
and its \multivariate extension 
    $\wh{Q} = \bigvee_{\bm\sigma \in \Sigma_k}\wh Q_{\bm\sigma}$ (Eq \eqref{eq:multivariate}).
We show that each component $\wh Q_{\bm\sigma}$ of $\wh{Q}$
is acyclic. WLOG we only consider the component 
for $\bm \sigma = (1, \ldots , k)$:
$$\wh{Q}_{\bm \sigma} = \wh{R}_1(Z_1 \bm X_1) \wedge \wh{R}_2(Z_1, Z_2, \bm X_2) \wedge \cdots \wedge 
\wh{R}_k(Z_1, \ldots, Z_k, \bm X_k).$$
The other components are symmetric. 

The proof is by induction on the number $k$ of atoms in $\wh Q_{\bm \sigma}$.

\paragraph{Base Case}
Assume that $k=1$, that is,  $\wh Q_{\bm \sigma}$ consists of a single atom. In this case, the query is obviously acyclic.

\paragraph{Induction Step}
Assume now that $\wh Q_{\bm \sigma}$ has $k >1$ atoms. 
The variable $Z_1$ is a root variable. We eliminate it
and obtain 
$$\wh Q_{\bm \sigma}' = \wh{R}_1'(\bm X_1) \wedge \wh{R}_2'(Z_2, \bm X_2) \wedge \cdots \wedge \wh{R}_k'(Z_2, \ldots, Z_k, \bm X_k).$$ 
The atom $\wh{R}_1'(\bm X_1)$ has now the same schema as 
the atom ${R}_1(\bm X_1)$ in the original hierarchical query $Q$.
We remove from $\wh{R}_1'(\bm X_1)$ all variables that are leaves in 
$\wh Q_{\bm \sigma}'$ and obtain the atom $\wh{R}_1''(\bm X_1')$.
It follows from the definition of hierarchical queries that
the schema $\bm X_1'$ must be subsumed by the schema of one of the atoms 
$\wh{R}_2'(Z_2, \bm X_2),$ $\ldots,$ $\wh{R}_k'(Z_2, \ldots, Z_k, \bm X_k)$. We remove the atom $\wh{R}_1''(\bm X_1')$ from $\wh Q_{\bm\sigma}'$ 
and obtain the query 
$$\wh Q_{\bm\sigma}'' = \wh{R}_2'(Z_2, \bm X_2) \wedge \cdots \wedge 
\wh{R}_k'(Z_2, \ldots, Z_k,\bm X_k).$$ 
This query has $k-1$ atoms and contains the variables 
$Z_2, \ldots, Z_k$ but not $Z_1$. 
Observe that $\wh Q_{\bm\sigma}''$ is the \multivariate 
extension of a query $Q'$ that results from 
$Q$ by skipping the atom $R_1(\bm X_1)$.
Hierarchical queries remain hierarchical under atom elimination. Hence, $Q'$ is hierarchical.
By induction hypothesis, $\wh Q_{\bm\sigma}''$ is acyclic. 
\nop{
Observe that 
$Z_{i+1}$ in $Q_{\bm \sigma}''$ takes the role of 
$Z_{i}$ in $Q_{\bm\sigma}$ for $i \in [k]\setminus \{k\}$. 
By induction hypothesis, $Q_{\bm \sigma}''$ is acyclic.
}
Since we obtained $\wh Q_{\bm \sigma}''$ from $\wh Q_{\bm\sigma}$ by eliminating 
root variables, leaf variables, or subsumed atoms, 
it follows from 
Claim~\ref{claim:GYO_reduction_steps} that $\wh Q_{\bm\sigma}$ must be acyclic. This completes the induction step.

\subsection{Proof of Proposition~\ref{prop:nonq-hierarchical_multivariate_lower_bound}}
\begin{citedprop}[\ref{prop:nonq-hierarchical_multivariate_lower_bound}]
    For any non-hierarchical query $Q$, $\fw(\wh Q) \geq \frac{3}{2}$.
\end{citedprop}

\begin{proof}
Consider a non-hierarchical query $Q$. By definition, $Q$
must contain two variables $X$ and $Y$ and three atoms 
$R_1(\bm X_1)$, $R_2(\bm X_2)$, and $R_3(\bm X_3)$ such that
\begin{itemize}
\item $X \in \bm X_1$ but $X \notin \bm X_3$,  
\item $Y \in \bm X_3$ but $Y \notin \bm X_1$, and
\item $X,Y \in \bm X_2$.
\end{itemize}
The \multivariate extension $\wh{Q}$ of $Q$ must have a component 
$\wh Q_{\bm \sigma}$ that contains three atoms 
of the form 
\begin{align*}
    \wh{R}_1(Z_1, \ldots, Z_{k-1},\bm X_1), \wh{R}_2(Z_1, \ldots, Z_{k-2},\bm X_2), \wh{R}_3(Z_1, \ldots, Z_{k}, \bm X_3)
\end{align*}
where $\{Z_1, \ldots, Z_k\}$ are fresh variables.
Moreover, no other atoms in $\wh Q_{\bm\sigma}$ contains any of the the three variables \linebreak
$\{Z_{k-2}, Z_{k-1}, Z_k\}$.
Let $\bm Y = \{X, Y, Z_{k-1}\}$. When restricted to the variables in $\bm Y$,
the resulting query, which we call $\wh Q_{\bm\sigma}'$, is a triangle query.
In particular, let $\wh Q_{\bm\sigma}' \defeq (\wh Q_{\bm\sigma})_{\bm Y}$, where the notation $(\wh Q_{\bm\sigma})_{\bm Y}$ was defined in~\eqref{eq:bag_query}. Then,
\begin{align*}
    \wh Q_{\bm\sigma}' = \wh{R}_1(Z_{k-1}, X), \wh{R}_2(X, Y), \wh{R}_3(Z_{k-1}, Y)
\end{align*}
Therefore, $\rho^*(\wh Q_{\bm\sigma}') = 3/2$.
We use this fact to to show that $\fw(\wh Q_{\bm\sigma})$ must be at least $3/2$.
To that end, we use {\em variable elimination orders} (see Appendix~\ref{app:prelims}).
In particular, we use the definition of fractional hypertree width in terms of variable elimination orders, as given by Eq.~\eqref{eq:fhtw:vo}.
Let $\bm\mu$ be an {\em optimal} variable elimination order for $\wh Q_{\bm\sigma}$, i.e., $\fw(\bm\mu) = \fw(\wh Q_{\bm\sigma})$.
We make the following claim:
\begin{claim}
    There must be a variable $W \in \bm\mu$ where $\bm U_{W}^{\bm\mu} \supseteq \bm Y$.
\end{claim}
We prove the above claim as follows. Initially, the restricted query $(\wh Q_{\bm\sigma})_{\bm Y}$
is a triangle query.
Every time we eliminate a variable $W$ that does {\em not} satisfy $\bm U_{W}^{\bm\mu} \supseteq \bm Y$, the restricted query $(\wh Q_{\bm\sigma})_{\bm Y}$ remains 
a triangle query. Hence, elimination could not have terminated because we haven't yet eliminated
any of the variables in $\bm Y$. Therefore, there must be a variable $W \in \bm\mu$ where $\bm U_{W}^{\bm\mu} \supseteq \bm Y$. Now, we have:
\begin{align*}
    \fw(\wh Q_{\bm \sigma}) =\fw(\bm\mu) \geq \rho^*(\bm U_{W}^{\bm\mu}) \geq \rho^*(\bm Y) =3/2.
\end{align*}
The first equality above follows from our assumption that $\bm\mu$ is an optimal variable elimination order for $\wh Q_{\bm\sigma}$.
The second inequality above follows from the monotonicity of the fractional edge cover number $\rho^*$ (see Appendix~\ref{app:prelims}).
\end{proof}

\subsection{Proof of Proposition~\ref{prop:fhtw_multi_var_LW}}
\begin{citedprop}[\ref{prop:fhtw_multi_var_LW}]
    For any Loomis-Whitney query $Q$, $\fw(\wh Q) = \frac{3}{2}$.
\end{citedprop}
    
Consider an arbitrary Loomis-Whitney query $Q$ of some degree 
$k \geq 3$ and its \multivariate extension $\wh Q$
with fresh variables $Z_1, \ldots , Z_k$.
Proposition~\ref{prop:nonq-hierarchical_multivariate_lower_bound} implies that
the fractional hypertree width of
$\wh Q$ is at least $\frac{3}{2}$.
Let $\wh Q_{\bm \sigma}$ be an arbitrary  component of $\wh Q$.
Next, we show that the fractional hypertree width of $\wh Q_{\bm \sigma}$
is at most
$\frac{3}{2}$, which concludes the proof. 

Let $\wh Q_{\bm \sigma}'$ be the query that we obtain from
$\wh Q_{\bm \sigma}$ by eliminating the variable $Z_k$.
Since this variable occurs in only a single atom in 
$\wh Q_{\bm \sigma}$, the fractional hypertree width of 
$\wh Q_{\bm \sigma}$ and $\wh Q_{\bm \sigma}'$ must be the same. This follows from the equivalence between variable elimination orders and tree decompositions, as explained in Section~\ref{app:prelims}.
It remains to show that $\wh Q_{\bm \sigma}'$ has 
a fractional hypertree width of $\frac{3}{2}$.
By definition of a multivariate extension, $\wh Q_{\bm\sigma}'$ must contain three atoms
of the form:
\begin{align*}
    \wh{R}_1(Z_1, \ldots, Z_{k-2},\bm X_1), \wh{R}_2(Z_1, \ldots, Z_{k-1},\bm X_2), \wh{R}_3(Z_1, \ldots, Z_{k-1}, \bm X_3)
\end{align*}
Note that each of the variables $\{Z_1, \ldots, Z_{k-1}\}$ is contained in at least 
two of the atoms $\wh R_1, \wh R_2, \wh R_3$ above.
Moreover, because $\bm X_1, \bm X_2, \bm X_3$ are the schemas of three atoms in a Loomis-Whitney query $Q$, each variable $X \in \vars(Q)$ must occur in at least two of the three atoms $\wh R_1, \wh R_2, \wh R_3$ above.
Consider now a fractional edge cover $\lambda_{\wh R} = \lambda_{\wh S} = \lambda_{\wh T} = \frac{1}{2}$.
This is indeed a valid fractional edge cover for $\wh Q_{\bm \sigma}'$, and it proves that
$\rho^*(\wh Q_{\bm \sigma}') \leq \frac{3}{2}$.
This implies that $\fw(\wh Q_{\bm \sigma})\leq \frac{3}{2}$.

\subsection{Proof of Proposition~\ref{prop:upper_bound_fhtw_multi}}
\begin{citedprop}[\ref{prop:upper_bound_fhtw_multi}]
    For any query $Q$, $\fw(Q) \leq  \fw(\wh Q) \leq \fw(Q)+ 1$.
\end{citedprop}

We consider  a query $Q = R_1(\bm X_1)\wedge \ldots \wedge R_k(\bm X_k)$ and its \multivariate extension 
$\wh{Q} = \bigvee_{\bm\sigma \in \Sigma_k}
\wh Q_{\bm\sigma}$ 
with fresh variables $Z_1, \ldots , Z_k$ (Eq.~\eqref{eq:multivariate}).
 
\paragraph{Proof of $\fw(Q)\leq  \fw(\wh Q)$}
 For the sake of contradiction, assume that   
 $\fw(Q)>  \fw(\wh Q)$.
 Consider an arbitrary component $\wh Q_{\bm\sigma}$
 of $\wh{Q}$. 
 It follows from the definition of the fractional hypertree width that 
 $\fw(\wh Q_{\bm\sigma}) \leq \fw(\wh Q)$.
 Hence,  
 $\fw(Q)>  \fw(\wh Q_{\bm\sigma})$.
 Let $(\calT_{\bm \sigma}, \chi_{\bm \sigma})$ 
 be a tree decomposition of 
 $\wh Q_{\bm\sigma}$ with 
 $\fw(\calT_{\bm \sigma}, \chi_{\bm \sigma})$ $=$ $\fw(\wh Q_{\bm\sigma})$. 
 We will construct from $(\calT_{\bm \sigma}, \chi_{\bm \sigma})$
 a tree decomposition $(\calT, \chi)$ for $Q$ with 
  $\fw(\calT, \chi) \leq \fw(\calT_{\bm \sigma}, \chi_{\bm \sigma})$.
  This implies that $\fw(\calT, \chi) \leq \fw(\wh Q)$, hence 
  $\fw(Q) \leq \fw(\wh Q)$, which is a contradiction.

Without loss of generality, assume that 
$\wh Q_{\bm\sigma}$ is of the form 
$\wh Q_{\bm\sigma} = \wh{R}_1(Z_1, \bm X_1), \ldots,  
\wh{R}_k(Z_1, \ldots, Z_k, \bm X_k)$.
For all other components of $\wh Q$, the argument 
is completely analogous.
 Let $\wh B_1, \ldots, \wh B_n$ be the bags of 
 $(\calT_{\bm \sigma}, \chi_{\bm \sigma})$.
 The pair $(\calT, \chi)$ results from 
 $(\calT_{\bm \sigma}, \chi_{\bm \sigma})$ by
 eliminating the variables $Z_1, \ldots , Z_k$
 from all bags.
 Let $B_1, \ldots, B_n$ denote the bags of 
 $(\calT, \chi)$.
 The pair $(\calT, \chi)$ constitutes 
 a tree decomposition for $Q$, since 
 (a) for each atom $R_i(\bm X_i)$ of $Q$ there must still be one bag 
 that contains $\bm X_i$ and (b) for each variable 
 $X \in \bigcup_{i \in [k]} \bm X_i$ the bags in $(\calT, \chi)$
 containing $X$ build a connected subtree. 
 Given $i \in [n]$, let 
 $\wh Q_{\bm \sigma}^i$ be the query that we obtain 
 from $\wh Q_{\bm \sigma}$ by restricting the schema 
 of each atom to the variables in $\wh B_i$.
 Similarly, let $Q^i$ be the query obtained 
 from $Q$ by restricting the schema 
 of each atom to the variables in the bag $B_i$.
Consider now a fractional edge cover 
 $\big(\lambda_{\wh{R}_i(\bm X_i)}\big)_{i \in [k]}$
 for $\wh Q_{\bm \sigma}^i$.
This tuple is also a fractional edge cover for 
$Q^i$,
 since $B_i \subseteq \wh B_i$.
 This means that $\rho^*(Q^i) \leq 
 \rho^*(\wh Q_{\bm \sigma}^i)$.
 Hence, $\fw(\calT, \chi)$ $\leq$
$\fw(\calT_{\bm \sigma}, \chi_{\bm \sigma})$.  

\paragraph{Proof of $\wh \fw(Q)\leq \fw(Q)+ 1$}
Let $\wh Q_{\bm\sigma}$ be an arbitrary
component of $\wh Q$. 
Consider a tree decomposition 
$(\calT, \chi)$ of $Q$ with fractional hypertree width $\fw$.
We will transform $(\calT, \chi)$ into a tree decomposition 
$(\calT_{\bm \sigma}, \chi_{\bm \sigma})$ for $\wh Q_{\bm\sigma}$
with fractional hypertree width of at most $\fw(Q)+ 1$. This implies that $\fw(\wh{Q})$ is at most $\fw(Q)+ 1$. 

Without loss of generality, assume that 
$\wh Q_{\bm\sigma}$ is of the form 
$\wh Q_{\bm\sigma} = \wh{R}_1(Z_1, \bm X_1), \ldots,  
\wh{R}_k(Z_1, \ldots, Z_k, \bm X_k)$.
For other components of $\wh Q$, the argument 
is analogous.
 Let $B_1, \ldots, B_n$ be the bags of $(\calT, \chi)$.
The pair $(\calT_{\bm\sigma}, \chi_{\bm\sigma})$ 
results from $(\calT, \chi)$ by adding to each 
bag the variables $Z_1, \ldots , Z_k$. 
 Let $\wh B_1 = B_1 \cup \{Z_1, \ldots, Z_k \}, \ldots, 
 \wh B_n \cup \{Z_1, \ldots, Z_k\}$ be the bags of 
 $(\calT_{\bm\sigma}, \chi_{\bm\sigma})$.
By construction, $(\calT_{\bm\sigma}, \chi_{\bm\sigma})$
constitutes a tree decomposition for $\wh Q_{\bm\sigma}$, since
(a) for each $i \in [k]$, the variables 
$\{Z_1, \ldots, Z_i\} \cup \bm X_i$
must be included in one bag  
and (b) for each variable $X \in \{Z_1, \ldots , Z_k\}
\cup \bigcup_{i \in [k]} \bm X_i$, the bags that include $X$ still build a connected subtree. 
It remains to show 
that $\fw(\calT_{\bm\sigma}, \chi_{\bm\sigma}) \leq \fw(Q)+1$.
 For some $i \in [n]$, let $Q^i$ be the query obtained 
 from $Q$ by restricting each atom to the variables in $B_i$
 and let 
 $\wh Q_{\bm \sigma}^i$ be the query obtained 
 from $\wh Q_{\bm \sigma}$ by restricting 
 each atom to the variables in $\wh B_i$.
Consider a fractional edge cover 
 $\big(\lambda_{R_i(\bm X_i)}\big)_{i \in [k]}$
 for $Q^i$.
We can turn it into a fractional edge cover for 
$\wh Q^i$ by setting $\lambda_{R_k(\bm X_k)}$ to $1$.
 This means that $\rho^*(\wh Q^i_{\bm \sigma}) \leq 
 \rho^*(Q^i) + 1$. This implies 
$\fw(\calT_{\bm \sigma}, \chi_{\bm \sigma})$ $\leq$
$\fw(\calT, \chi) + 1$.  

\subsection{Proof of Proposition~\ref{thm:fully_dynamic_OMV_lower_bound}}
\begin{citedprop}[\ref{thm:fully_dynamic_OMV_lower_bound}]
    For any non-hierarchical query $Q$ without self-joins, update stream of length $N$, and $\gamma >0$, there is no algorithm that solves 
    $\ivmpm[Q]$ or $\ivmpmd[Q]$ with $\tildeO{N^{1/2-\gamma}}$
    amortized update time, unless the OMv-Conjecture fails. 
\end{citedprop}

The proof is based on the Online Vector-Matrix-Vector Multiplication (OuMv)
Conjecture, which is implied by the Online Matrix-Vector Multiplication (OMv)
Conjecture~\cite{HenzingerKNS15}. See Appendix~\ref{app:prelims}.

Consider the simplest non-hierarchical query $Q_{nh} = R(A) \wedge S(A,B) \wedge T(B)$.
The proof of Proposition~\ref{thm:fully_dynamic_OMV_lower_bound} consists of three parts. 
First, we show hardness for the problem  $\ivmpm[Q_{nh}]$ 
(Section~\ref{sec:hardness_ivmpm_RST}). 
Afterwards, we show how the hardness proof for $\ivmpm[Q_{nh}]$ extends to a hardness proof 
for $\ivmpmd[Q_{nh}]$
(Section~\ref{sec:hardness_ivmpm_RST_delta}).
Finally, we explain how these two results generalize to arbitrary non-hierarchical queries without self-joins 
(Section~\ref{sec:hardness_ivmpm_general}). The technique used in the latter part is standard in the literature
(see, e.g.~\cite{BerkholzKS17}), so we only give the high-level idea.

\subsubsection{Hardness of $\ivmpm[Q_{\text{nh}}]$}
\label{sec:hardness_ivmpm_RST}
We reduce the OuMv problem to $\ivmpm[Q_{\text{nh}}]$. 
Assume that there is some $\gamma > 0$ and an algorithm $\calA$ that solves $\ivmpm[Q_{\text{nh}}]$
with $\tildeO{N^{1/2-\gamma}}$ amortized update time for update streams of length $N$. We show how to use $\calA$  
to solve the OuMv problem in sub-cubic time, which is a contradiction to the OuMv Conjecture.
Let the $n \times n$ matrix $\bm M$ and the $n\times 1$ column vectors $(\bm u_1, \bm v_1), \ldots , (\bm u_n, \bm v_n)$
be the input to the OuMv problem. For $i,j \in [n]$, we denote by $\bm M(i,j)$ the Boolean value in $\bm M$
at position $(i,j)$ and by $\bm u(i)$ (and $\bm v(i)$) the value in $\bm u(i)$ (and $\bm v(i)$) at position $i$. The idea is to use relation $S$ to encode the matrix $\bm M$ and the relations $R$ and $T$ to encode 
the vectors $\bm u_r$ and $\bm v_r$ for $r \in [n]$. 
The encoding closely follows prior work \cite{BerkholzKS17}.
Starting from the empty database, we execute at most $n^2$ inserts to $S$ such that 
$(i,j) \in S \Leftrightarrow \bm M(i,j) = 1$. In each round $r \in [n]$, we do the following. 
We use at most $4n$ updates to delete the tuples in $R$ and $T$ and to insert new tuples such that 
$i \in \bm R \Leftrightarrow \bm u_r(i) =1$ and $j \in \bm T \Leftrightarrow \bm v_r(j) =1$.
At the end of round $r$ we request the enumeration of the result of $Q_{\text{nh}}$.
It holds $\bm u_r^{\top}\bm M \bm v_r = 1$ if and only if the result of $Q_{\text{nh}}$ contains at least one tuple. 
So, we can stop the enumeration process after constant time. Summed up over $n$ rounds,   
we execute $n^2 + 4n^2 = \bigO{n^2}$ updates, some of which are followed by enumeration requests that do not require more than constant time. 
Hence, using the update mechanism of algorithm $\calA$, the overall time to solve the OuMv problem takes 
$\bigO{\polylog(n^2) \cdot n^2 \cdot (n^2)^{1/2-\gamma}}$ $=$
$\bigO{\polylog(n^2) \cdot n^{3-2\gamma}}$. As the polylogarithmic factor 
$\polylog(n^2)$ can be asymptotically upper-bounded by $n^{\gamma'}$
for arbitrarily small $\gamma'$, we conclude that
the OuMv problem can be solved in time 
$\bigO{n^{3-\gamma''}}$ for some $\gamma'' >0$.
This contradicts the OuMv Conjecture. 

\subsubsection{Hardness of $\ivmpmd[Q_{\text{nh}}]$}
\label{sec:hardness_ivmpm_RST_delta}
We reduce the OuMv problem to $\ivmpmd[Q_{\text{nh}}]$.
The encoding of the matrix and the vectors is exactly as in 
Section~\ref{sec:hardness_ivmpm_RST}. The only difference is how we obtain the result of 
$\bm u_r^{\top} \bm M \bm v_r$ in each round $r \in [n]$. We first execute the necessary updates to 
encode $\bm u_r$ in $R$. Then, we delete the contents of $T$. Afterwards, we insert, one by one,
each value $j$ with $\bm v_r(j) = 1$ to $T$. After each insert, we request the enumeration of the delta in the 
result of $Q_{\text{nh}}$ and check whether the delta contains 
at lease one  tuple. 
The main observation is that 
$\bm u^{\top}\bm M \bm v = 1$ if and only if one of the  enumeration requests in round $r$ reports at least one tuple. The overall number of updates is as in Section~\ref{sec:hardness_ivmpm_RST}. The additional cost caused 
by each enumeration request is constant. 

\subsubsection{Hardness for Arbitrary Non-Hierarchical Queries}
\label{sec:hardness_ivmpm_general}
Given any non-hierarchical query $Q$ without self-joins, we can reduce the problems $\ivmpm[Q_{\text{nh}}]$ and 
$\ivmpmd[Q_{\text{nh}}]$to
$\ivmpm[Q]$ and respectively $\ivmpmd[Q]$. We give the high-level idea of the reduction. Since $Q$ is non-hierarchical, 
it must contain two variables $A'$ and $B'$ and three atoms $R'(\bm X_R)$, $S'(\bm X_S)$, and 
$T'(\bm X_T)$ such that $A' \in \bm X_R$, $B' \notin \bm X_R$, $A', B' \in \bm X_S$, $B' \in \bm X_T$, 
and $A' \notin \bm X_T$. We use $R'(\bm X_R)$, $S'(\bm X_S)$, and $T'(\bm X_T)$ to encode
$R(A)$, $S(A,B)$, and respectively $T(B)$. While the variables $A'$ and $B'$ take on the $A$- and respectively 
$B$-values from the database for $Q_{\text{nh}}$, all other variables in $Q$ take on a fixed dummy value. Hence, we can easily translate updates
and enumeration requests to $Q_{\text{nh}}$ to updates and enumeration requests to $Q$.

\end{document}